\documentclass[a4paper,10pt,notitlepage]{article}
\usepackage{amsmath}
\usepackage{amssymb}
\usepackage{xstring}
\usepackage[english]{babel}
\usepackage{hyperref}
\usepackage[amsthm,hyperref,thmmarks]{ntheorem}
\usepackage{caption}
\usepackage{subcaption}
\usepackage{algorithmic}
\usepackage{algorithm}
\usepackage{xcolor} 
\usepackage{tikz} 
\usepackage{pgfplots} 
\pgfplotsset{compat=1.14}
\usepackage{graphicx}
\usepackage[a4paper,left=20mm,right=20mm,top=20mm,bottom=20mm]{geometry}
\newcommand{\refP}[1]{%
	\def\InputString{#1}%
	\IfBeginWith{\InputString}{Equation}{%
		(\ref{#1})}{%
	\IfBeginWith{\InputString}{Section}{%
		Section \ref{#1}}{%
	\IfBeginWith{\InputString}{Subsection}{%
		Subsection \ref{#1}}{%
	\IfBeginWith{\InputString}{Chapter}{%
		Chapter \ref{#1}}{%
	\IfBeginWith{\InputString}{Subsubsection}{%
		Subsubsection \ref{#1}}{%
	\IfBeginWith{\InputString}{Problem}{%
		(\ref{#1})}{%
	\IfBeginWith{\InputString}{Property}{%
		Property (\ref{#1})}{%
	\IfBeginWith{\InputString}{Algorithm}{%
		Algorithm \ref{#1}}{%
	\IfBeginWith{\InputString}{Figure}{%
		Figure (\ref{#1})}{%
	\IfBeginWith{\InputString}{Question}{%
		Question (\ref{#1})}{%
	\IfBeginWith{\InputString}{Footnote}{%
		Footnote \ref{#1}}{%
		\ref{#1}}}}}}}}}}}}%
}

\newcommand{\thref}[1]{%
	\def\InputString{#1}%
	\IfBeginWith{\InputString}{Definition}{%
		Definition \ref{#1}}{%
	\IfBeginWith{\InputString}{Theorem}{%
		Theorem \ref{#1}}{%
	\IfBeginWith{\InputString}{Proposition}{%
		Proposition \ref{#1}}{%
	\IfBeginWith{\InputString}{Lemma}{%
		Lemma \ref{#1}}{%
	\IfBeginWith{\InputString}{Corollary}{%
		Corollary \ref{#1}}{%
	\IfBeginWith{\InputString}{Remark}{%
		Remark \ref{#1}}{%
	\IfBeginWith{\InputString}{Citedresult}{%
		Cited Result \ref{#1}}{%
	\IfBeginWith{\InputString}{Example}{%
		Example \ref{#1}}{%
	\IfBeginWith{\InputString}{Conjecture}{%
		Conjecture \ref{#1}}{%
	\IfBeginWith{\InputString}{Question}{%
		Question \ref{#1}}{%
	\IfBeginWith{\InputString}{Experiment}{%
		Experiment \ref{#1}}{%
	\IfBeginWith{\InputString}{Goal}{%
		Goal \ref{#1}}{%
		\ref{#1}}}}}}}}}}}}}%
}
\definecolor{TodoRed}{RGB}{150,50,0}

\newtheorem{CounterTheorem}{}[section]

\newtheorem{Definition}[CounterTheorem]{Definition}
\newtheorem{Theorem}[CounterTheorem]{Theorem}
\newtheorem{Proposition}[CounterTheorem]{Proposition}
\newtheorem{Lemma}[CounterTheorem]{Lemma}

\newtheorem{Remark}[CounterTheorem]{Remark}

\newtheorem{AlgorithmEnvirorementForNTheorem}[CounterTheorem]{Algorithm}

\renewcommand{\Vec}[1]{\mathbf{#1}}

\newcommand{\aVec}{\Vec{a}}

\newcommand{\eVec}{\Vec{e}}

\newcommand{\hVec}{\Vec{h}}

\newcommand{\vVec}{\Vec{v}}

\newcommand{\xVec}{\Vec{x}}
\newcommand{\yVec}{\Vec{y}}
\newcommand{\zVec}{\Vec{z}}

\newcommand{\muVec}{\Vec{\mu}}

\newcommand{\Mat}[1]{\mathbf{#1}}
\newcommand{\MatGreek}[1]{\boldsymbol{#1}}
\newcommand{\AMat}{\Mat{A}}
\newcommand{\BMat}{\Mat{B}}

\newcommand{\EMat}{\Mat{E}}

\newcommand{\HMat}{\Mat{H}}

\newcommand{\NMat}{\Mat{N}}

\newcommand{\WMat}{\Mat{W}}
\newcommand{\XMat}{\Mat{X}}
\newcommand{\YMat}{\Mat{Y}}
\newcommand{\ZMat}{\Mat{Z}}

\newcommand{\SigmaMat}{\MatGreek{\Sigma}}

\newcommand{\IDMat}{\Mat{I}}

\newcommand{\TextForAll}{\hspace{2pt} \text{ for all } \hspace{2pt}}

\newcommand{\TextIf}{\hspace{2pt}\text{ if }\hspace{2pt}}

\newcommand{\TextAnd}{\hspace{2pt}\text{ and }\hspace{2pt}}

\newcommand{\ZeroSet}{\left\{0\right\}}

\newcommand{\abs}[1]{\left|#1\right|}

\newcommand{\norm}[1]{\left\|#1\right\|}
\newcommand{\scprod}[2]{\langle#1,#2\rangle}

\newcommand{\SetOf}[1]{\left[#1\right]}

\newcommand{\argmin}[1]{\underset{#1}{\textnormal{argmin}}\hspace{1pt}}

\newcommand{\Prob}[1]{\mathbb{P}\left[#1\right]}
\newcommand{\Expect}[1]{\mathbb{E}\left[#1\right]}

\newcommand{\GaussianRV}[2]{\mathcal{N}\left(#1,#2\right)}
\newcommand{\CGaussianRV}[2]{\mathcal{CN}\left(#1,#2\right)}

\def\AddBasicFunction#1#2{
	\expandafter\def\csname #1\endcsname##1{
		\def\InputString{##1}
		\def\CheckString{}
		\ifx\InputString\CheckString 
			#2
		\else
			#2 \left(##1\right)
		\fi
	}
}
\AddBasicFunction{Exp}{\exp}
\AddBasicFunction{Cos}{\cos}
\AddBasicFunction{Sin}{\sin}
\AddBasicFunction{Tan}{\tan}
\AddBasicFunction{Ln}{\ln}
\AddBasicFunction{Log}{\log}
\AddBasicFunction{ArcCos}{\arccos}
\AddBasicFunction{ArcSin}{\arcsin}
\AddBasicFunction{ArcTan}{\arctan}
\AddBasicFunction{Cosh}{\cosh}
\AddBasicFunction{Sinh}{\sinh}
\AddBasicFunction{Tanh}{\tanh}
\AddBasicFunction{ArcCosh}{\textnormal{arccosh}}
\AddBasicFunction{ArcSinh}{\textnormal{arcsinh}}
\AddBasicFunction{ArcTanh}{\textnormal{arctanh}}

\AddBasicFunction{Kernel}{\ker}
\AddBasicFunction{Rank}{\textnormal{rank}}
\AddBasicFunction{Range}{\textnormal{Ran}}

\AddBasicFunction{supp}{\textnormal{supp}}

\AddBasicFunction{sgn}{\textnormal{sgn}}

\newcommand{\ExpE}{\mathrm{e}}

\newcommand{\SetSize}[1]{\#\left(#1\right)}

\title{Robustness of Covariance Estimators with Application in Activity Detection}
\date{}
\author{
	Hendrik Bernd Zarucha
	\footnote{
		Communications and Information Theory Group,
		Technische Universtität Berlin, Berlin,
		\href{mailto:petersen@tu-berlin.de}{petersen@tu-berlin.de}
	}
	\and
	Peter Jung
	\footnote{
		Communications and Information Theory Group,
		Technische Universtität Berlin, Berlin,
		and German Aerospace Center (DLR)
		\href{mailto:peter.jung@tu-berlin.de}{peter.jung@tu-berlin.de}
	}
	\and
	Giuseppe Caire
	\footnote{
		Communications and Information Theory Group,
		Technische Universtität Berlin, Berlin,
		\href{mailto:caire@tu-berlin.de}{caire@tu-berlin.de}
	}
}
\newcommand{\Trace}[1]{\textnormal{trace}\left(#1\right)}
\newcommand{\Det}[1]{\textnormal{det}\left(#1\right)}
\newcommand{\Diag}[1]{\textnormal{diag}\left(#1\right)}
\AddBasicFunction{f}{f}
\AddBasicFunction{g}{g}
\AddBasicFunction{h}{h}
\AddBasicFunction{obj}{obj}
\newcommand{\Real}[1]{\textnormal{real}\left(#1\right)}
\newcommand{\Imag}[1]{\textnormal{imag}\left(#1\right)}
\newcommand{\CC}[1]{\overline{#1}}

\definecolor{AddGreen}{RGB}{0,255,0}

\definecolor{RemoveRed}{RGB}{255,0,0}

\definecolor{EditYellow}{RGB}{255,255,0}

\begin{document}
	\maketitle
	\begin{abstract}
The first part of this work
considers a general class of covariance estimators.
Each estimator of that class is generated by a real-valued function $\g{}$
and a set of model covariance matrices $\mathcal{H}$.
If $\WMat$ is a potentially perturbed observation of a searched covariance matrix,
then the estimator is the minimizer of the sum of $g$ applied to each eigenvalue
of $\WMat^\frac{1}{2}\ZMat^{-1}\WMat^\frac{1}{2}$
under the constraint that $\ZMat$ is from $\mathcal{H}$.
It is shown that under mild conditions on $\g{}$ and $\mathcal{H}$
such estimators are robust, meaning the estimation error
can be made arbitrarily small if the perturbation of $\WMat$
gets small enough.
\par
In the second part of this work the previous results are applied to
activity detection in random access with multiple receive antennas.
In activity detection recovering the large scale fading coefficients
is a sparse recovery problem
which can be reduced to a structured covariance estimation problem.
The recovery can be done with a non-negative least
squares estimator or with a relaxed maximum likelihood estimator.
It is shown that under suitable assumptions on the distributions of the noise
and the channel coefficients, the relaxed maximum likelihood estimator
is from the general class of covariance estimators considered in the first
part of this work.
Then, codebooks based upon a signed kernel condition are proposed.
It is shown that with the proposed codebooks both estimators
can recover the large-scale fading coefficients
if the number of receive antennas
is high enough and $S\leq\left\lceil\frac{1}{2}M^2\right\rceil-1$
where $S$ is the number of active users and $M$ is number of pilot symbols per user.
\end{abstract}
	\section{Introduction}\label{Section:Introduction}
\noindent
This work considers a compressed sensing problem with
$K$ measurement processes of the form
\begin{align}
	\label{Equation:measurement}
		\yVec_k
	=
		\AMat\sqrt{\Diag{\xVec}}\hVec_k+\eVec_k
	\in
		\mathbb{C}^{M}
\end{align}
for $k=1,\dots,K$. Here $\xVec\in\mathbb{R}^N$ is assumed to be $S$-sparse
and non-negative,
$\hVec_k$ are mutually independent, complex normal distributed random vectors with expectation $0$ and identity $\IDMat$ as
covariance matrix, $\eVec_k$ are mutually independent, complex normal distributed random vectors
with expectation $0$ and covariance matrix $\SigmaMat$
and
$\AMat\in\mathbb{C}^{M\times N}$ has columns $\aVec_n$.
\par
Activity detection in random access with multiple receive antennas
can be modeled as such a problem.
In this case $\aVec_n$ are the pilot symbols assigned to the $n$-th user,
$\AMat$ is the codebook,
and the channel coefficient
of the $n$-th user transmitting to the $k$-th receive antenna
is given by $(\sqrt{\Diag{\xVec}}\hVec_k)_n$.
The variances of the channel coefficients are called large-scale fading coefficients
and are the entries of $\xVec$ so that the entries of $\xVec$ are non-negative.
Further, $N$ is the total number of devices, $S$ is the number of active devices,
$K$ is the number of receive antennas,
and the $k$-th receive antenna observes the signal $\yVec_k$.
Since devices only transmit sporadically and the large scale fading coefficient
vanishes if a device is inactive, $\xVec$ is indeed a sparse vector.
\par
Activity detection can be modeled in other ways, such as \cite{other_model_1,other_model_2}, and the
validity of this model has been discussed frequently.
A further discussion about the validity of the model
is thus omitted. It suffices to say that it is justified to study it,
since this work focuses on mathematical properties instead of performance,
and since this model has been studied several times,
for instance in \cite{model_1,model_2,model_3,model_4,model_5,model_6}.
Note that the measurement process \eqref{Equation:measurement} also appears in works about
unsourced random access, such as \cite{ura_mimo,fengler}. Unsourced random access differs in two aspects
from activity detection. In unsourced random access the columns $\aVec_n$ represent the
codewords of a common codebook used by all users, and $x_n$ will then be the $\ell_2$-norm
of the vector of all large-scale fading coefficients of users transmitting the $n$-th codeword so that this work's results can also be applied to
unsourced random access.
\par
If the columns of $\YMat\in\mathbb{C}^{M\times K}$ are $\yVec_k$,
the columns of $\HMat$ are $\hVec_k$, the columns of $\EMat$ are $\eVec_k$,
and $K$ is the number of receive antennas, then
\eqref{Equation:measurement} yields
\begin{align}
	\nonumber
		\frac{1}{K}\YMat\YMat^H
	=&
		\AMat\Diag{\xVec}\AMat^H+\SigmaMat
		+\AMat\sqrt{\Diag{\xVec}}\left(
			\frac{1}{K}\HMat\HMat^H-\IDMat
		\right)\sqrt{\Diag{\xVec}}\AMat^H
		+\frac{1}{K}\EMat\EMat^H-\SigmaMat
	\\&
	\label{Equation:measurement_covariance}
		+\AMat\sqrt{\Diag{\xVec}}\frac{1}{K}\HMat\EMat^H
		+\frac{1}{K}\EMat\HMat^H\sqrt{\Diag{\xVec}}\AMat^H.
\end{align}
Due to \eqref{Equation:measurement_covariance} the sample covariance matrix
$\frac{1}{K}\YMat\YMat^H$ is a perturbed observation of
the covariance matrix
\begin{align}
	\nonumber
		\Expect{\frac{1}{K}\YMat\YMat^H}
	=&
		\AMat\Diag{\xVec}\AMat^H+\SigmaMat
	=
		\sum_{n=1}^N\aVec_n\aVec_n^Hx_n+\SigmaMat
\end{align}
with the mean-zero perturbation
\begin{align}
	\nonumber
		\AMat\sqrt{\Diag{\xVec}}\left(
			\frac{1}{K}\HMat\HMat^H-\IDMat
		\right)\sqrt{\Diag{\xVec}}\AMat^H
		+\frac{1}{K}\EMat\EMat^H-\SigmaMat
		+\AMat\sqrt{\Diag{\xVec}}\frac{1}{K}\HMat\EMat^H
		+\frac{1}{K}\EMat\HMat^H\sqrt{\Diag{\xVec}}\AMat^H.
\end{align}
Thus, \eqref{Equation:measurement_covariance} describes a covariance estimation
problem where the searched covariance matrix $\sum_{n=1}^N\aVec_n\aVec_n^Hx_n+\SigmaMat$
is from the structured model
$\left\{\sum_{n=1}^N\aVec_n\aVec_n^Hz_n+\SigmaMat:\zVec\geq 0\right\}$.
Several estimators can be considered to estimate $\xVec$.
One can consider the non-negative least squares estimator which
is given as any solution of
\begin{align}
	\label{Equation:NNLR}
		\min_{\zVec\geq 0}
		\norm{\sum_{n=1}^N\aVec_n\aVec_n^Hz_n+\SigmaMat-\frac{1}{K}\YMat\YMat^H}_2
\end{align}
where the norm is the Frobenius norm.
The non-negative least squares can be considered a relaxed version
of more common $\ell_1$-regularized estimators in compressed sensing \cite{NNLS}.
It can recover sparse non-negative signals even without
$\ell_0$-regularization or $\ell_1$-regularization \cite{NNLS}.
On the other hand, one can consider the maximum likelihood estimator which
can be found by solving
\begin{align}
	\nonumber
			\min_{\zVec\geq 0 \textnormal{ is $S$-sparse}}
			\Trace{\left(\sum_{n=1}^N\aVec_n\aVec_n^Hz_n+\SigmaMat\right)^{-1}\frac{1}{K}\YMat\YMat^H}
			+\Ln{\Det{\sum_{n=1}^N\aVec_n\aVec_n^Hz_n+\SigmaMat}},
\end{align}
see for instance \cite{fengler}.
Due to the combinatorial nature of the constraints, one
often considers the relaxed maximum likelihood estimator
which is any minimizer of
\begin{align}
	\label{Equation:trace_log_det}
		\min_{\zVec\geq 0}
		\Trace{\left(\sum_{n=1}^N\aVec_n\aVec_n^Hz_n+\SigmaMat\right)^{-1}\frac{1}{K}\YMat\YMat^H}
		+\Ln{\Det{\sum_{n=1}^N\aVec_n\aVec_n^Hz_n+\SigmaMat}}.
\end{align}
Since the relaxation removes the combinatorial constraints,
finding the minimizer is significantly easier, and one can use,
for example, coordinate-wise descent methods to approximate a minimizer.
Due to the $\ell_0$-regularizer being obsolete for the non-negative least squares,
one could hope that a similar result also holds
for the relaxed maximum likelihood estimator.
In particular, robust recovery guarantees are sought after.
Very generally speaking,
robustness in this work will refer to a property that bounds the estimation
error of an estimator as a function of the magnitude of the perturbation;
however, the exact statements will be specified in theorems below.
\subsection{Prior Work}
The question arises under what conditions the estimators \eqref{Equation:NNLR}
and \eqref{Equation:trace_log_det} can recover the unknown
$\xVec$ and thus be used to estimate the active users. This has been investigated in \cite{fengler,chen} for the relaxed maximum likelihood
estimator and in \cite{fengler} for the non-negative least squares.
In \cite[Theorem~2]{fengler} it was shown that for a certain randomly drawn codebook $\AMat$ with
\begin{align}
	\label{Equation:M^2=SLn}
		M^2\asymp S\left(\Ln{\ExpE\frac{N}{S}}\right)^2
\end{align}
the linear operator $\mathcal{A}\left(\zVec\right)=\sum_{n=1}^N\aVec_n\aVec_n^Hz_n$
satisfies a restricted isometry property with a high probability
and in \cite[Theorem~2]{sparse_psd} a similar result was shown for more general
$\AMat$ for a similar number of measurements.
If the restricted isometry property is fulfilled, then
a robust recovery guarantee is given for \eqref{Equation:NNLR} that enables the recovery
of large large scale fading and determination of active users according to
\cite[Theorem~3]{fengler}.
In \cite[Theorem~2]{chen} a unique identifiability condition was established 
under which the minimizers of \eqref{Equation:trace_log_det}
converge in probability to the vector of large scale fading coefficients $\xVec$ as
$K\rightarrow\infty$.
In \cite[Theorem~5]{chen} a condition equivalent to the unique identifiability condition was established.
In \cite[Theorem~9]{chen} it was established that if a restricted
isometry property is fulfilled,  
the unique identifiability condition from \cite[Theorem~2~and~Theorem~5]{chen} is also fulfilled.
The unique identifiability condition was also investigated
in \cite{identifiability}.
The result \cite[Theorem~1]{fengler} considered the minimizer of a discretized version of
\eqref{Equation:trace_log_det} under additional knowledge about the large scale fading coefficients.
It was shown that for a certain randomly drawn codebook the discretized estimator of \eqref{Equation:trace_log_det} is
estimating the active users correctly with a high probability
if \eqref{Equation:M^2=SLn} is fulfilled.
\par
The work \cite{minimizers_geodesic_convexity} considers
a generalized version of \eqref{Equation:trace_log_det}
where the trace operator 
in \eqref{Equation:trace_log_det}
is replaced by a general geodesic function
and discusses a fixed point method to solve the generalized version
of \eqref{Equation:trace_log_det}.
However, the work \cite{minimizers_geodesic_convexity} does not answer
when the optimization problem can accurately recover the unknown vector
or when the fixed point method converges to a minimizer.
\subsection{This Work's Contribution}
In the first part of this work a general class of covariance estimators is considered.
Given some closed set of positive definite matrices
$\mathcal{H}$ of structured covariance matrices and $\g{}:\left(0,\infty\right)\rightarrow\mathbb{R}$
the covariance estimators considered in this work are minimizers of
\begin{align}
	\label{Equation:covariance_estimator}
		\min_{\ZMat\in\mathcal{H}}
		\sum_{m=1}^M\g{\lambda_m\left(\WMat^\frac{1}{2}\ZMat^{-1}\WMat^\frac{1}{2}\right)},
\end{align}
where $\lambda_m\left(\ZMat\right)$ is the $m$-th largest eigenvalue of $\ZMat$
and $\WMat$ is a Hermitian positive definite perturbed observation of a searched and unknown covariance matrix
$\XMat$. For instance, $\WMat$ could be a sample covariance matrix.
Note that, since $\WMat$ is not necessarily in $\mathcal{H}$,
the minimization problem is non-trivial.
It is shown that if certain conditions on $g$ are fulfilled, then this estimator is robust,
meaning that its minimizers are arbitrarily close to the searched and unknown covariance matrix $\XMat$ as long as the perturbed $\WMat$ is close enough to $\XMat$.
\par
The second part of this work considers deterministic codebook constructions
from \cite{NNLR} with
\begin{align}
	\label{Equation:M^2=S}
		M^2\asymp S
\end{align}
such that the linear operator $\mathcal{A}\left(\zVec\right)=\sum_{n=1}^N\aVec_n\aVec_n^Hz_n$ satisfies a signed kernel condition \cite{NNLR} instead of
a restricted isometry property. It is shown that for such constructions the minimizers of
\eqref{Equation:NNLR} and \eqref{Equation:trace_log_det}
each converge in probability to the unknown $\xVec$.
It is discussed that this improvement comes with a trade-off. The number
of receive antennas has to increase significantly if
one improves from \eqref{Equation:M^2=SLn} to \eqref{Equation:M^2=S}.
\par
Further, it is proven that, if $\mathcal{A}\left(\zVec\right)=\sum_{n=1}^N\aVec_n\aVec_n^Hz_n$ satisfies the signed kernel condition,
the optimization problem \eqref{Equation:trace_log_det}
is robust, meaning that the estimation error of the minimizer \eqref{Equation:trace_log_det}
can be controlled by making
\begin{align}
	\nonumber
		\sum_{n=1}^N\aVec_n\aVec_n^Hx_n+\SigmaMat-\frac{1}{K}\YMat\YMat
\end{align}
small enough.
This gives a direct relation between the number of receive antennas and the probability to
make the estimation error smaller than a given target.
This result is proven by applying the results of the first part of this work with
$\WMat=\frac{1}{K}\YMat\YMat^H$, $\XMat=\sum_{n=1}^N\aVec_n\aVec_n^Hx_n+\SigmaMat$,
$\mathcal{H}=\left\{\sum_{n=1}^N\aVec_n\aVec_n^Hz_n+\SigmaMat\TextForAll \zVec\geq 0\right\}$
and $\g{x}=x-\Ln{x}$.
\par
At last, it is shown that, if $\mathcal{A}\left(\zVec\right)=\sum_{n=1}^N\aVec_n\aVec_n^Hz_n$ satisfies the signed kernel condition,
cluster points of a common coordinate descent method to solve
\eqref{Equation:trace_log_det} are indeed stationary points of \eqref{Equation:trace_log_det}.
\subsection{Notation}\label{Section:notation}
Given $N\in\mathbb{N}$ set $\SetOf{N}:=\left\{1,\dots,N\right\}$.
The set of Hermitian matrices is denoted by
$\mathbb{H}^M:=\big\{\AMat\in\mathbb{C}^{M\times M}:\AMat=\AMat^H\big\}$
and the set of Hermitian positive definite matrices by
$\mathbb{HPD}^M:=\left\{\AMat\in\mathbb{H}^M:\textnormal{$\AMat$ is positive definite}\right\}$.
The $m$-th largest eigenvalue of $\AMat\in\mathbb{H}^M$ is denoted by
$\lambda_m\left(\AMat\right)$.
For any $\AMat\in\mathbb{C}^{M\times N}$ the $\ell_p$ norm of its entries is denoted by
$\norm{\AMat}_p$ so that $\norm{\AMat}_2$ is the frobenious norm.
For any $\AMat\in\mathbb{C}^{M\times N}$ the operator norm as an operator from $\ell_p$ to $\ell_q$
is denoted by $\norm{\AMat}_{p\rightarrow q}:=\sup_{\norm{\xVec}_p\leq 1}\norm{\AMat\xVec}_q$.
The space $\mathbb{H}^M$ is embedded with the topology induced by the frobenius norm $\norm{\cdot}_2$.
Note that for any $\AMat\in\mathbb{H}^M$ one has
$\norm{\AMat}_2^2=\sum_{m=1}^M\lambda_m\left(\AMat\right)^2$,
which will be used frequently.
$\mathbb{HPD}^M\subset\mathbb{H}^M$ is equipped with the subspace topology.
Note that sets are compact in $\mathbb{HPD}^M$ if and only if they are compact in
$\mathbb{H}^M$. However, sets that are closed in $\mathbb{HPD}^M$ are not necessarily closed in
$\mathbb{H}^M$, since $\mathbb{HPD}^M$ is not closed in $\mathbb{H}^M$,
and thus not a complete metric space.
The same is true for any set $\mathcal{H}\subset\mathbb{HPD}^M$
that is closed in $\mathbb{HPD}^M$ which is always embedded with the subspace topology of
$\mathbb{HPD}^M$.
The set of $S$-sparse vectors is denoted by $\Sigma_S^N:=\left\{\xVec\in\mathbb{R}^N:\textnormal{$\xVec$ has at most $S$ non-zero coordinates}\right\}$ and the set of non-negative vectors is denoted by $\mathbb{R}_+^N:=\left\{\xVec\in\mathbb{R}^N:x_n\geq 0\right\}$.
By $\xVec\sim\GaussianRV{\muVec}{\SigmaMat}$ it is denoted that $\xVec$ is a normal distributed random vector with expectation $\muVec\in\mathbb{R}^{M}$ and covariance $\SigmaMat\in\mathbb{R}^{M\times M}$.
By $\xVec\sim\CGaussianRV{\muVec}{\SigmaMat}$ it is denoted that $\xVec$ is a complex normal distributed random variable with expectation $\muVec\in\mathbb{C}^{M}$ and covariance $\SigmaMat\in\mathbb{C}^{M\times M}$.
	\section{Main Results}\label{Section:main_results}
\subsection{Robustness of Covariance Estimation}
Given some $\mathcal{H}\subset\mathbb{HPD}^M$ and
$\g{}:\left(0,\infty\right)\rightarrow\mathbb{R}$
the covariance estimators considered in this work are minimizers of
\begin{align}
	\nonumber
		\min_{\ZMat\in\mathcal{H}}
		\sum_{m=1}^M\g{\lambda_m\left(\WMat^\frac{1}{2}\ZMat^{-1}\WMat^\frac{1}{2}\right)}.
\end{align}
Here $\mathcal{H}$ is any set of potential covariance matrices fitting a structured model
and $\WMat$ is a perturbed observation of a searched and unknown
covariance matrix $\XMat$. For instance, $\WMat$
could be a sample covariance matrix.
Certainly some restrictions on $\mathcal{H}$ and $\g{}$ are required, because otherwise there might not
even be a solution to the optimization problem. In this work $\mathcal{H}$ will be a closed set in
$\mathbb{HPD}^M$ and one of two different conditions on $g$ is considered.
The following requirements on $g$ are the minimal requirements for the proof of the main result.
\begin{Definition}\label{Definition:tuple_nice}
	Let $\g{}:\left(0,\infty\right)\rightarrow\mathbb{R}$,
	$g_1:\left[\g{1},\infty\right)\rightarrow\left(0,1\right]$,
	$g_2:\left[\g{1},\infty\right)\rightarrow\left[1,\infty\right)$
	and
	$\delta_1,\delta_2:\left(0,\infty\right)\rightarrow\left(0,\infty\right)$.
	$g$ and the tuple $\left(\g{},g_1,g_2,\delta_1,\delta_2\right)$ are each called
	sufficiently nice if the following properties are fulfilled.
	\begin{enumerate}
		\item
			\label{Property:Definition:tuple_nice:grow}
			$\g{}$ is sufficiently growing, namely
			$\lim_{x\rightarrow 0}\g{x}=\infty=\lim_{x\rightarrow\infty}\g{x}$.
		\item
			\label{Property:Definition:tuple_nice:cont}
			$\g{}$ is continuous everywhere.
		\item
			\label{Property:Definition:tuple_nice:cont_around_1}
			$\g{}$ is continuous around $1$ with
			$\abs{x-1}\leq\delta_1\left(\epsilon\right)\Rightarrow\abs{\g{x}-\g{1}}\leq\epsilon$
			for all $\epsilon>0$.
		\item
			\label{Property:Definition:tuple_nice:minimizer}
			The minimizer is sufficiently explicit, namely
			$\g{x}-\g{1}\leq\delta_2\left(\epsilon\right)
				\Rightarrow\abs{x-1}\leq\epsilon$
			for all $\epsilon>0$.
		\item
			\label{Property:Definition:tuple_nice:inver}
			The almost inverse functions exist and are defined by
			$g_1\left(y\right)=\inf_{z\in\left(0,1\right]:\g{z}\leq y}z$
			and \\
			$g_2\left(y\right)=\sup_{z\in\left[1,\infty\right):\g{z}\leq y}z$
			for all $y\in\left[\g{1},\infty\right)$.
	\end{enumerate}
\end{Definition}
Note that \refP{Property:Definition:tuple_nice:minimizer}
implies that $\g{x}\geq\g{1}$ for all $x\in\left(0,1\right)$
and that $1$ is the unique global minimizer of $\g{}$.
This further guarantees that the minimizer for
$\mathcal{H}=\mathbb{HPD}^M$ is always $\WMat$ since the identity is the only matrix
with only $1$ as eigenvalue. In this case the robustness is trivial.
Whenever $\mathcal{H}\neq\mathbb{HPD}^M$, solving the optimization problem
and proving robustness is not
trivial however. Due to \refP{Property:Definition:tuple_nice:grow} the almost inverse functions
$g_1,g_2$ satisfying \refP{Property:Definition:tuple_nice:inver} always exist and are well defined.
Further, \refP{Property:Definition:tuple_nice:grow} and \refP{Property:Definition:tuple_nice:cont}
yield that they are strictly monotonic.
The following stricter condition on $g$ can be used to improve the robustness.
\begin{Definition}\label{Definition:tuple_convex}
	Let $\g{}:\left(0,\infty\right)\rightarrow\mathbb{R}$,
	$g_1:\left[\g{1},\infty\right)\rightarrow\left(0,1\right]$,
	$g_2:\left[\g{1},\infty\right)\rightarrow\left[1,\infty\right)$,
	$\nu>0$ and $\epsilon_0\in\left(0,1\right)$.
	$g$ and the tuple $\left(\g{},g_1,g_2,\nu,\epsilon_0\right)$ are each called
	sufficiently convex if the following properties are fulfilled.
	\begin{enumerate}
		\item
			\label{Property:Definition:tuple_convex:grow}
			$\g{}$ is sufficiently growing, namely
			$\lim_{x\rightarrow 0}\g{x}=\infty=\lim_{x\rightarrow\infty}\g{x}$.
		\item
			\label{Property:Definition:tuple_convex:cont}
			$\g{}$ is continuous everywhere.
		\item
			\label{Property:Definition:tuple_convex:monotonic_g1}
			$\g{}$ is strictly monotonically falling in $\left(0,1\right]$ with inverse function
			$g_1$.
		\item
			\label{Property:Definition:tuple_convex:monotonic_g2}
			$\g{}$ is strictly monotonically increasing in $\left[1,\infty\right)$ with inverse function
			$g_2$.
		\item
			\label{Property:Definition:tuple_convex:inver}
			$\g{1+\epsilon}\leq\g{1-\epsilon}$ for all $\epsilon\in\left(0,1\right)$.
		\item
			\label{Property:Definition:tuple_convex:convex}
			$\g{}$ is convex on $\left[1,\infty\right)$.
		\item
			\label{Property:Definition:tuple_convex:composition_derivative}
			$\g{}$ is differentiable everywhere with
			$g'\left(x\right)\neq 0$ for all $x\neq 1$ and
			$-\frac{g'\left(1+\epsilon\right)}{g'\left(1-\epsilon\right)}\geq \nu$
			for all $\epsilon\in\left(0,\epsilon_0\right]$.
	\end{enumerate}
\end{Definition}
It is later proven that $g$ being sufficiently convex is indeed a strictly stronger condition
than $g$ being sufficiently nice.
Under any of these conditions, the corresponding covariance estimators are robust.
\begin{Theorem}\label{Theorem:covariance_robust}
	Let the tuple $g$ be sufficiently nice or sufficiently convex,
	$\mathcal{H}\subset\mathbb{HPD}^M$ be closed in $\mathbb{HPD}^M$ and $\XMat\in\mathcal{H}$.
	Then, there exists a function $\delta:\left(0,\infty\right)\rightarrow\left(0,\infty\right)$ such that
	the following holds true:
	For every $\epsilon>0$ and $\WMat\in\mathbb{HPD}^M$ with
	$\norm{\WMat-\XMat}_{2\rightarrow 2}\leq\delta\left(\epsilon\right)$,
	any minimizer $\ZMat$ of
	\begin{align}
		\label{Equation:Theorem:covariance_robust:optimizer}
		&
			\min_{\ZMat\in\mathcal{H}}
			\sum_{m=1}^M\g{\lambda_m\left(\WMat^\frac{1}{2}\ZMat^{-1}\WMat^\frac{1}{2}\right)}
	\end{align}
	obeys $\norm{\XMat-\ZMat}_{2\rightarrow 2}\leq \epsilon$.
	If $g$ is sufficiently convex, then $\delta$ can be chosen such that it is
	linear in a neighborhood around $0$.
\end{Theorem}
\begin{proof}
	This is just a shortened version of \thref{Theorem:tuple_nice} and \thref{Theorem:tuple_convex}.
\end{proof}
The robustness of \thref{Theorem:covariance_robust} yields that any such covariance estimator indeed
recovers matrices correctly if $\WMat$ is not perturbed.
If $\WMat=\XMat$, then $\norm{\WMat-\XMat}_{2\rightarrow 2}\leq\delta\left(\epsilon\right)$ and thus $\norm{\XMat-\ZMat}_{2\rightarrow 2}\leq\epsilon$ for all $\epsilon>0$.
It follows that $\ZMat=\XMat$.
Moreover, the robustness gives control over the estimation error.
If the estimation error $\norm{\XMat-\ZMat}_{2\rightarrow 2}$ is supposed to be
small, then one just needs to control the magnitude of the perturbation $\norm{\WMat-\XMat}_{2\rightarrow 2}$.
This control implies that any function that maps $\WMat$ to any
minimizer of \eqref{Equation:Theorem:covariance_robust:optimizer}
is continuous in all $\XMat\in\mathcal{H}$.
\thref{Theorem:covariance_robust} will be used with $g(x)=x-\Ln{x}$, $\WMat=\frac{1}{K}\YMat\YMat^H$, $\XMat=\sum_{n=1}^N\aVec_n\aVec_n^Hx_n+\SigmaMat$
and $\mathcal{H}=\left\{\sum_{n=1}^N\aVec_n\aVec_n^Hz_n+\SigmaMat\TextForAll \zVec\geq 0\right\}$
to get a result for the relaxed maximum likelihood estimator in \thref{Theorem:ad_errors}.
However, this will only be a part of the proof.
	\subsection{Application to Activity Detection}
In this work codebooks that generate linear operators $\mathcal{A}\left(\zVec\right)=\sum_{n=1}^N\aVec_n\aVec_n^Hz_n$
with a signed kernel condition are considered.
The signed kernel condition was introduced in \cite{NNLR}.
\begin{Definition}
	\label{Definition:skc}
	Let $\mathcal{A}:\mathbb{C}^N\rightarrow\mathbb{C}^{M\times M}$ be a linear operator and $S\in\mathbb{N}$.
	$\mathcal{A}$ is said to have the signed kernel condition of order $S$
	if
	\begin{align}
		\nonumber
			\SetSize{\left\{n\in \SetOf{N}:\vVec_n<0\right\}}
		>
			S
		\TextForAll
			\vVec\in\Kernel{\mathcal{A}}\cap\mathbb{R}^N\setminus\ZeroSet
	\end{align}
	holds true.
\end{Definition}
The signed kernel condition is an equivalent condition for robust
recovery with the non-negative least squares
\cite[Theorem~3.2, Theorem~2.2, Proposition~2.8]{NNLR}.
By \cite[Proposition~3.11]{NNLR} codebooks $\AMat$
such that the linear operator $\mathcal{A}\left(\zVec\right)=\sum_{n=1}^N\aVec_n\aVec_n^Hz_n$ satisfies the signed kernel condition exist
whenever $S\asymp M^2$. In particular, one has:
\begin{Theorem}
	\label{Theorem:sample_rate}
	There exists $\AMat\in\mathbb{C}^{M\times N}$
	with columns $\aVec_n$ for all $n\in\SetOf{N}$ such that
	the linear operator defined by
	$\mathcal{A}\left(\zVec\right):=\sum_{n=1}^N\aVec_n\aVec_n^Hz_n$
	has signed kernel condition of order $S$
	for all $S\leq\left\lceil\frac{1}{2}M^2\right\rceil-1$.
\end{Theorem}
\begin{proof}
	The proof is given in \cite[Proposition~3.11]{NNLR}
	and the entries of $\AMat$ can be chosen as
	\begin{align}
		\label{Equation:Theorem:sample_rate:defA}
			a_{m,n}
		=
			m^{-\frac{1}{2}}
			\Exp{
				i\sqrt{\frac{\pi_m}{\pi_{M+1}}}
				\frac{\pi}{N+N'+1-M^2}
				\left(n-1+N'\right)
			}
	\end{align}
	for all $m\in\SetOf{M},n\in\SetOf{N}$
	where $N':=\max\left\{M^2-N,0\right\}$
	and $\pi_m$ is the $m$-th prime number.
\end{proof}

The matrix $\AMat$ and the linear operator $\mathcal{A}$ are independent of $S$
and one can always choose $S=\left\lceil\frac{1}{2}M^2\right\rceil-1$.
Thus, at this point the parameter $S$ is superfluous in the theorem.
However, the robustness constant introduced in \thref{Definition:robustness_constant}
will be relevant for the later proofs and it depends on $S$ and can improve for smaller choices of $S$.
\par
It should be noted that the construction
\eqref{Equation:Theorem:sample_rate:defA} is not suitable for implementation.
In fact, the construction \eqref{Equation:Theorem:sample_rate:defA}
struggles in simulations since the robustness constant introduced in
\thref{Definition:robustness_constant} of the linear operator
$\mathcal{A}\left(\zVec\right):=\sum_{n=1}^N\aVec_n\aVec_n^Hz_n$ is nearly zero
almost violating \eqref{Equation:Theorem:skc_properties:nnlr_robustness}.
For an implementation one would need to construct other codebooks
with better robustness constants. This is done exemplary once in 
\refP{Section:Simulations}.
\par
Further, it should be noted that the construction
\eqref{Equation:Theorem:sample_rate:defA}
used in the proof of \eqref{Theorem:sample_rate}
is only of theoretical value since it proves
that matrices with a signed kernel condition of the claimed order exist.
By \cite[Remark~3.14]{NNLR} this construction is optimal
in the sense that no other construction can have a higher order
of the signed kernel condition.
However, it should be understood that all further results of this work hold
for all codebooks $\AMat$ such that
$\mathcal{A}\left(\zVec\right):=\sum_{n=1}^N\aVec_n\aVec_n^Hz_n$
has the signed kernel condition of order $S$
and are not constrained to the construction in \eqref{Equation:Theorem:sample_rate:defA}.
More details on the signed kernel condition and the robustness constant
can be found in \refP{Subsection:proof_signed_kernel_condition}
or in \cite{NNLR}.
\par
Given a codebook such that the
$\mathcal{A}\left(\zVec\right):=\sum_{n=1}^N\aVec_n\aVec_n^Hz_n$
has the signed kernel condition of order $S$
the large scale fading coefficients can be estimated
by the non-negative least squares estimator
or the relaxed maximum likelihood estimator arbitrarily good
if the number of receive antennas is sufficiently large.
\begin{Theorem}\label{Theorem:ad_errors}
	Let the matrix $\AMat\in\mathbb{C}^{M\times N}$
	with columns $\aVec_n\in\mathbb{C}^M$ for all $n\in\SetOf{N}$
	be such that the linear operator defined by
	$\mathcal{A}\left(\zVec\right):=\sum_{n=1}^N\aVec_n\aVec_n^Hz_n$
	has signed kernel condition of order
	$S\leq\left\lceil\frac{1}{2}M^2\right\rceil-1$.
	Then, for all $\SigmaMat\in\mathbb{HPD}^M$, $\xVec\in\Sigma_S^N\cap\mathbb{R}^N_+$,
	$\epsilon>0$ and $p\in\left(0,1\right)$
	there exists a sufficiently large $K_0>0$ such that for all $K\geq K_0$ the following holds true:
	Let the columns of
	$\HMat\in\mathbb{C}^{N\times K}$ be $\hVec_k\sim\CGaussianRV{0}{\IDMat}$ and mutually independent,
	the columns of
	$\EMat\in\mathbb{C}^{M\times K}$ be $\eVec_k\sim\CGaussianRV{0}{\SigmaMat}$ and mutually independent
	and $\YMat=\AMat\sqrt{\Diag{\xVec}}\HMat+\EMat$.
	Then, any minimizer $\zVec$ of
	\begin{align}
		&\label{Equation:Theorem:ad_errors:nnlr_optimization}
			\min_{\zVec\in\mathbb{R}_{\geq 0}^N}
			\norm{\sum_{n=1}^N\aVec_n\aVec_n^Hz_n+\SigmaMat-\frac{1}{K}\YMat\YMat^H}_2
	\end{align}
	or
	\begin{align}
		&\label{Equation:Theorem:ad_errors:ml_optimization}
			\min_{\zVec\in\mathbb{R}_{\geq 0}^N}
			\Trace{\left(\sum_{n=1}^N\aVec_n\aVec_n^Hz_n+\SigmaMat\right)^{-1}\frac{1}{K}\YMat\YMat^H}
			+\Ln{\Det{\sum_{n=1}^N\aVec_n\aVec_n^Hz_n+\SigmaMat}}
	\end{align}
	obeys
	\begin{align}
		\nonumber
			\norm{\xVec-\zVec}_2
		\leq
			\epsilon
	\end{align}
	with probability of at least $p$.
\end{Theorem}
\begin{proof}
	The proof is given in
	\refP{Subsection:concentration_argument}
\end{proof}
This theorem states that the minimizers of
\eqref{Equation:NNLR} and \eqref{Equation:trace_log_det}
each converge in probability to the unknown vector of large scale fading coefficients $\xVec$
as $K\rightarrow\infty$.
In general one can use \cite[Theorem~9]{chen} in conjunction with \cite[Theorem~5]{chen}
to show that relaxed maximum likelihood estimation converges in probability to the
true solution.
\thref{Theorem:ad_errors} improves the result \cite[Theorem~5,Theorem~9]{chen}
by reducing the number of required pilot symbols from
$M^2\asymp S\left(\Ln{\ExpE\frac{N}{S}}\right)^2$ to $M^2\asymp S$.
\thref{Theorem:ad_errors} further uses a deterministic construction unlike
\cite[Theorem~9]{chen} and \cite[Theorem~1]{fengler} which use a random construction.
In \cite[Theorem~9]{chen} and \cite[Theorem~1]{fengler}
there always remains a slight chance to not create a matrix with
restricted isometry property. This chance is unaffected by the number of receive antennas $K$
and thus, the probability to achieve any error tolerance by applying \cite[Theorem~5, Theorem~9]{chen}
can not be made arbitrarily high by increasing $K$. Due to the deterministic construction \thref{Theorem:ad_errors}
allows one to do exactly that by increasing the number of receive antennas.
\par
\thref{Theorem:ad_errors} states that the minimizers of relaxed maximum likelihood estimation
converge in probability to the true solution. Thus, the unique identifiability condition in
\cite[Theorem~2~and~Theorem~5]{chen} must be fulfilled for all combinations of $S$ active users
for the codebook of this theorem.
Unlike the convergence from \cite[Theorem~2~and~Theorem~5]{chen},
the proof of \thref{Theorem:ad_errors} gives a direct condition on the number of receive antennas $K$.
To describe the condition additional properties need to be defined due to which
the discussion of the magnitude of $K$ is postponed to \refP{Section:number_of_receive_antennas}.
It will be shown that there is a trade off.
If the number of pilot symbols is reduced from 
$M^2\asymp S\left(\Ln{\ExpE\frac{N}{S}}\right)^2$ to $M^2\asymp S$
the number of receive antennas has to increase significantly.
Further, $K$ may be different depending on whether \eqref{Equation:Theorem:ad_errors:nnlr_optimization}
or \eqref{Equation:Theorem:ad_errors:ml_optimization} is considered.
In particular, $K$ might be larger for the relaxed maximum likelihood estimator.
\par
\thref{Theorem:ad_errors} can be coupled with thresholding to determine the active users.
By choosing $\epsilon>0$ small enough, one can make error probabilities in this case
arbitrarily small.
\begin{Remark}[Thresholding]\label{Remark:thresholding}
	Let the matrix $\AMat\in\mathbb{C}^{M\times N}$
	with columns $\aVec_n\in\mathbb{C}^M$ for all $n\in\SetOf{N}$
	be such that the linear operator defined by
	$\mathcal{A}\left(\zVec\right):=\sum_{n=1}^N\aVec_n\aVec_n^Hz_n$
	has signed kernel condition of order
	$S\leq\left\lceil\frac{1}{2}M^2\right\rceil-1$.
	Then, for all $\SigmaMat\in\mathbb{HPD}^M$, $\xVec\in\Sigma_S^N\cap\mathbb{R}^N_+$,
	$\epsilon\in\left(0,\frac{1}{2}\min_{n:x_n\neq 0}\abs{x_n}\right)$ and $p\in\left(0,1\right)$
	there exists a sufficiently large $K_0>0$ such that for all $K\geq K_0$ the following holds true:
	Let the columns of
	$\HMat\in\mathbb{C}^{N\times K}$ be $\hVec_k\sim\CGaussianRV{0}{\IDMat}$ and mutually independent,
	the columns of
	$\EMat\in\mathbb{C}^{M\times K}$ be $\eVec_k\sim\CGaussianRV{0}{\SigmaMat}$ and mutually independent
	and $\YMat=\AMat\sqrt{\Diag{\xVec}}\HMat+\EMat$.
	Let $\zVec$ be any minimizer of
	\begin{align}
		\nonumber
			\min_{\zVec\in\mathbb{R}_{\geq 0}^N}
			\norm{\sum_{n=1}^N\aVec_n\aVec_n^Hz_n+\SigmaMat-\frac{1}{K}\YMat\YMat^H}_2
	\end{align}
	or
	\begin{align}
		\nonumber
			\min_{\zVec\in\mathbb{R}_{\geq 0}^N}
			\Trace{\left(\sum_{n=1}^N\aVec_n\aVec_n^Hz_n+\SigmaMat\right)^{-1}\frac{1}{K}\YMat\YMat^H}
			+\Ln{\Det{\sum_{n=1}^N\aVec_n\aVec_n^Hz_n+\SigmaMat}}.
	\end{align}
	Let $T:=\{n:x_n\neq 0\}$, $T_1:=\{n:z_n>\epsilon\}$ and $T_2$
	be the indices of the $\SetSize{T}$ largest entries of $\zVec$.
	Then, the probability that $T=T_1$ and $T=T_2$ is at least $p$.
\end{Remark}
\begin{proof}[Proof of \thref{Remark:thresholding}]
	The proof follows from \thref{Theorem:ad_errors}
	after noting that
	$\norm{\xVec-\zVec}_\infty\leq\norm{\xVec-\zVec}_2\leq\epsilon$
	implies $T_1=T=T_2$.
\end{proof}
Both thresholding methods require prior knowledge of $\xVec$ however. In applications
users are generally defined to be active if $x_n>\epsilon_1$ for some known $\epsilon_1>0$.
All other users are treated as noise. Thus, in many applications the prior knowledge required to
choose $\epsilon$ and $T_1$ is known a priori.
\par
It should be noted that both \thref{Theorem:ad_errors} and \thref{Remark:thresholding} are independent
of the signal-to-noise ratio, i.e. any ratio between $\xVec$ and $\SigmaMat$.
Any change in the signal-to-noise ratio is compensated by increasing
the number of receive antennas in \thref{Remark:thresholding}
as explained in \refP{Section:number_of_receive_antennas}.
The convergence of the relaxed maximum likelihood estimators
to the vector of large scale fading coefficients as in \cite{chen}
or as a consequence of \thref{Theorem:ad_errors}
is a weak property.
This property is not enough to precisely pin down the probability of misdetection in the
finite antenna regime.
For this
one needs to consider the finite antenna case and understand the dependence
of $K$ on the other parameters. The discussion in \refP{Section:number_of_receive_antennas}
and proof of \thref{Theorem:ad_errors} explain this dependence up to some degree.
Due to this, the proof of \thref{Theorem:ad_errors} is significantly more important than
the statement itself.
	\subsection{Coordinate Descent for Relaxed Maximum Likelihood Estimation}
Coordinate descent is a common method to solve \eqref{Equation:Theorem:ad_errors:ml_optimization}.
An implementation of coordinate descent with optimal step size is given in
\refP{Algorithm:coord_desc} and was derived in \cite{fengler}.
\begin{algorithm}[H]
\caption{Coordinate Descent for Relaxed Maximum Likelihood Decoding}
\label{Algorithm:coord_desc}
\begin{algorithmic}
	\STATE {\textsc{INPUT:}}
	\STATE \hspace{0.5cm}
	measurement $\YMat\in\mathbb{C}^{M\times K}$, measurement matrix $\AMat\in\mathbb{C}^{M\times N}$
	with columns $\aVec_n$ for $n\in\SetOf{N}$,
	\STATE \hspace{0.5cm}
	covariance matrix $\SigmaMat\in\mathbb{HPD}^M$,
	permutation $\sigma:\SetOf{N}\rightarrow\SetOf{N}$,
	initialization $\xVec^0\in\mathbb{R}_+^N$
	\STATE {\textsc{OUTPUT:}}
	\STATE \hspace{0.5cm}
	estimator $\xVec^\#\in\mathbb{R}_+^N$
	\STATE
	\STATE
	$\SigmaMat'\leftarrow\SigmaMat^{-1}$
	\STATE {\textsc{WHILE}}
		any suitable stopping condition is not yet fulfilled
	{\textsc{DO}}
		\STATE \hspace{0.5cm} {\textsc{FOR}}
			$n'=1,\dots,N$ (ordered!)
		{\textsc{DO}}
			\STATE \hspace{1.0cm}
			$n:=\sigma\left(n'\right)$
			\STATE \hspace{1.0cm}
			$t\leftarrow \max\left\{-x_n,\left(\aVec_n^H\SigmaMat'\frac{1}{K}\YMat\YMat^H\SigmaMat'\aVec_n-\aVec_n^H\SigmaMat'\aVec_n\right)\left(\aVec_n^H\SigmaMat'\aVec_n\right)^{-2}\right\}$
			\STATE \hspace{1.0cm}
			$x_n\leftarrow x_n+t$
			\STATE \hspace{1.0cm}
			$\SigmaMat'\leftarrow\SigmaMat'-t\left(1+t\aVec_n^H\SigmaMat'\aVec_n\right)^{-1}\SigmaMat'\aVec_n\aVec_n^H\SigmaMat'$
			\STATE \hspace{1.0cm}
	\STATE {\textsc{RETURN:}}
		$\xVec^\#\leftarrow\xVec$
\end{algorithmic}
\end{algorithm}
The signed kernel condition will guarantee that any cluster point of this algorithm is a stationary point of the problem.
\begin{Theorem}\label{Theorem:coord_desc_convergence}
	Let the matrix $\AMat\in\mathbb{C}^{M\times N}$
	with columns $\aVec_n\in\mathbb{C}^M$ for all $n\in\SetOf{N}$
	be such that the linear operator defined by
	$\mathcal{A}\left(\zVec\right):=\sum_{n=1}^N\aVec_n\aVec_n^Hz_n$
	has signed kernel condition of order
	$S\leq\left\lceil\frac{1}{2}M^2\right\rceil-1$
	and $\YMat$ have full rank.
	Let $\xVec'_{i,n'}$ and $\SigmaMat'_{i,n'}$ be the vector $\xVec$ and the matrix $\SigmaMat'$
	from \refP{Algorithm:coord_desc} at the end of the $n'$-th iteration of the
	for loop in the $i$-th while loop.
	Let $\xVec_{\left(i-1\right)N+n'}:=\xVec'_{i,n'}$.
	Then, any cluster point of $\left(\xVec_{j}\right)_{j\in\mathbb{N}}$ is a stationary point
	and a coordinate-wise global minimum of
	\begin{align}
		\nonumber
			\min_{\zVec\in\mathbb{R}_{\geq 0}^N}
			\Trace{\left(\sum_{n=1}^N\aVec_n\aVec_n^Hz_n+\SigmaMat\right)^{-1}\frac{1}{K}\YMat\YMat^H}
			+\Ln{\Det{\sum_{n=1}^N\aVec_n\aVec_n^Hz_n+\SigmaMat}}.
	\end{align}
\end{Theorem}
\begin{proof}
	The proof is given in \refP{Subsection:proof_coordinate_descent}.
\end{proof}
The stationary point however, does not have to be a global minimizer.
The algorithm can still suffer from a bad initialization, get stuck in a local minimizer or converge slowly.
	\section{Proof of Theorem \ref{Theorem:covariance_robust}: Robustness of Covariance Estimation}
\noindent
In order to shorten notation one sets
\begin{align}
		\nonumber
		\f{\ZMat,\WMat}
	:=
		\sum_{m=1}^M\g{\lambda_m\left(\WMat^\frac{1}{2}\ZMat^{-1}\WMat^\frac{1}{2}\right)}.
\end{align}
\subsection{Sufficiently Nice Tuples}
In this subsection the part about sufficiently nice $g$ in \thref{Theorem:covariance_robust} is proven.
The following three lemmas contain simple statements about eigenvalues and the compactness of certain sets.
\begin{Lemma}
	\label{Lemma:level_set_compact1}
	Let the tuple $\left(\g{},g_1,g_2,\delta_1,\delta_2\right)$ be sufficiently nice
	and $\mathcal{H}\subset\mathbb{HPD}^M$ be closed in $\mathbb{HPD}^M$.
	For all $\WMat\in\mathbb{HPD}^M$ the problem
	\begin{align}
		\nonumber
		\min_{\ZMat\in\mathcal{H}}
		\sum_{m=1}^M\g{\lambda_m\left(\WMat^\frac{1}{2}\ZMat^{-1}\WMat^\frac{1}{2}\right)}
	\end{align}
	has a minimizer. Further, for all $\gamma\in\mathbb{R}$
	the level set
	$\mathcal{G}:=\left\{\ZMat\in\mathcal{H}:\f{\ZMat,\WMat}\leq \gamma\right\}$
	is compact.
	In particular, for $\gamma\geq M\g{1}$ one has
	\begin{align}
		\label{Equation:Lemma:level_set_compact1:ev_bound}
			\frac{\lambda_1\left(\WMat\right)}{
				g_2\left(\gamma-(M-1)\g{1}\right)
			}
		&\leq
			\lambda_m\left(\ZMat\right)
		\leq
			\frac{\lambda_M\left(\WMat\right)}{
				g_1\left(\gamma-(M-1)\g{1}\right)
			}
		\TextForAll
			\ZMat\in\mathcal{G}
		\TextAnd
			m\in\SetOf{M}.
	\end{align}
\end{Lemma}
\begin{proof}
	Note that if $\gamma<M\g{1}$, then $\mathcal{G}$ is empty and hence compact,
	so without loss of generality let $\gamma\geq M\g{1}$ and $\ZMat\in\mathcal{G}$.
	Let $\vVec$ be an eigenvector for the eigenvalue $\lambda_m\left(\ZMat^{-1}\right)$
	with $\norm{\vVec}_2=1$. Then
	\begin{align}
		\nonumber
			\lambda_m\left(\ZMat^{-1}\right)
		&=
			\lambda_m\left(\ZMat^{-1}\right)\scprod{\vVec}{\vVec}
		=
			\scprod{\ZMat^{-1}\vVec}{\vVec}
		=
			\scprod{
				\ZMat^{-1}\WMat^\frac{1}{2}\WMat^{-\frac{1}{2}}\vVec
			}{\WMat^\frac{1}{2}\WMat^{-\frac{1}{2}}\vVec}
		=
			\scprod{
				\WMat^\frac{1}{2}\ZMat^{-1}\WMat^\frac{1}{2}\WMat^{-\frac{1}{2}}\vVec
			}{\WMat^{-\frac{1}{2}}\vVec}.
	\end{align}
	From this one can get the lower bound
	\begin{align}
		\label{Equation:Lemma:level_set_compact1:eq1}
			\lambda_m\left(\ZMat^{-1}\right)
		\geq&
			\lambda_1\left(\WMat^\frac{1}{2}\ZMat^{-1}\WMat^\frac{1}{2}\right)
			\lambda_1\left(\WMat^{-\frac{1}{2}}\right)^2\norm{\vVec}_2^2
		=
			\frac{
				\lambda_1\left(\WMat^\frac{1}{2}\ZMat^{-1}\WMat^\frac{1}{2}\right)
			}{\lambda_M\left(\WMat\right)}
	\end{align}
	and the upper bound
	\begin{align}
		\label{Equation:Lemma:level_set_compact1:eq2}
			\lambda_m\left(\ZMat^{-1}\right)
		\leq&
			\lambda_M\left(\WMat^\frac{1}{2}\ZMat^{-1}\WMat^\frac{1}{2}\right)
			\lambda_M\left(\WMat^{-\frac{1}{2}}\right)^2\norm{\vVec}_2^2
		=
			\frac{
				\lambda_M\left(\WMat^\frac{1}{2}\ZMat^{-1}\WMat^\frac{1}{2}\right)
			}{\lambda_1\left(\WMat\right)}
	\end{align}
	for every $m$.
	Using $\lambda_m\left(\ZMat\right)=\lambda_{M+1-m}\left(\ZMat^{-1}\right)^{-1}$
	as well as \eqref{Equation:Lemma:level_set_compact1:eq1} and
	\eqref{Equation:Lemma:level_set_compact1:eq2}
	with $M+1-m$ instead of $m$ yields
	\begin{align}
		\label{Equation:Lemma:level_set_compact1:eq3}
			\frac{\lambda_1\left(\WMat\right)}{
				\lambda_M\left(\WMat^\frac{1}{2}\ZMat^{-1}\WMat^\frac{1}{2}\right)
			}
		&\leq
			\lambda_m\left(\ZMat\right)
		\leq
			\frac{\lambda_M\left(\WMat\right)}{
				\lambda_1\left(\WMat^\frac{1}{2}\ZMat^{-1}\WMat^\frac{1}{2}\right)
			}
	\end{align}
	for all $m\in\SetOf{M}$.
	On the other hand, $\ZMat\in \mathcal{G}$ yields
	\begin{align}
		\nonumber
			\g{\lambda_m\left(\WMat^\frac{1}{2}\ZMat^{-1}\WMat^\frac{1}{2}\right)}
		\leq
			-(M-1)\g{1}+\sum_{m'=1}^M
			\g{\lambda_{m'}\left(\WMat^\frac{1}{2}\ZMat^{-1}\WMat^\frac{1}{2}\right)}
		\leq
			\gamma-(M-1)\g{1}.
	\end{align}
	for all $m\in\SetOf{M}$. Since $\gamma-(M-1)\g{1}\geq \g{1}$, one can apply the definition of
	the almost inverse functions $g_1$ and $g_2$ to this and get
	\begin{align}
		\nonumber
			g_1\left(\gamma-(M-1)\g{1}\right)
		\leq
			\lambda_m\left(\WMat^\frac{1}{2}\ZMat^{-1}\WMat^\frac{1}{2}\right)
		\leq
			g_2\left(\gamma-(M-1)\g{1}\right)
	\end{align}
	for all $m\in\SetOf{M}$. 
	Together with \eqref{Equation:Lemma:level_set_compact1:eq3}
	this yields \eqref{Equation:Lemma:level_set_compact1:ev_bound}
	and that $\mathcal{G}$ is bounded.
	Note that by continuity $\mathcal{G}$ is closed in $\mathcal{H}$
	and thus in $\mathbb{HPD}^M$ but that does not mean $\mathcal{G}$ is closed in $\mathbb{H}^M$
	since $\mathbb{HPD}^M$ is not a complete metric space.
	However, due to \eqref{Equation:Lemma:level_set_compact1:ev_bound} the level set is bounded
	away from the boundary of $\mathbb{HPD}^M$ and hence the level set is also closed in
	the linear space $\mathbb{H}^M$. This proves the compactness of the level set.
	\par
	In order to prove the existence of minimizers let $\gamma>\inf_{\ZMat\in\mathcal{H}}\f{\ZMat,\WMat}$.
	Since $\ZMat\mapsto\f{\ZMat,\WMat}$ is continuous the function attains its minimal value
	over the compact set $\mathcal{G}$ at some $\ZMat$. Due to the definition of the level set,
	this has to be a minimizer of the problem.
\end{proof}
\begin{Lemma}
	\label{Lemma:level_set_compact2}
	Let the tuple $\left(\g{},g_1,g_2,\delta_1,\delta_2\right)$ be sufficiently nice.
	For all $\XMat\in\mathbb{HPD}^M$ and $\gamma\in\mathbb{R}$ the level set
	\begin{align}
			\mathcal{F}
		:=
			\left\{\WMat\in\mathbb{HPD}^M:\f{\XMat,\WMat}\leq \gamma\right\}
	\end{align}
	is compact.
	In particular, for $\gamma\geq M\g{1}$ one has
	\begin{align}
		\label{Equation:Lemma:level_set_compact2:ev_bound}
			\lambda_1\left(\XMat\right)g_1\left(\gamma-(M-1)\g{1}\right)
		\leq
			\lambda_m\left(\WMat\right)
		\leq
			\lambda_M\left(\XMat\right)g_2\left(\gamma-(M-1)\g{1}\right)
		\TextForAll
			\WMat\in\mathcal{F}
		\TextAnd
			m\in\SetOf{M}.
	\end{align}
\end{Lemma}
\begin{proof}
	Note that if $\gamma<M\g{1}$, then $\mathcal{F}$ is empty and hence compact,
	so without loss of generality let $\gamma\geq M\g{1}$.
	Let $\vVec$ be an eigenvector for the eigenvalue $\lambda_m\left(\WMat\right)$
	with $\norm{\vVec}_2=1$. Then $\vVec$ is an eigenvector for the eigenvalue
	$\lambda_m\left(\WMat^\frac{1}{2}\right)=\lambda_m\left(\WMat\right)^\frac{1}{2}$
	and hence
	\begin{align}
		\nonumber
			\lambda_m\left(\WMat\right)
		&=
			\lambda_m\left(\WMat\right)
			\frac{
				\scprod{\XMat^{-1}\vVec}{\vVec}
			}{
				\scprod{\XMat^{-1}\vVec}{\vVec}
			}
		=
			\frac{
				\scprod{\XMat^{-1}\lambda_m\left(\WMat\right)^\frac{1}{2}\vVec}{
					\lambda_m\left(\WMat\right)^\frac{1}{2}\vVec
				}
			}{
				\scprod{\XMat^{-1}\vVec}{\vVec}
			}
		=
			\frac{
				\scprod{\XMat^{-1}\lambda_m\left(\WMat^\frac{1}{2}\right)\vVec}{
					\lambda_m\left(\WMat^\frac{1}{2}\right)\vVec
				}
			}{
				\scprod{\XMat^{-1}\vVec}{\vVec}
			}
		\\=&\nonumber
			\frac{
				\scprod{\XMat^{-1}\WMat^\frac{1}{2}\vVec}{\WMat^\frac{1}{2}\vVec}
			}{
				\scprod{\XMat^{-1}\vVec}{\vVec}
			}
		=
			\frac{
				\scprod{\WMat^\frac{1}{2}\XMat^{-1}\WMat^\frac{1}{2}\vVec}{\vVec}
			}{
				\scprod{\XMat^{-1}\vVec}{\vVec}
			}.
	\end{align}
	From this one can get the lower bound
	\begin{align}
		\label{Equation:Lemma:level_set_compact2:eq1}
			\lambda_m\left(\WMat\right)
		\geq&
			\frac{
				\lambda_1\left(\WMat^\frac{1}{2}\XMat^{-1}\WMat^\frac{1}{2}\right)
			}{
				\lambda_M\left(\XMat^{-1}\right)
			}
		=
			\lambda_1\left(\WMat^\frac{1}{2}\XMat^{-1}\WMat^\frac{1}{2}\right)
			\lambda_1\left(\XMat\right)
	\end{align}
	and the upper bound
	\begin{align}
		\label{Equation:Lemma:level_set_compact2:eq2}
			\lambda_m\left(\WMat\right)
		\leq&
			\frac{
				\lambda_M\left(\WMat^\frac{1}{2}\XMat^{-1}\WMat^\frac{1}{2}\right)
			}{
				\lambda_1\left(\XMat^{-1}\right)
			}
		=
			\lambda_M\left(\WMat^\frac{1}{2}\XMat^{-1}\WMat^\frac{1}{2}\right)
			\lambda_M\left(\XMat\right)
	\end{align}
	for all $m\in\SetOf{M}$.
	On the other hand, $\WMat\in \mathcal{F}$ yields
	\begin{align}
		\nonumber
			\g{\lambda_m\left(\WMat^\frac{1}{2}\XMat^{-1}\WMat^\frac{1}{2}\right)}
		\leq
			-(M-1)\g{1}+\sum_{m'=1}^M
			\g{\lambda_{m'}\left(\WMat^\frac{1}{2}\XMat^{-1}\WMat^\frac{1}{2}\right)}
		\leq
			\gamma-(M-1)\g{1}
	\end{align}
	for all $m\in\SetOf{M}$.
	Since $\gamma-(M-1)\g{1}\geq \g{1}$, one can apply the definition of
	the almost inverse functions $g_1$ and $g_2$ to this and get
	\begin{align}
		\nonumber
			g_1\left(\gamma-(M-1)\g{1}\right)
		\leq
			\lambda_m\left(\WMat^\frac{1}{2}\XMat^{-1}\WMat^\frac{1}{2}\right)
		\leq
			g_2\left(\gamma-(M-1)\g{1}\right)
	\end{align}
	for all $m\in\SetOf{M}$.
	Together with \eqref{Equation:Lemma:level_set_compact2:eq1} and
	\eqref{Equation:Lemma:level_set_compact2:eq2}
	this yields \eqref{Equation:Lemma:level_set_compact2:ev_bound}
	and that $\mathcal{F}$ is bounded.
	Note that by continuity $\mathcal{F}$ is closed in $\mathcal{H}$
	and thus in $\mathbb{HPD}^M$ but that does not mean $\mathcal{F}$ is closed in $\mathbb{H}^M$
	since $\mathbb{HPD}^M$ is not a complete metric space.
	However, due to \eqref{Equation:Lemma:level_set_compact2:ev_bound} the level set is bounded
	away from the boundary of $\mathbb{HPD}^M$ and hence the level set is also closed in
	the linear space $\mathbb{H}^M$. This proves the compactness of the level set.
\end{proof}
\begin{Lemma}\label{Lemma:compact_neighboorhood}
	For all $\XMat\in\mathbb{HPD}^M$ and $\beta<\lambda_1\left(\XMat\right)$ the set
	$\mathcal{C}:=\left\{\ZMat\in\mathbb{HPD}^M:\norm{\ZMat-\XMat}_{2\rightarrow 2}\leq\beta\right\}$
	is compact in $\mathbb{HPD}^M$, and
	\begin{align}
		\label{Equation:Lemma:compact_neighboorhood:ev_bound}
			\lambda_1\left(\XMat\right)-\beta
		\leq
			\lambda_m\left(\ZMat\right)
		\leq
			\lambda_M\left(\XMat\right)+\beta
	\end{align}
	for all $\ZMat\in\mathcal{C}$ and $m\in\SetOf{M}$.
\end{Lemma}
\begin{proof}
	For every $\ZMat\in\mathcal{C}$ and $m\in\SetOf{M}$ one has
	\begin{align}
		&\label{Equation:Lemma:compact_neighboorhood:eq1}
			\abs{\lambda_m\left(\ZMat-\XMat\right)}
		\leq
			\sup_{m'=1,\dots,M}\lambda_{m'}\abs{\left(\ZMat-\XMat\right)}
		=
			\norm{\ZMat-\XMat}_{2\rightarrow 2}
		\leq
			\beta.
	\end{align}
	If $\vVec$ is an eigenvector to the eigenvalue $\lambda_m\left(\ZMat\right)$ with $\norm{\vVec}_2=1$,
	\begin{align}
		\label{Equation:Lemma:compact_neighboorhood:eq2}
			\lambda_m\left(\ZMat\right)
		=&
			\scprod{\ZMat\vVec}{\vVec}
		=
			\scprod{\left(\ZMat-\XMat\right)\vVec}{\vVec}+\scprod{\XMat\vVec}{\vVec}.
	\end{align}
	Note that $\XMat$ is positive definite, but $\ZMat-\XMat$ not necessarily is.
	Keeping this in mind one can use \eqref{Equation:Lemma:compact_neighboorhood:eq1} on
	\eqref{Equation:Lemma:compact_neighboorhood:eq2} to get
	\begin{align}
		\label{Equation:Lemma:compact_neighboorhood:eq3}
			\lambda_m\left(\ZMat\right)
		\leq&
			\lambda_M\left(\ZMat-\XMat\right)+\lambda_M\left(\XMat\right)
		\leq
			\beta+\lambda_M\left(\XMat\right).
	\end{align}
	On the other hand, one can use \eqref{Equation:Lemma:compact_neighboorhood:eq1} on
	\eqref{Equation:Lemma:compact_neighboorhood:eq2} to obtain
	\begin{align}
		\label{Equation:Lemma:compact_neighboorhood:eq4}
			\lambda_m\left(\ZMat\right)
		\geq&
			\lambda_1\left(\ZMat-\XMat\right)+\lambda_1\left(\XMat\right)
		\geq
			-\beta+\lambda_1\left(\XMat\right).
	\end{align}
	From \eqref{Equation:Lemma:compact_neighboorhood:eq3} and
	\eqref{Equation:Lemma:compact_neighboorhood:eq4} the statement
	\eqref{Equation:Lemma:compact_neighboorhood:ev_bound} follows.
	From \eqref{Equation:Lemma:compact_neighboorhood:ev_bound} one gets
	\begin{align}
		\nonumber
			\mathcal{C}
		\subset
			\mathcal{F}
		:=
			\left\{
				\ZMat\in\mathbb{H}^M:
				\beta-\lambda_1\left(\XMat\right)
				\leq\lambda_m\left(\ZMat\right)\leq
				\lambda_M\left(\XMat\right)+\beta
				\TextForAll
				m\in\SetOf{M}
			\right\}
	\end{align}
	and the right hand side is compact in the finite-dimensional linear space $\mathbb{H}^M$ over
	$\mathbb{C}$ and hence compact in $\mathbb{HPD}^M$.
	Due to continuity of the norm
	$\mathcal{C}$ is closed in $\mathbb{HPD}^M$ and thus also closed in the compact
	$\mathcal{F}$. It follows that $\mathcal{C}$ is compact.
\end{proof}
Before the statement about sufficiently nice $g$ can be proven, it is shown that
small perturbations in $\WMat$ lead to small changes in the objective function.
This can be shown directly by a continuity argument; however, the next result
also gives the exact dependence of that on $\delta_1$ for a sufficiently nice tuple.
\begin{Proposition}\label{Proposition:obj_cont}
	Let the tuple $\left(\g{},g_1,g_2,\delta_1,\delta_2\right)$ be sufficiently nice
	and $\mathcal{H}\subset\mathbb{HPD}^M$ be closed in $\mathbb{HPD}^M$.
	For all $\XMat\in\mathcal{H}$ there exists a function
	$\delta_3:\left(0,\infty\right)\rightarrow\left(0,\infty\right)$
	such that the following holds true:
	For every $\epsilon>0$ and $\WMat\in\mathbb{HPD}^M$ with
	$\norm{\XMat-\WMat}_{2\rightarrow 2}\leq\delta_3\left(\epsilon\right)$
	one has $\f{\XMat,\WMat}\leq M\g{1}+\epsilon$.
	Moreover, one can choose
	\begin{align}
		\label{Equation:Proposition:obj_cont:delta_3_def}
			\delta_3\left(\epsilon\right)
		:=
			\min\left\{
				\lambda_1\left(\XMat\right)
				\left(
					\frac{\lambda_1\left(\XMat\right)-\beta}{\lambda_M\left(\XMat\right)+\beta}
				\right)^\frac{1}{2}
				\delta_1\left(\frac{\epsilon}{M}\right),
				\beta
			\right\}.
	\end{align}
\end{Proposition}
\begin{proof}
	Consider the choice \eqref{Equation:Proposition:obj_cont:delta_3_def} and let $\WMat\in\mathbb{HPD}^M$
	with $\norm{\XMat-\WMat}_{2\rightarrow 2}\leq\delta_3\left(\epsilon\right)$.
	Since $\norm{\XMat-\WMat}_{2\rightarrow 2}\leq\beta$, \thref{Lemma:compact_neighboorhood} yields
	\eqref{Equation:Lemma:compact_neighboorhood:ev_bound} which can be used to get
	\begin{align}
		&\nonumber
			\lambda_1\left(\XMat\right)
			\left(
				\frac{\lambda_1\left(\XMat\right)-\beta}{\lambda_M\left(\XMat\right)+\beta}
			\right)^\frac{1}{2}
			\delta_1\left(\frac{\epsilon}{M}\right)
		\geq
			\delta_3\left(\epsilon\right)
		\geq
			\norm{\XMat-\WMat}_{2\rightarrow 2}
		=
			\norm{
				\WMat^\frac{1}{2}
				\left(\IDMat-\WMat^\frac{1}{2}\XMat^{-1}\WMat^\frac{1}{2}\right)
				\WMat^{-\frac{1}{2}}\XMat
			}_{2\rightarrow 2}
		\\\geq&\nonumber
			\lambda_1\left(\XMat\right)\lambda_1\left(\WMat^\frac{1}{2}\right)
			\lambda_1\left(\WMat^{-\frac{1}{2}}\right)
			\norm{\IDMat-\WMat^\frac{1}{2}\XMat^{-1}\WMat^\frac{1}{2}}_{2\rightarrow 2}
		=
			\lambda_1\left(\XMat\right)
			\left(\frac{\lambda_1\left(\WMat\right)}{\lambda_M\left(\WMat\right)}\right)^\frac{1}{2}
			\norm{\IDMat-\WMat^\frac{1}{2}\XMat^{-1}\WMat^\frac{1}{2}}_{2\rightarrow 2}
		\\\geq&\nonumber
			\lambda_1\left(\XMat\right)
			\left(
				\frac{\lambda_1\left(\XMat\right)-\beta}{\lambda_M\left(\XMat\right)+\beta}
			\right)^\frac{1}{2}
			\norm{\IDMat-\WMat^\frac{1}{2}\XMat^{-1}\WMat^\frac{1}{2}}_{2\rightarrow 2}.
	\end{align}
	and thus
	\begin{align}
		\nonumber
			\abs{\lambda_m\left(\IDMat\right)
				-\lambda_m\left(
					\WMat^\frac{1}{2}\XMat^{-1}\WMat^\frac{1}{2}
				\right)
			}
		\leq
			\norm{\IDMat-\WMat^\frac{1}{2}\XMat^{-1}\WMat^\frac{1}{2}}_{2\rightarrow 2}
		\leq
			\delta_1\left(\frac{\epsilon}{M}\right).
	\end{align}
	By this and the continuity of $\g{}$ around $1$ one gets
	\begin{align}
		\nonumber
			\abs{\sum_{m=1}^M\g{\lambda_m\left(\WMat^\frac{1}{2}\XMat^{-1}\WMat^\frac{1}{2}\right)}-M\g{1}}
		=&
			\abs{\sum_{m=1}^M\g{\lambda_m\left(\WMat^\frac{1}{2}\XMat^{-1}\WMat^\frac{1}{2}\right)}
				-\sum_{m=1}^M\g{\lambda_m\left(\IDMat\right)}
			}
		\\\leq&\nonumber
			\sum_{m=1}^M\abs{\g{\lambda_m\left(\WMat^\frac{1}{2}\XMat^{-1}\WMat^\frac{1}{2}\right)}
				-\sum_{m=1}^M\g{\lambda_m\left(\IDMat\right)}
			}
		\leq
			M\frac{\epsilon}{M}
		\leq
			\epsilon
	\end{align}
	which yields the claim.
\end{proof}
At this point the part about sufficiently nice $g$ in \thref{Theorem:covariance_robust} can be proven.
\begin{Theorem}\label{Theorem:tuple_nice}
	Let the tuple $\left(\g{},g_1,g_2,\delta_1,\delta_2\right)$ be sufficiently nice
	and $\mathcal{H}\subset\mathbb{HPD}^M$ be closed in $\mathbb{HPD}^M$.
	For all $\XMat\in\mathcal{H}$, $0<\beta<\lambda_1\left(\XMat\right)$ and $\eta>0$
	there exists a function $\delta:\left(0,\infty\right)\rightarrow\left(0,\infty\right)$ such that
	the following holds true:
	For every $\epsilon>0$ and $\WMat\in\mathbb{HPD}^M$ with
	$\norm{\WMat-\XMat}_{2\rightarrow 2}\leq\delta\left(\epsilon\right)$,
	any minimizer $\ZMat$ of
	\begin{align}
		&\label{Equation:tuple_nice:optimization}
			\min_{\ZMat\in\mathcal{H}}
			\sum_{m=1}^M\g{\lambda_m\left(\WMat^\frac{1}{2}\ZMat^{-1}\WMat^\frac{1}{2}\right)}
	\end{align}
	obeys $\norm{\XMat-\ZMat}_{2\rightarrow 2}\leq \epsilon$.
	In particular, $\delta$ can be chosen as
	\begin{align}
		\nonumber
			\delta\left(\epsilon\right)
		:=&
			\min\Bigg\{
				\lambda_1\left(\XMat\right)
				\left(
					\frac{\lambda_1\left(\XMat\right)-\beta}{\lambda_M\left(\XMat\right)+\beta}
				\right)^\frac{1}{2}
				\delta_1\left(M^{-1}\delta_2\left(
					\frac{1}{2\lambda_M\left(\XMat\right)}
					\left(
						\frac{\lambda_1\left(\XMat\right)-\beta}{\lambda_M\left(\XMat\right)+\beta}
					\right)^\frac{1}{2}
					\frac{g_1\left(\g{1}+\eta\right)}{g_2\left(\g{1}+\eta\right)}
					\epsilon
				\right)\right),
		\\&\label{Equation:Theorem:tuple_nice:delta_def}
				\frac{1}{2}\epsilon,
				\lambda_1\left(\XMat\right)
				\left(
					\frac{\lambda_1\left(\XMat\right)-\beta}{\lambda_M\left(\XMat\right)+\beta}
				\right)^\frac{1}{2}
				\delta_1\left(\frac{\eta}{M}\right),
				\beta
			\Bigg\}
	\end{align}
\end{Theorem}
\begin{proof}
	Consider the choice of \eqref{Equation:Theorem:tuple_nice:delta_def} and let
	$\delta_3$ be the function from \thref{Proposition:obj_cont} with $\beta$. Then
	\begin{align}
		\nonumber
			\delta\left(\epsilon\right)
		=&
			\min\left\{
				\delta_3\left(\delta_2\left(
					\frac{1}{2\lambda_M\left(\XMat\right)}
					\left(
						\frac{\lambda_1\left(\XMat\right)-\beta}{\lambda_M\left(\XMat\right)+\beta}
					\right)^\frac{1}{2}
					\frac{g_1\left(\g{1}+\eta\right)}{g_2\left(\g{1}+\eta\right)}
					\epsilon
				\right)\right),
				\frac{1}{2}\epsilon,
				\delta_3\left(\eta\right),
				\beta
			\right\}.
	\end{align}
	Now let $\WMat\in\mathbb{HPD}^M$ such that $\norm{\WMat-\XMat}_{2\rightarrow 2}\leq\delta\left(\epsilon\right)$
	and $\ZMat$ be the minimizer of \eqref{Equation:tuple_nice:optimization}.
	The minimization property yields	
	\begin{align}
		\label{Equation:Theorem:tuple_nice:eq1}
			\sum_{m'=1}^M\g{\lambda_{m'}\left(\WMat^\frac{1}{2}\ZMat^{-1}\WMat^\frac{1}{2}\right)}
		=&
			\f{\ZMat,\WMat}
		\leq
			\f{\XMat,\WMat}.
	\end{align}
	By
	\begin{align}
		\nonumber
			\norm{\XMat-\WMat}_{2\rightarrow 2}
		\leq
			\delta\left(\epsilon\right)
		\leq
			\delta_3\left(\delta_2\left(
				\frac{1}{2\lambda_M\left(\XMat\right)}
				\left(
					\frac{\lambda_1\left(\XMat\right)-\beta}{\lambda_M\left(\XMat\right)+\beta}
				\right)^\frac{1}{2}
				\frac{g_1\left(\g{1}+\eta\right)}{g_2\left(\g{1}+\eta\right)}
				\epsilon
			\right)\right)
	\end{align}
	and by the definition of $\delta_3$ it follows that
	\begin{align}
		\label{Equation:Theorem:tuple_nice:eq2}
			\f{\XMat,\WMat}
		\leq
			M\g{1}
			+\delta_2\left(
				\frac{1}{2\lambda_M\left(\XMat\right)}
				\left(
					\frac{\lambda_1\left(\XMat\right)-\beta}{\lambda_M\left(\XMat\right)+\beta}
				\right)^\frac{1}{2}
				\frac{g_1\left(\g{1}+\eta\right)}{g_2\left(\g{1}+\eta\right)}
				\epsilon
			\right).
	\end{align}
	Combining $\g{}\geq \g{1}$ with \eqref{Equation:Theorem:tuple_nice:eq1} and
	\eqref{Equation:Theorem:tuple_nice:eq2}, yields
	\begin{align}
		\nonumber
			\g{\lambda_m\left(\WMat^\frac{1}{2}\ZMat^{-1}\WMat^\frac{1}{2}\right)}
		\leq&
			\sum_{m'=1}^M\g{\lambda_{m'}\left(\WMat^\frac{1}{2}\ZMat^{-1}\WMat^\frac{1}{2}\right)}
			-(M-1)\g{1}
		\\\leq&\nonumber
			\g{1}
			+\delta_2\left(
				\frac{1}{2\lambda_M\left(\XMat\right)}
				\left(
					\frac{\lambda_1\left(\XMat\right)-\beta}{\lambda_M\left(\XMat\right)+\beta}
				\right)^\frac{1}{2}
				\frac{g_1\left(\g{1}+\eta\right)}{g_2\left(\g{1}+\eta\right)}
				\epsilon
			\right)
	\end{align}
	for all $m$. By the definition of $\delta_2$ one has
	\begin{align}
		\nonumber
			\abs{\lambda_m\left(\WMat^\frac{1}{2}\ZMat^{-1}\WMat^\frac{1}{2}\right)-1}
		\leq
			\frac{1}{2\lambda_M\left(\XMat\right)}
			\left(
				\frac{\lambda_1\left(\XMat\right)-\beta}{\lambda_M\left(\XMat\right)+\beta}
			\right)^\frac{1}{2}
			\frac{g_1\left(\g{1}+\eta\right)}{g_2\left(\g{1}+\eta\right)}
			\epsilon
	\end{align}
	for all $m\in\SetOf{M}$.
	Using this yields
	\begin{align}
		\nonumber
			\norm{\WMat-\ZMat}_{2\rightarrow 2}
		=&
			\norm{
				\ZMat\WMat^{-\frac{1}{2}}
				\left(\WMat^\frac{1}{2}\ZMat^{-1}\WMat^\frac{1}{2}-\IDMat\right)
				\WMat^\frac{1}{2}
			}_{2\rightarrow 2}
		\\\leq&\nonumber
			\lambda_M\left(\ZMat\right)
			\lambda_M\left(\WMat^{-\frac{1}{2}}\right)
			\lambda_M\left(\WMat^\frac{1}{2}\right)
			\norm{\WMat^\frac{1}{2}\ZMat^{-1}\WMat^\frac{1}{2}-\IDMat}_{2\rightarrow 2}
		\\=&\nonumber
			\lambda_M\left(\ZMat\right)
			\left(\frac{\lambda_M\left(\WMat\right)}{\lambda_1\left(\WMat\right)}\right)^\frac{1}{2}
			\sup_{m=1,\dots,M}
			\abs{\lambda_m\left(\WMat^\frac{1}{2}\ZMat^{-1}\WMat^\frac{1}{2}-1\right)}
		\\=&\nonumber
			\lambda_M\left(\ZMat\right)
			\left(\frac{\lambda_M\left(\WMat\right)}{\lambda_1\left(\WMat\right)}\right)^\frac{1}{2}
			\sup_{m=1,\dots,M}
			\abs{\lambda_m\left(\WMat^\frac{1}{2}\ZMat^{-1}\WMat^\frac{1}{2}\right)-1}
		\\\leq&\label{Equation:Theorem:tuple_nice:eq3}
			\frac{1}{2}
			\frac{\lambda_M\left(\ZMat\right)}{\lambda_M\left(\XMat\right)}
			\left(
				\frac{\lambda_1\left(\XMat\right)-\beta}{\lambda_M\left(\XMat\right)+\beta}
			\right)^\frac{1}{2}
			\left(\frac{\lambda_M\left(\WMat\right)}{\lambda_1\left(\WMat\right)}\right)^\frac{1}{2}
			\frac{g_1\left(\g{1}+\eta\right)}{g_2\left(\g{1}+\eta\right)}
			\epsilon.
	\end{align}
	Due to $\norm{\XMat-\WMat}_{2\rightarrow 2}\leq\delta\left(\epsilon\right)\leq\beta$
	and \thref{Lemma:compact_neighboorhood}	one gets
	$\lambda_1\left(\XMat\right)-\beta\leq\lambda_m\left(\WMat\right)\leq\lambda_M\left(\XMat\right)+\beta$
	which can be applied to \eqref{Equation:Theorem:tuple_nice:eq3} to yield
	\begin{align}
		\label{Equation:Theorem:tuple_nice:eq4}
			\norm{\WMat-\ZMat}_{2\rightarrow 2}
		\leq&
			\frac{1}{2}
			\frac{\lambda_M\left(\ZMat\right)}{\lambda_M\left(\XMat\right)}
			\frac{g_1\left(\g{1}+\eta\right)}{g_2\left(\g{1}+\eta\right)}
			\epsilon.
	\end{align}
	By the fact that $\ZMat$ is a minimizer, the inequality
	$\norm{\XMat-\WMat}_{2\rightarrow 2}\leq\delta\left(\epsilon\right)\leq\delta_3\left(\eta\right)$
	and the definition of $\delta_3$ one gets
	\begin{align}
		\label{Equation:Theorem:tuple_nice:eq5}
			\f{\ZMat,\WMat}
		\leq	
			\f{\XMat,\WMat}
		\leq
			M\g{1}+\eta.
	\end{align}
	Thus $\ZMat\in\left\{\ZMat'\in\mathcal{H}:\f{\ZMat',\WMat}\leq M\g{1}+\eta\right\}$
	and \thref{Lemma:level_set_compact1} yields
	\begin{align}
		\label{Equation:Theorem:tuple_nice:eq6}
			\frac{\lambda_1\left(\WMat\right)}{g_2\left(\g{1}+\eta\right)}
		\leq
			\lambda_m\left(\ZMat\right)
		\leq
			\frac{\lambda_M\left(\WMat\right)}{g_1\left(\g{1}+\eta\right)}.
	\end{align}
	Due to \eqref{Equation:Theorem:tuple_nice:eq5}
	$\WMat\in\left\{\WMat'\in\mathbb{HPD}^M:\f{\XMat,\WMat'}\leq M\g{1}+\eta\right\}$
	and \thref{Lemma:level_set_compact2} yields
	\begin{align}
		\label{Equation:Theorem:tuple_nice:eq7}
			\lambda_1\left(\XMat\right)g_1\left(\g{1}+\eta\right)
		\leq
			\lambda_m\left(\WMat\right)
		\leq
			\lambda_M\left(\XMat\right)g_2\left(\g{1}+\eta\right).
	\end{align}
	Combining \eqref{Equation:Theorem:tuple_nice:eq6} with \eqref{Equation:Theorem:tuple_nice:eq7}
	yields
	\begin{align}
		\nonumber
			\lambda_1\left(\XMat\right)\frac{g_1\left(\g{1}+\eta\right)}{g_2\left(\g{1}+\eta\right)}
		\leq
			\lambda_m\left(\ZMat\right)
		\leq
			\lambda_M\left(\XMat\right)\frac{g_2\left(\g{1}+\eta\right)}{g_1\left(\g{1}+\eta\right)}.
	\end{align}
	Applying this to \eqref{Equation:Theorem:tuple_nice:eq4} gives
	\begin{align}
		\nonumber
			\norm{\WMat-\ZMat}_{2\rightarrow 2}
		\leq
			\frac{1}{2}\epsilon.
	\end{align}
	Applying this and $\norm{\XMat-\WMat}_{2\rightarrow 2}\leq\delta\left(\epsilon\right)\leq\frac{1}{2}\epsilon$
	yields
	\begin{align}
		\nonumber
			\norm{\XMat-\ZMat}_{2\rightarrow 2}
		\leq&
			\norm{\XMat-\WMat}_{2\rightarrow 2}
			+\norm{\WMat-\ZMat}_{2\rightarrow 2}
		\leq
			\epsilon
	\end{align}
	which completes the proof.
\end{proof}
In general, one wants $\delta$ to be as large as possible.
$\beta$ and $\eta$ should be considered constants so that $\delta$ scales at best
linearly in $\epsilon$ around $0$. However, the first part of the minimimum which includes
$\delta_1,\delta_2$ can create a worse scaling.
	\subsection{Sufficiently Convex Tuples}
In this subsection the part about sufficiently convex $g$ of \thref{Theorem:covariance_robust} is proven.
Before that it is shown that sufficiently convex tuples actually are sufficiently nice.
\begin{Lemma}\label{Lemma:tuple_convex}
	Let $\left(\g{},g_1,g_2,\nu,\epsilon_0\right)$ be sufficiently convex and
	\begin{align}
		\label{Equation:Lemma:tuple_convex:def:delta1delta2}
			\delta_1\left(\epsilon\right)
		:=
			1-g_1\left(\g{1}+\epsilon\right)
		\TextAnd
			\delta_2\left(\epsilon\right)
		:=
			\g{1+\epsilon}-\g{1}.
	\end{align}
	Then, $\left(\g{},g_1,g_2,\delta_1,\delta_2\right)$ is sufficiently nice and
	\begin{align}
		\label{Equation:Lemma:tuple_convex:linear_scaling}
			\nu\epsilon M^{-1}
		\leq
			\delta_1\left(M^{-1}\delta_2\left(\epsilon\right)\right)
		\TextForAll
			\epsilon\in\left(0,\epsilon_0M\right].
	\end{align}
\end{Lemma}
\begin{proof}
	\refP{Property:Definition:tuple_nice:grow} and \refP{Property:Definition:tuple_nice:cont}
	are given by assumption.
	Due to the monotonicity the almost inverse functions
	need to be the inverse functions of $\g{}$ on $\left(0,1\right]$
	and $\left[1,\infty\right)$ respectively, and thus \refP{Property:Definition:tuple_nice:inver}
	is fulfilled.
	\par
	By the monotonicity $1$ is the unique minimizer of $\g{}$ and
	$\delta_2\left(\epsilon\right)>0$ is well defined for all $\epsilon>0$.
	In order to show \refP{Property:Definition:tuple_nice:minimizer} let
	$x\in\left(0,\infty\right)$ and $\epsilon>0$ be such that
	$\g{x}-\g{1}\leq\delta_2\left(\epsilon\right)$.
	It follows that $\g{x}\leq\g{1+\epsilon}$ and the monotonicity of $\g{}$ in
	$\left[1,\infty\right)$ yields
	\begin{align}
		\label{Equation:Lemma:tuple_convex:delta_2:upper_bound}
			x
		\leq
			1+\epsilon.
	\end{align}
	If $\epsilon\geq 1$, then $1-\epsilon\leq 0\leq x$ holds true as well.
	Now suppose that $\epsilon<1$. Then,
	$\g{x}\leq\g{1+\epsilon}\leq\g{1-\epsilon}$ and the monotonicity of $\g{}$ in
	$\left(0,1\right]$ yields $1-\epsilon\leq x$.
	Thus, for all $\epsilon>0$ one gets
	\begin{align}
		\nonumber
			1-\epsilon
		\leq
			x.
	\end{align}
	Together with \eqref{Equation:Lemma:tuple_convex:delta_2:upper_bound} this yields
	$\abs{x-1}\leq\epsilon$ and thus \refP{Property:Definition:tuple_nice:minimizer}.
	\par
	For any $\epsilon>0$ assume that
	$\epsilon_1:=1-g_1\left(\g{1}+\epsilon\right)>g_2\left(\g{1}+\epsilon\right)-1=:\epsilon_2$.
	Then,
	\begin{align}
		\nonumber
			\g{1-\epsilon_1}
		=&
			\g{g_1\left(\g{1}+\epsilon\right)}
		=
			\g{1}+\epsilon
		=
			\g{g_2\left(\g{1}+\epsilon\right)}
		=
			\g{1+\epsilon_2}.
	\end{align}
	Using that $\g{}$ is strictly monotonically increasing in $\left[1,\infty\right)$ 
	and $\g{1+\epsilon_1}\leq\g{1-\epsilon_1}$ on this yields
	\begin{align}
		\nonumber
			\g{1-\epsilon_1}
		=&
			\g{1+\epsilon_2}
		<
			\g{1+\epsilon_1}
		\leq
			\g{1-\epsilon_1}
	\end{align}
	which is a contradiction. Hence,
	$1-g_1\left(\g{1}+\epsilon\right)\leq g_2\left(\g{1}+\epsilon\right)-1$ and
	\begin{align}
		\label{Equation:Lemma:tuple_convex:delta1_equal}
			\delta_1\left(\epsilon\right)
		=
			\min\{1-g_1\left(\g{1}+\epsilon\right),
			g_2\left(\g{1}+\epsilon\right)-1\}
	\end{align}
	for all $\epsilon>0$.
	Due to the properties of $g_1,g_2$ one gets that
	$\delta_1\left(\epsilon\right)>0$ is well defined for all $\epsilon>0$.
	In order to show \refP{Property:Definition:tuple_nice:cont_around_1} let
	$x\in\left(0,\infty\right)$ and $\epsilon>0$ be such that
	$\abs{x-1}\leq\delta_1\left(\epsilon\right)$.
	This, together with \eqref{Equation:Lemma:tuple_convex:delta1_equal} yields
	\begin{align}
		\label{Equation:Lemma:tuple_convex:eq1}
			g_1\left(\g{1}+\epsilon\right)
		\leq&
			x
		\leq
			g_2\left(\g{1}+\epsilon\right).
	\end{align}
	If $x\leq 1$, then applying the monotonicity of $f$ to the left hand side of
	\eqref{Equation:Lemma:tuple_convex:eq1} yields $\g{1}+\epsilon\geq \g{x}$.
	If $x\geq 1$, then applying the monotonicity of $f$ to the right hand side of
	\eqref{Equation:Lemma:tuple_convex:eq1} yields $\g{x}\leq\g{1}+\epsilon$.
	In both cases, one gets $\abs{\g{x}-\g{1}}=\g{x}-\g{1}\leq\epsilon$
	so that \refP{Property:Definition:tuple_nice:cont_around_1} is fulfilled.
	It follows that the tuple $\left(\g{},g_1,g_2,\delta_1,\delta_2\right)$
	is sufficiently nice.
	\par
	In order to show \eqref{Equation:Lemma:tuple_convex:linear_scaling} let
	$\epsilon\in\left(0,\epsilon_0M\right]$.
	Note that due to \eqref{Equation:Lemma:tuple_convex:def:delta1delta2} one has
	\begin{align}
		\nonumber
			\delta_1\left(M^{-1}\delta_2\left(\epsilon\right)\right)
		=&
			1-g_1\left(\g{1}+M^{-1}\left(\g{1+\epsilon}-\g{1}\right)\right)
		=
			1-g_1\left(\left(1-M^{-1}\right)\g{1}+M^{-1}\g{1+\epsilon}\right).
	\end{align}
	Applying the convexity of $\g{}$ in $\left[1,\infty\right)$ and the monotonicity of
	$g_1$ to this yields
	\begin{align}
		\label{Equation:Lemma:tuple_convex:def_h}
			\delta_1\left(M^{-1}\delta_2\left(\epsilon\right)\right)
		\geq&
			1-g_1\left(\g{\left(1-M^{-1}\right)+M^{-1}\left(1+\epsilon\right)}\right)
		=
			1-g_1\left(\g{1+M^{-1}\epsilon}\right)
		=:
			h\left(\epsilon\right).
	\end{align}
	Applying basic arithmetic and that $g'$ is non-vanishing yields that $h$ is differentiable with
	\begin{align}
		\label{Equation:Lemma:tuple_convex:h_deriv_equal}
			h'\left(\epsilon\right)
		=&
			-g_1'\left(\g{1+M^{-1}\epsilon}\right)g'\left(1+M^{-1}\epsilon\right)M^{-1}
		=
			-M^{-1}\frac{g'\left(1+M^{-1}\epsilon\right)}{
				g'\left(g_1\left(\g{1+M^{-1}\epsilon}\right)\right)
			}.
	\end{align}
	Since $\g{}$ is convex on $\left[1,\infty\right)$, $g'$ is non-decreasing on $\left[1,\infty\right)$.
	Applying $\g{1+M^{-1}\epsilon}\leq\g{1-M^{-1}\epsilon}$ from the assumption with
	the monotonicity of $g_1$ and $g'$ to \eqref{Equation:Lemma:tuple_convex:h_deriv_equal}
	yields
	\begin{align}
		\label{Equation:Lemma:tuple_convex:h_deriv_inequal}
			h'\left(\epsilon\right)
		\geq&
			-M^{-1}\frac{g'\left(1+M^{-1}\epsilon\right)}{
				g'\left(g_1\left(\g{1-M^{-1}\epsilon}\right)\right)
			}
		=
			-M^{-1}\frac{g'\left(1+M^{-1}\epsilon\right)}{g'\left(1-M^{-1}\epsilon\right)}
	\end{align}
	for all $\epsilon\in\left(0,M\right)$.
	Using \eqref{Equation:Lemma:tuple_convex:def_h}, \eqref{Equation:Lemma:tuple_convex:h_deriv_inequal}
	and \refP{Property:Definition:tuple_convex:composition_derivative} from the assumption yields
	\begin{align}
		\nonumber
			\delta_1\left(M^{-1}\delta_2\left(\epsilon\right)\right)
		\geq&
			h\left(\epsilon\right)
		=
			h\left(0\right)+\int_{0}^\epsilon h'\left(x\right)dx
		=
			\int_{0}^\epsilon h'\left(x\right)dx
		\\\geq&\nonumber
			\int_{0}^\epsilon -M^{-1}\frac{g'\left(1+M^{-1}x\right)}{g'\left(1-M^{-1}x\right)}dx
		=
			\int_{0}^{M^{-1}\epsilon}
			-\frac{g'\left(1+x\right)}{g'\left(1-x\right)}dx
		\geq
			\nu\epsilon M^{-1}
	\end{align}
	for all $\epsilon\in\left(0,\epsilon_0 M\right]\subset\left(0,M\right)$.
\end{proof}
With this it is straightforward to prove that sufficiently convex tuples
also generate robust estimators.
\begin{Theorem}\label{Theorem:tuple_convex}
	Let the tuple $\left(\g{},g_1,g_2,\nu,\epsilon_0\right)$ be sufficiently convex
	and $\mathcal{H}\subset\mathbb{HPD}^M$ be closed in $\mathbb{HPD}^M$.
	For all $\XMat\in\mathcal{H}$, $0<\beta<\lambda_1\left(\XMat\right)$ and $\eta>0$
	there exists a function $\delta:\left(0,\infty\right)\rightarrow\left(0,\infty\right)$ such that
	the following holds true:
	For every $\epsilon>0$ and $\WMat\in\mathbb{HPD}^M$ with
	$\norm{\WMat-\XMat}_{2\rightarrow 2}\leq\delta\left(\epsilon\right)$,
	any minimizer $\ZMat$ of
	\begin{align}
		&\label{Equation:Theorem:tuple_convex:optimization}
			\min_{\ZMat\in\mathcal{H}}
			\sum_{m=1}^M\g{\lambda_m\left(\WMat^\frac{1}{2}\ZMat^{-1}\WMat^\frac{1}{2}\right)}
	\end{align}
	obeys $\norm{\XMat-\ZMat}_{2\rightarrow 2}\leq \epsilon$.
	In particular, $\delta$ can be chosen as
	\begin{align}
		\nonumber
			\delta_{c}\left(\epsilon\right)
		:=&
			\min\Bigg\{
				\frac{\nu M^{-1}}{2}
				\frac{\lambda_1\left(\XMat\right)}{\lambda_M\left(\XMat\right)}
				\left(
					\frac{\lambda_1\left(\XMat\right)-\beta}{\lambda_M\left(\XMat\right)+\beta}
				\right)
				\frac{g_1\left(\g{1}+\eta\right)}{g_2\left(\g{1}+\eta\right)}
				\epsilon,
				\frac{1}{2}\epsilon,
		\\&\nonumber
				\epsilon_0M\lambda_M\left(\XMat\right)
				\frac{g_2\left(\g{1}+\eta\right)}{g_1\left(\g{1}+\eta\right)}
				\left(
					\frac{\lambda_M\left(\XMat\right)+\beta}{\lambda_1\left(\XMat\right)-\beta}
				\right)^\frac{1}{2},
		\\&\nonumber
				\lambda_1\left(\XMat\right)
				\left(
					\frac{\lambda_1\left(\XMat\right)-\beta}{\lambda_M\left(\XMat\right)+\beta}
				\right)^\frac{1}{2}
				\left(1-g_1\left(\g{1}+\frac{\eta}{M}\right)\right),
		\\&\label{Equation:Theorem:tuple_convex:delta_def}
				\lambda_1\left(\XMat\right)
				\left(
					\frac{\lambda_1\left(\XMat\right)-\beta}{\lambda_M\left(\XMat\right)+\beta}
				\right)^\frac{1}{2}
				\left(1-g_1\left(\g{1+\epsilon_0}\right)\right),
				\beta
			\Bigg\}.
	\end{align}
\end{Theorem}
\begin{proof}
	Due to \thref{Lemma:tuple_convex} the tuple $\left(f,\delta_1,\delta_2,g_1,g_2\right)$
	with $\delta_1,\delta_2$ from \eqref{Equation:Lemma:tuple_convex:def:delta1delta2}
	is sufficiently nice so that one can apply \thref{Theorem:tuple_nice}.
	Let $\delta$ be from \eqref{Equation:Theorem:tuple_nice:delta_def}.
	It remains to show that $\delta_{c}\left(\epsilon\right)\leq\delta\left(\epsilon\right)$
	where $\delta_{c}$ is from \eqref{Equation:Theorem:tuple_convex:delta_def}
	since then applying \thref{Theorem:tuple_nice} yields the proof.
	\par
	Plugging in \eqref{Equation:Lemma:tuple_convex:def:delta1delta2} into
	\eqref{Equation:Theorem:tuple_nice:delta_def} yields
	\begin{align}
		\nonumber
			\delta\left(\epsilon\right)
		=&
			\min\Bigg\{
				\lambda_1\left(\XMat\right)
				\left(
					\frac{\lambda_1\left(\XMat\right)-\beta}{\lambda_M\left(\XMat\right)+\beta}
				\right)^\frac{1}{2}
				\delta_1\left(M^{-1}\delta_2\left(
					\frac{1}{2\lambda_M\left(\XMat\right)}
					\left(
						\frac{\lambda_1\left(\XMat\right)-\beta}{\lambda_M\left(\XMat\right)+\beta}
					\right)^\frac{1}{2}
					\frac{g_1\left(\g{1}+\eta\right)}{g_2\left(\g{1}+\eta\right)}
					\epsilon
				\right)\right),
				\frac{1}{2}\epsilon,
		\\&\label{Equation:Theorem:tuple_convex:eq1}
				\lambda_1\left(\XMat\right)
				\left(
					\frac{\lambda_1\left(\XMat\right)-\beta}{\lambda_M\left(\XMat\right)+\beta}
				\right)^\frac{1}{2}
				\left(1-g_1\left(\g{1}+\frac{\eta}{M}\right)\right),
				\beta
			\Bigg\}.
	\end{align}
	Now let
	\begin{align}
		\nonumber
			\epsilon
		\leq
			2\epsilon_0M\lambda_M\left(\XMat\right)
			\left(
				\frac{\lambda_M\left(\XMat\right)+\beta}{\lambda_1\left(\XMat\right)-\beta}
			\right)^\frac{1}{2}
			\frac{g_2\left(\g{1}+\eta\right)}{g_1\left(\g{1}+\eta\right)}.
	\end{align}
	It follows that
	\begin{align}
		\nonumber
			\frac{1}{2\lambda_M\left(\XMat\right)}
			\left(
				\frac{\lambda_1\left(\XMat\right)-\beta}{\lambda_M\left(\XMat\right)+\beta}
			\right)^\frac{1}{2}
			\frac{g_1\left(\g{1}+\eta\right)}{g_2\left(\g{1}+\eta\right)}
			\epsilon
		\leq
			M\epsilon_0.	
	\end{align}
	Thus, one can apply \eqref{Equation:Lemma:tuple_convex:linear_scaling}
	to \eqref{Equation:Theorem:tuple_convex:eq1} which yields
	\begin{align}
		\nonumber
			\delta\left(\epsilon\right)
		\geq&
			\min\Bigg\{
				\frac{\nu M^{-1}}{2}
				\frac{\lambda_1\left(\XMat\right)}{\lambda_M\left(\XMat\right)}
				\left(
					\frac{\lambda_1\left(\XMat\right)-\beta}{\lambda_M\left(\XMat\right)+\beta}
				\right)
				\frac{g_1\left(\g{1}+\eta\right)}{g_2\left(\g{1}+\eta\right)}
				\epsilon,
				\frac{1}{2}\epsilon,
		\\&\nonumber
				\lambda_1\left(\XMat\right)
				\left(
					\frac{\lambda_1\left(\XMat\right)-\beta}{\lambda_M\left(\XMat\right)+\beta}
				\right)^\frac{1}{2}
				\left(1-g_1\left(\g{1}+\frac{\eta}{M}\right)\right),
				\beta
			\Bigg\}
		\\\geq&\nonumber
			\delta_{c}\left(\epsilon\right).
	\end{align}
	On the other hand, assume
	\begin{align}
		\label{Equation:Theorem:tuple_convex:eq2}
			\epsilon
		>
			2\epsilon_0M\lambda_M\left(\XMat\right)
			\left(
				\frac{\lambda_M\left(\XMat\right)+\beta}{\lambda_1\left(\XMat\right)-\beta}
			\right)^\frac{1}{2}
			\frac{g_2\left(\g{1}+\eta\right)}{g_1\left(\g{1}+\eta\right)}.
	\end{align}
	Applying \eqref{Equation:Lemma:tuple_convex:def:delta1delta2} yields
	\begin{align}
		&\nonumber
			\delta_1\left(M^{-1}\delta_2\left(
				\frac{1}{2\lambda_M\left(\XMat\right)}
				\left(
					\frac{\lambda_1\left(\XMat\right)-\beta}{\lambda_M\left(\XMat\right)+\beta}
				\right)^\frac{1}{2}
				\frac{g_1\left(\g{1}+\eta\right)}{g_2\left(\g{1}+\eta\right)}
				\epsilon
			\right)\right)
		\\=&\nonumber
			\delta_1\left(M^{-1}\left(\g{1+
				\frac{1}{2\lambda_M\left(\XMat\right)}
				\left(
					\frac{\lambda_1\left(\XMat\right)-\beta}{\lambda_M\left(\XMat\right)+\beta}
				\right)^\frac{1}{2}
				\frac{g_1\left(\g{1}+\eta\right)}{g_2\left(\g{1}+\eta\right)}
				\epsilon
			}-\g{1}\right)\right)
		\\=&\nonumber
			1-g_1\left(\g{1}+M^{-1}\left(\g{1+
				\frac{1}{2\lambda_M\left(\XMat\right)}
				\left(
					\frac{\lambda_1\left(\XMat\right)-\beta}{\lambda_M\left(\XMat\right)+\beta}
				\right)^\frac{1}{2}
				\frac{g_1\left(\g{1}+\eta\right)}{g_2\left(\g{1}+\eta\right)}
				\epsilon
			}-\g{1}\right)\right)
		\\=&\nonumber
			1-g_1\left(\left(1-M^{-1}\right)\g{1}+M^{-1}\g{1+
				\frac{1}{2\lambda_M\left(\XMat\right)}
				\left(
					\frac{\lambda_1\left(\XMat\right)-\beta}{\lambda_M\left(\XMat\right)+\beta}
				\right)^\frac{1}{2}
				\frac{g_1\left(\g{1}+\eta\right)}{g_2\left(\g{1}+\eta\right)}
				\epsilon
			}\right).
	\end{align}
	Plugging \eqref{Equation:Theorem:tuple_convex:eq2} with
	the monotonicity of $\g{}$ in $\left[1,\infty\right)$
	and the monotonicity of $g_1$ into this gives
	\begin{align}
		&\nonumber
			\delta_1\left(M^{-1}\delta_2\left(
				\frac{1}{2\lambda_M\left(\XMat\right)}
				\left(
					\frac{\lambda_1\left(\XMat\right)-\beta}{\lambda_M\left(\XMat\right)+\beta}
				\right)^\frac{1}{2}
				\frac{g_1\left(\g{1}+\eta\right)}{g_2\left(\g{1}+\eta\right)}
				\epsilon
			\right)\right)
		\\>&\nonumber
			1-g_1\left(\left(1-M^{-1}\right)\g{1}+M^{-1}\g{1+\epsilon_0 M}\right).
	\end{align}
	Applying the convexity of $\g{}$ in $\left[1,\infty\right)$
	and the monotonicity of $g_1$ to this yields
	\begin{align}
		&\nonumber
			\delta_1\left(M^{-1}\delta_2\left(
				\frac{1}{2\lambda_M\left(\XMat\right)}
				\left(
					\frac{\lambda_1\left(\XMat\right)-\beta}{\lambda_M\left(\XMat\right)+\beta}
				\right)^\frac{1}{2}
				\frac{g_1\left(\g{1}+\eta\right)}{g_2\left(\g{1}+\eta\right)}
				\epsilon
			\right)\right)
		\\>&\nonumber
			1-g_1\left(\g{\left(1-M^{-1}\right)+M^{-1}\left(1+\epsilon_0 M\right)}\right)
		\\=&\nonumber
			1-g_1\left(\g{1+\epsilon_0}\right).
	\end{align}
	Plugging this and \eqref{Equation:Theorem:tuple_convex:eq2} into
	\eqref{Equation:Theorem:tuple_convex:eq1} results in
	\begin{align}
		\nonumber
			\delta\left(\epsilon\right)
		\geq&
			\min\Bigg\{
				\lambda_1\left(\XMat\right)
				\left(
					\frac{\lambda_1\left(\XMat\right)-\beta}{\lambda_M\left(\XMat\right)+\beta}
				\right)^\frac{1}{2}
				\left(1-g_1\left(\g{1+\epsilon_0}\right)\right),
				\frac{1}{2}\epsilon,
		\\&\nonumber
				\lambda_1\left(\XMat\right)
				\left(
					\frac{\lambda_1\left(\XMat\right)-\beta}{\lambda_M\left(\XMat\right)+\beta}
				\right)^\frac{1}{2}
				\left(1-g_1\left(\g{1}+\frac{\eta}{M}\right)\right),
				\beta
			\Bigg\}
		\\\geq&\nonumber
			\min\Bigg\{
				\lambda_1\left(\XMat\right)
				\left(
					\frac{\lambda_1\left(\XMat\right)-\beta}{\lambda_M\left(\XMat\right)+\beta}
				\right)^\frac{1}{2}
				\left(1-g_1\left(\g{1+\epsilon_0}\right)\right),
				\epsilon_0M\lambda_M\left(\XMat\right)
				\left(
					\frac{\lambda_M\left(\XMat\right)+\beta}{\lambda_1\left(\XMat\right)-\beta}
				\right)^\frac{1}{2}
				\frac{g_2\left(\g{1}+\eta\right)}{g_1\left(\g{1}+\eta\right)},
		\\&\nonumber
				\lambda_1\left(\XMat\right)
				\left(
					\frac{\lambda_1\left(\XMat\right)-\beta}{\lambda_M\left(\XMat\right)+\beta}
				\right)^\frac{1}{2}
				\left(1-g_1\left(\g{1}+\frac{\eta}{M}\right)\right),
				\beta
			\Bigg\}
		\\\geq&\nonumber
			\delta_{c}\left(\epsilon\right).
	\end{align}
	So in any case
	$\delta\left(\epsilon\right)\geq\delta_{c}\left(\epsilon\right)$
	which finishes the proof.
\end{proof}
The main advantage compared to \thref{Theorem:tuple_nice} is that $\delta_c$ is piecewise linear
and linear in a neighborhood around $0$ unlike $\delta$ from \thref{Theorem:tuple_nice}.
Thus, $\delta_c$ can be inverted for all $\epsilon$ small enough,
and its inverse is a linear function, here called
$\delta_c^{-1}\left(\cdot\right)=D\cdot$.
It follows that for all $\norm{\WMat-\XMat}_{2\rightarrow 2}$ small enough
one can choose $\epsilon:=\delta_c^{-1}\left(\norm{\XMat-\WMat}_{2\rightarrow 2}\right)$
and the estimation error satisfies
\begin{align}
	\label{Equation:tuple_convex:homogeinity}
		\norm{\XMat-\ZMat}_{2\rightarrow 2}
	\leq
		\delta_c^{-1}\left(\norm{\XMat-\WMat}_{2\rightarrow 2}\right)
	=
		D\norm{\XMat-\WMat}_{2\rightarrow 2}.
\end{align}
In this error bound the estimation error $\norm{\XMat-\ZMat}_{2\rightarrow 2}$
is linear in the magnitude of the perturbation $\norm{\XMat-\ZMat}_{2\rightarrow 2}$.
This property is similar to error bounds of other homogeneous estimators,
like for instance the result \cite[Theorem~1]{low_rank} for the non-negative least squares in
low-rank matrix estimation or the result \cite{sparse_psd} for the non-negative least squares with
measurement operators as in this work.
However, the homogeneity derived in \eqref{Equation:tuple_convex:homogeinity} only holds as long as the magnitude of the perturbation is
sufficiently small already.
	\section{Proof of Theorem \ref{Theorem:ad_errors}: Application to Activity Detection}
\noindent
\subsection{Trace-Log-Det Covariance Estimation}
In this subsection it is shown that the trace-log-det covariance estimator can be rewritten
as a covariance estimator of a sufficiently convex tuple.
To describe the inverse functions $g_1,g_2$, the two branches of the Lambert $W$ function given by
$W_{0}:\left[-\Exp{-1},\infty\right)\rightarrow\left[-1,\infty\right)$ and
$W_{-1}:\left[-\Exp{-1},0\right)\rightarrow\left(-\infty,-1\right]$
as introduced in \cite{lambertWfunction} are required.
These functions obey $W_i\left(y\right)\Exp{W_i\left(y\right)}=y$ for all
$y$ in their respective domains and are the inverse of the function $x\mapsto x\Exp{x}$
on the corresponding domain of definition respectively.
Further, $W_{0}\left(0\right)=0$.
Due to
$
		\frac{\Ln{4}}{4}
	=
		\frac{\Ln{2^2}}{4}
	=
		\frac{\Ln{2}}{2}
$
one gets
\begin{align}
	\nonumber
		t
	:=
		\left(-\Ln{2}\right)\Exp{-\Ln{2}}
	=
		-\frac{\Ln{2}}{2}
	=
		-\frac{\Ln{4}}{4}
	=
		\left(-\Ln{4}\right)\Exp{-\Ln{4}}.
\end{align}
Hence, $-\Ln{4}$ and $-\Ln{2}$ are the two solutions of $s\Exp{s}=t$.
Since $-\Ln{4}\leq -1$, it must correspond to the branch $W_{-1}$, and it follows that
\begin{align}
	\label{Equation:W_0:value}
		W_0\left(-\frac{\Ln{4}}{4}\right)
	=
		W_0\left(t\right)
	=
		-\Ln{2}.
\end{align}
At first, it is established that the trace-log-det covariance estimator is generated by a sufficiently convex tuple.
\begin{Lemma}\label{Lemma:trace_logdet_tuple}
	Let $\g{x}:=x-\Ln{x}$ for all $x\in\left(0,\infty\right)$,
	\begin{align}
		\label{Equation:Lemma:trace_logdet_tuple:g1g2}
			g_1\left(y\right)
		=
			-W_{0}\left(-\Exp{-y}\right)
		\TextAnd
			g_2\left(y\right)
		=
			-W_{-1}\left(-\Exp{-y}\right)
		\TextForAll
			y\in\left[1,\infty\right),
	\end{align}
	$\nu:=\frac{1-\Ln{2}}{\Ln{2}}$
	and
	$\epsilon_0:=\Ln{4}-1$.
	Then, the tuple $\left(\g{},g_1,g_2,\nu,\epsilon_0\right)$ is sufficiently convex
	and
	\begin{align}
		\label{Equation:Lemma:trace_logdet_tuple:trace_logdet=sumg}
			\Trace{\ZMat^{-1}\WMat}
			+\Ln{\Det{\ZMat}}
		=
			\sum_{m=1}^M
			\g{\lambda_m\left(
				\WMat^\frac{1}{2}
				\ZMat^{-1}
				\WMat^\frac{1}{2}
			\right)}
			+\Ln{\Det{\WMat}}
	\end{align}
	for all $\WMat,\ZMat\in\mathbb{HPD}^M$.
\end{Lemma}
\begin{proof}
	It is clear that \refP{Property:Definition:tuple_convex:grow}
	and \refP{Property:Definition:tuple_convex:cont} are fulfilled.
	By differentiation $\g{}$ is strictly monotonically decreasing in $\left(0,1\right]$
	and strictly monotonically increasing in $\left[1,\infty\right)$
	and hence is invertible in each of those intervals.
	It can be validated that $g_1,g_2$ are the inverses of $g$ in those intervals
	by putting in the properties of the Lambert
	$W$ function and considering the appropriate domain of definitions.
	Thus, \refP{Property:Definition:tuple_convex:monotonic_g1} and
	\refP{Property:Definition:tuple_convex:monotonic_g2} are fulfilled.
	\par
	Now consider the function
	$h(\epsilon):=-2\epsilon-\Ln{1-\epsilon}+\Ln{1+\epsilon}$
	which obeys
	$h\left(0\right)=0$ and
	$h'\left(\epsilon\right)=\frac{2\epsilon^2}{\left(1-\epsilon\right)\left(1+\epsilon\right)}\geq 0$
	for all $\epsilon\in\left[0,1\right)$ due to which $h\left(\epsilon\right)\geq 0$ for all
	$\epsilon\in\left[0,1\right)$.
	This is equivalent to $\g{1-\epsilon}\geq\g{1+\epsilon}$
	and yields \refP{Property:Definition:tuple_convex:inver}.
	\par
	By differentiating twice, it follows that $\g{}$ is convex on $\left(0,\infty\right)$
	and \refP{Property:Definition:tuple_convex:convex} is fulfilled.
	At last note that
	\begin{align}
		\nonumber
			-\frac{g'\left(1+\epsilon\right)}{g'\left(1-\epsilon\right)}
		=&
			-\frac{1-\frac{1}{1+\epsilon}}{1-\frac{1}{1-\epsilon}}
		=
			\frac{1-\epsilon}{1+\epsilon}.
	\end{align}
	From this and
	$\frac{\Ln{4}}{4}=\frac{\Ln{2^2}}{4}=\frac{\Ln{2}}{2}$
	it follows that
	$-\frac{g'\left(1+\epsilon\right)}{g'\left(1-\epsilon\right)}\geq\frac{2-\Ln{4}}{\Ln{4}}
		=\frac{1-\Ln{2}}{\Ln{2}}=\nu>0$
	for all $\epsilon\leq\Ln{4}-1=\epsilon_0\in\left(0,1\right)$, and
	\refP{Property:Definition:tuple_convex:composition_derivative}
	is fulfilled.
	By applying the definition of $g$ one gets
	\begin{align}
		&\nonumber
			\Trace{\ZMat^{-1}\WMat}
			+\Ln{\Det{\ZMat}}
		\\=&\nonumber
			\Trace{
				\WMat^\frac{1}{2}
				\ZMat^{-1}
				\WMat^\frac{1}{2}
			}
			-\Ln{\Det{
				\WMat^\frac{1}{2}
				\ZMat^{-1}
				\WMat^\frac{1}{2}
			}}
			+\Ln{\Det{\WMat}}
		\\=&\nonumber
			\sum_{m=1}^M
			\lambda_m\left(
				\WMat^\frac{1}{2}
				\ZMat^{-1}
				\WMat^\frac{1}{2}
			\right)
			-\Ln{
				\prod_{m=1}^M
				\lambda_m\left(
					\WMat^\frac{1}{2}
					\ZMat^{-1}
					\WMat^\frac{1}{2}
				\right)
			}
			+\Ln{\Det{\WMat}}
		\\=&\nonumber
			\sum_{m=1}^M
			\lambda_m\left(
				\WMat^\frac{1}{2}
				\ZMat^{-1}
				\WMat^\frac{1}{2}
			\right)
			-\sum_{m=1}^M\Ln{
				\lambda_m\left(
					\WMat^\frac{1}{2}
					\ZMat^{-1}
					\WMat^\frac{1}{2}
				\right)
			}
			+\Ln{\Det{\WMat}}
		\\=&\nonumber
			\sum_{m=1}^M
			\g{\lambda_m\left(
				\WMat^\frac{1}{2}
				\ZMat^{-1}
				\WMat^\frac{1}{2}
			\right)}
			+\Ln{\Det{\WMat}}.
	\end{align}
\end{proof}
Now the robustness of the trace-log-det covariance estimator can be shown.
\begin{Theorem}\label{Theorem:trace_logdet}
	Let $\mathcal{H}\subset\mathbb{HPD}^M$ be closed in $\mathbb{HPD}^M$,
	$0<\beta<\lambda_1\left(\XMat\right)$ and $\eta>0$.
	Then, there exists a function $\delta:\left(0,\infty\right)\rightarrow\left(0,\infty\right)$ such that
	the following holds true:
	For every $\epsilon>0$ and $\WMat\in\mathbb{HPD}^M$ with
	$\norm{\WMat-\XMat}_{2\rightarrow 2}\leq\delta\left(\epsilon\right)$,
	any minimizer $\ZMat$ of
	\begin{align}
		&\label{Equation:Theorem:trace_logdet:optimization}
			\min_{\ZMat\in\mathcal{H}}
			\Trace{\ZMat^{-1}\WMat}
			+\Ln{\Det{\ZMat}}
	\end{align}
	obeys $\norm{\XMat-\ZMat}_{2\rightarrow 2}\leq \epsilon$.
	In particular, $\delta$ can be chosen as
	\begin{align}
		\nonumber
			\delta_{tld}\left(\epsilon\right)
		:=&
			\min\Bigg\{
				\frac{\left(1-\Ln{2}\right)M^{-1}}{2\Ln{2}}
				\frac{\lambda_1\left(\XMat\right)}{\lambda_M\left(\XMat\right)}
				\left(
					\frac{\lambda_1\left(\XMat\right)-\beta}{\lambda_M\left(\XMat\right)+\beta}
				\right)
				\frac{W_{0}\left(-\Exp{-\left(1+\eta\right)}\right)}{
					W_{-1}\left(-\Exp{-\left(1+\eta\right)}\right)
				}
				\epsilon,
				\frac{1}{2}\epsilon,
		\\&\nonumber
				\left(\Ln{4}-1\right)M\lambda_M\left(\XMat\right)
				\left(
					\frac{\lambda_M\left(\XMat\right)+\beta}{\lambda_1\left(\XMat\right)-\beta}
				\right)^\frac{1}{2}
				\frac{W_{-1}\left(-\Exp{-\left(1+\eta\right)}\right)}{
					W_{0}\left(-\Exp{-\left(1+\eta\right)}\right)
				},
		\\&\nonumber
				\lambda_1\left(\XMat\right)
				\left(
					\frac{\lambda_1\left(\XMat\right)-\beta}{\lambda_M\left(\XMat\right)+\beta}
				\right)^\frac{1}{2}
				\left(1+W_{0}\left(-\Exp{-\left(1+\frac{\eta}{M}\right)}\right)\right),
		\\&\nonumber
				\left(1-\Ln{2}\right)\lambda_1\left(\XMat\right)
				\left(
					\frac{\lambda_1\left(\XMat\right)-\beta}{\lambda_M\left(\XMat\right)+\beta}
				\right)^\frac{1}{2},
				\beta
			\Bigg\}.
	\end{align}
\end{Theorem}
\begin{proof}
	Let $\left(\g{},g_1,g_2,\nu,\epsilon_0\right)$ be the sufficiently convex tuple
	from \thref{Lemma:trace_logdet_tuple}.
	By \eqref{Equation:Lemma:trace_logdet_tuple:trace_logdet=sumg}
	the problem \eqref{Equation:Theorem:trace_logdet:optimization} has the same minimizers as
	\eqref{Equation:Theorem:tuple_convex:optimization}.
	Applying \thref{Theorem:tuple_convex} would yield the claim if
	one can show $\delta_{tld}\left(\epsilon\right)=\delta_{c}\left(\epsilon\right)$,
	which will be done next. To do this, \eqref{Equation:Lemma:trace_logdet_tuple:g1g2}
	and the definition of $\nu$ and $\epsilon_0$
	are plugged into \eqref{Equation:Theorem:tuple_convex:delta_def}, which yields
	\begin{align}
		\nonumber
			\delta_{c}\left(\epsilon\right)
		=&
			\min\Bigg\{
				\frac{\nu M^{-1}}{2}
				\frac{\lambda_1\left(\XMat\right)}{\lambda_M\left(\XMat\right)}
				\left(
					\frac{\lambda_1\left(\XMat\right)-\beta}{\lambda_M\left(\XMat\right)+\beta}
				\right)
				\frac{g_1\left(\g{1}+\eta\right)}{g_2\left(\g{1}+\eta\right)}
				\epsilon,
				\frac{1}{2}\epsilon,
		\\&\nonumber
				\epsilon_0M\lambda_M\left(\XMat\right)
				\left(
					\frac{\lambda_M\left(\XMat\right)+\beta}{\lambda_1\left(\XMat\right)-\beta}
				\right)^\frac{1}{2}
				\frac{g_2\left(\g{1}+\eta\right)}{g_1\left(\g{1}+\eta\right)},
		\\&\nonumber
				\lambda_1\left(\XMat\right)
				\left(
					\frac{\lambda_1\left(\XMat\right)-\beta}{\lambda_M\left(\XMat\right)+\beta}
				\right)^\frac{1}{2}
				\left(1-g_1\left(\g{1}+\frac{\eta}{M}\right)\right),
		\\&\nonumber
				\lambda_1\left(\XMat\right)
				\left(
					\frac{\lambda_1\left(\XMat\right)-\beta}{\lambda_M\left(\XMat\right)+\beta}
				\right)^\frac{1}{2}
				\left(1-g_1\left(\g{1+\epsilon_0}\right)\right),
				\beta
			\Bigg\}
		\\=&\nonumber
			\min\Bigg\{
				\frac{\left(1-\Ln{2}\right)M^{-1}}{2\Ln{2}}
				\frac{\lambda_1\left(\XMat\right)}{\lambda_M\left(\XMat\right)}
				\left(
					\frac{\lambda_1\left(\XMat\right)-\beta}{\lambda_M\left(\XMat\right)+\beta}
				\right)
				\frac{W_{0}\left(-\Exp{-\left(\g{1}+\eta\right)}\right)}{
					W_{-1}\left(-\Exp{-\left(\g{1}+\eta\right)}\right)
				}
				\epsilon,
				\frac{1}{2}\epsilon,
		\\&\nonumber
				\left(\Ln{4}-1\right)M\lambda_M\left(\XMat\right)
				\left(
					\frac{\lambda_M\left(\XMat\right)+\beta}{\lambda_1\left(\XMat\right)-\beta}
				\right)^\frac{1}{2}
				\frac{W_{-1}\left(-\Exp{-\left(\g{1}+\eta\right)}\right)}{
					W_{0}\left(-\Exp{-\left(\g{1}+\eta\right)}\right)
				},
		\\&\nonumber
				\lambda_1\left(\XMat\right)
				\left(
					\frac{\lambda_1\left(\XMat\right)-\beta}{\lambda_M\left(\XMat\right)+\beta}
				\right)^\frac{1}{2}
				\left(1+W_{0}\left(-\Exp{-\left(\g{1}+\frac{\eta}{M}\right)}\right)\right),
		\\&\nonumber
				\lambda_1\left(\XMat\right)
				\left(
					\frac{\lambda_1\left(\XMat\right)-\beta}{\lambda_M\left(\XMat\right)+\beta}
				\right)^\frac{1}{2}
				\left(
					1+W_{0}\left(-\Exp{-\g{\Ln{4}}}\right)
				\right),
				\beta
			\Bigg\}.
	\end{align}
	Using the definition of $\g{}$ in this yields
	\begin{align}
		\nonumber
			\delta_{c}\left(\epsilon\right)
		=&
			\min\Bigg\{
				\frac{\left(1-\Ln{2}\right)M^{-1}}{2\Ln{2}}
				\frac{\lambda_1\left(\XMat\right)}{\lambda_M\left(\XMat\right)}
				\left(
					\frac{\lambda_1\left(\XMat\right)-\beta}{\lambda_M\left(\XMat\right)+\beta}
				\right)
				\frac{W_{0}\left(-\Exp{-\left(1+\eta\right)}\right)}{
					W_{-1}\left(-\Exp{-\left(1+\eta\right)}\right)
				}
				\epsilon,
				\frac{1}{2}\epsilon,
		\\&\nonumber
				\left(\Ln{4}-1\right)M\lambda_M\left(\XMat\right)
				\left(
					\frac{\lambda_M\left(\XMat\right)+\beta}{\lambda_1\left(\XMat\right)-\beta}
				\right)^\frac{1}{2}
				\frac{W_{-1}\left(-\Exp{-\left(1+\eta\right)}\right)}{
					W_{0}\left(-\Exp{-\left(1+\eta\right)}\right)
				},
		\\&\nonumber
				\lambda_1\left(\XMat\right)
				\left(
					\frac{\lambda_1\left(\XMat\right)-\beta}{\lambda_M\left(\XMat\right)+\beta}
				\right)^\frac{1}{2}
				\left(1+W_{0}\left(-\Exp{-\left(1+\frac{\eta}{M}\right)}\right)\right),
		\\&\nonumber
				\lambda_1\left(\XMat\right)
				\left(
					\frac{\lambda_1\left(\XMat\right)-\beta}{\lambda_M\left(\XMat\right)+\beta}
				\right)^\frac{1}{2}
				\left(
					1+W_{0}\left(-\frac{\Ln{4}}{4}\right)
				\right),
				\beta
			\Bigg\}.
	\end{align}
	Applying \eqref{Equation:W_0:value} to this results in $\delta_{c}=\delta_{tld}$.
\end{proof}
	\subsection{Signed Kernel Condition}
\label{Subsection:proof_signed_kernel_condition}
The next part of the proof of \thref{Theorem:ad_errors} is based on
the signed kernel condition from \thref{Definition:skc} and the robustness constant from
\cite[Definition~2.1]{NNLR}.
\begin{Definition}
	\label{Definition:robustness_constant}
	Let $\mathcal{A}:\mathbb{C}^N\rightarrow\mathbb{C}^{M\times M}$ be a linear operator,
	$S\in\mathbb{N}$ and $\norm{\cdot}$ a norm on $\mathbb{C}^{M\times M}$.
	The constant
	\begin{align}
		\nonumber
			\tau\left(\mathcal{A}\right)
		:=
			\inf_{\zVec\in \mathbb{R}^N_+,\xVec\in \Sigma_S^N\cap\mathbb{R}_+^N\zVec\neq\xVec}
			\frac{
				\norm{\mathcal{A}\left(\zVec-\xVec\right)}
			}{
				\norm{\zVec-\xVec}_{2}
			}
	\end{align}
	is called robustness constant.
\end{Definition}
The general norm $\norm{\cdot}$ appearing in this definition is due to the
later explained relation of the robustness constant with the
non-negative least residual estimator which is any minimizer of
\begin{align}
	\nonumber
		\argmin{\zVec\in\mathbb{R}_+^N}\norm{\mathcal{A}\left(\zVec\right)-\WMat'}.
\end{align}
So the non-negative least residual is the non-negative least squares with
the $\ell_2$-norm being replaced by the general norm $\norm{\cdot}$.
The constant $\tau\left(\mathcal{A}\right)$ depends on the choice of the norm
$\norm{\cdot}$ as well as the sparsity $S$. Its dependence is omitted for the sake of brevity.
Due to the following result from \cite[Theorem~3.2, Theorem~2.2, Proposition~2.8]{NNLR}
an signed kernel condition is the equivalent condition for robust recovery with the
non-negative least residual, and the robustness constant gives a relation
between the estimation error and the magnitude of the perturbation.
\begin{Theorem}\label{Theorem:skc_properties}
	Let $\mathcal{A}:\mathbb{C}^N\rightarrow\mathbb{C}^{M\times M}$
	have signed kernel condition of order $S$ and let $\norm{\cdot}$ be a norm on
	$\mathbb{C}^{M\times M}$. Then, $\tau\left(\mathcal{A}\right)>0$
	and
	\begin{align}
		\label{Equation:Theorem:skc_properties:nnlr_robustness}
			\norm{\xVec-\zVec}_2
		\leq
			\frac{2}{\tau\left(\mathcal{A}\right)}\norm{\WMat'-\mathcal{A}\left(\xVec\right)}
		\TextForAll
			\xVec\in\Sigma_S^N\cap\mathbb{R}_+^N,\WMat'\in\mathbb{C}^{M\times M},
			\textnormal{$\zVec$ minimizer of $\min_{\zVec\in\mathbb{R}_+^N}\norm{\mathcal{A}\left(\zVec\right)-\WMat'}$}
	\end{align}
	holds true.
\end{Theorem}
\begin{proof}
	By \cite[Theorem~3.2]{NNLR} one has
	\begin{align}
		\nonumber
			\{\xVec\}
		=
			\min_{\zVec\in\mathbb{R}^N_+}\norm{\mathcal{A}\left(\zVec\right)-\mathcal{A}\left(\xVec\right)}
		\TextForAll
			\xVec\in\Sigma_S^N\cap\mathbb{R}^N_+.
	\end{align}
	By \cite[Theorem~2.2]{NNLR} together with \cite[Proposition~2.8]{NNLR}
	it follows that
	\eqref{Equation:Theorem:skc_properties:nnlr_robustness} holds true.
\end{proof}
It should be noted that a linear operator can have the signed kernel condition of order
$S$ and a non-zero but very small robustness constant $\tau\left(\mathcal{A}\right)$.
This can cause the estimation error to be very large even for very small perturbations.
In implementations observations without perturbation can have small perturbations
due to machine precision, or estimators will solve optimization problems only up to
a predefined precision effectively causing a small perturbation.
This together with a small robustness constant
can cause recovery to seemingly fail in implementations
even if no perturbation is present.
The construction in \eqref{Equation:Theorem:sample_rate:defA} exactly has these problems
and is thus only of theoretical value. For implementations
constructions with better robustness constants are required.
\par
Combining this robustness result with the robustness of the
trace-log-det covariance estimator
yields a robustness result for the relaxed maximum likelihood estimator.
\begin{Theorem}\label{Theorem:skc}
	Let $\mathcal{A}:\mathbb{C}^N\rightarrow\mathbb{C}^{M\times M}$
	have signed kernel condition of order $S$.
	For all $\SigmaMat\in\mathbb{HPD}^M$, $\xVec\in\Sigma_S^N\cap\mathbb{R}_+^N$,
	$0<\beta<\lambda_1\left(\mathcal{A}\left(\xVec\right)+\SigmaMat\right)$, $\eta>0$ with
	$\mathcal{H}:=\left\{\mathcal{A}\left(\zVec\right)+\SigmaMat:\zVec\in\mathbb{R}_+^N\right\}
		\subset\mathbb{HPD}^M$
	there exists a function $\delta:\left(0,\infty\right)\rightarrow\left(0,\infty\right)$ such that
	the following holds true:
	For every $\epsilon>0$ and $\WMat\in\mathbb{HPD}^M$ with
	$\norm{\WMat-\mathcal{A}\left(\xVec\right)-\SigmaMat}_{2\rightarrow 2}$
	$\leq\delta\left(\epsilon\right)$,
	any minimizer $\zVec$ of
	\begin{align}
		&\label{Equation:Theorem:skc:optimization}
			\min_{\zVec\in\mathbb{R}_+^N}
			\Trace{\left(\mathcal{A}\left(\zVec\right)+\SigmaMat\right)^{-1}\WMat}
			+\Ln{\Det{\mathcal{A}\left(\zVec\right)+\SigmaMat}}
	\end{align}
	obeys $\norm{\xVec-\zVec}_2\leq \epsilon$.
	In particular, $\delta$ can be chosen as
	\begin{align}
		\nonumber
			\delta_{skc}\left(\epsilon\right)
		:=&
			\min\Bigg\{
				\frac{\left(1-\Ln{2}\right)M^{-1}\tau\left(\mathcal{A}\right)}{\Ln{2}}
				\frac{\lambda_1\left(\mathcal{A}\left(\xVec\right)+\SigmaMat\right)}{\lambda_M\left(\mathcal{A}\left(\xVec\right)+\SigmaMat\right)}
				\left(
					\frac{\lambda_1\left(\mathcal{A}\left(\xVec\right)+\SigmaMat\right)-\beta}{\lambda_M\left(\mathcal{A}\left(\xVec\right)+\SigmaMat\right)+\beta}
				\right)
				\frac{W_{0}\left(-\Exp{-\left(1+\eta\right)}\right)}{
					W_{-1}\left(-\Exp{-\left(1+\eta\right)}\right)
				}
				\epsilon,		
		\\&\nonumber
				\tau\left(\mathcal{A}\right)\epsilon,
				\left(\Ln{4}-1\right)M\lambda_M\left(\mathcal{A}\left(\xVec\right)+\SigmaMat\right)
				\left(
					\frac{\lambda_M\left(\mathcal{A}\left(\xVec\right)+\SigmaMat\right)+\beta}{\lambda_1\left(\mathcal{A}\left(\xVec\right)+\SigmaMat\right)-\beta}
				\right)^\frac{1}{2}
				\frac{W_{-1}\left(-\Exp{-\left(1+\eta\right)}\right)}{
					W_{0}\left(-\Exp{-\left(1+\eta\right)}\right)
				},
		\\&\nonumber
				\lambda_1\left(\mathcal{A}\left(\xVec\right)+\SigmaMat\right)
				\left(
					\frac{\lambda_1\left(\mathcal{A}\left(\xVec\right)+\SigmaMat\right)-\beta}{\lambda_M\left(\mathcal{A}\left(\xVec\right)+\SigmaMat\right)+\beta}
				\right)^\frac{1}{2}
				\left(1+W_{0}\left(-\Exp{-\left(1+\frac{\eta}{M}\right)}\right)\right),
		\\&\label{Equation:Theorem:skc:delta_def}
				\left(1-\Ln{2}\right)\lambda_1\left(\mathcal{A}\left(\xVec\right)+\SigmaMat\right)
				\left(
					\frac{\lambda_1\left(\mathcal{A}\left(\xVec\right)+\SigmaMat\right)-\beta}{\lambda_M\left(\mathcal{A}\left(\xVec\right)+\SigmaMat\right)+\beta}
				\right)^\frac{1}{2},
				\beta
			\Bigg\}.
	\end{align}
\end{Theorem}
\begin{proof}
	Choose the norm $\norm{\cdot}:=\norm{\cdot}_{2\rightarrow 2}$ for the robustness constant.
	Note that $\mathcal{H}$ is closed since $\mathcal{A}$ is linear and thus continuous on
	the finite-dimensional space $\mathbb{C}^N$.
	Set $\XMat:=\mathcal{A}\left(\xVec\right)+\SigmaMat$ and let $\delta_{tld}$ be from
	\thref{Theorem:trace_logdet} so that
	$\delta_{skc}\left(\epsilon\right)=\delta_{tld}\left(\frac{\tau\left(\mathcal{A}\right)}{2}\epsilon\right)$.
	Now let $\WMat$ be such that
	$\norm{\WMat-\mathcal{A}\left(\xVec\right)-\SigmaMat}_{2\rightarrow 2}
		\leq\delta_{skc}\left(\epsilon\right)$.
	It follows
	$\norm{\WMat-\XMat}_{2\rightarrow 2}\leq\delta_{skc}\left(\epsilon\right)
		=\delta_{tld}\left(\frac{\tau\left(\mathcal{A}\right)}{2}\epsilon\right)$
	and $0<\beta<\lambda_1\left(\XMat\right)$.
	If $\zVec$ is an optimizer of \eqref{Equation:Theorem:skc:optimization} then
	$\ZMat:=\mathcal{A}\left(\zVec\right)+\SigmaMat$ is an optimizer of
	\eqref{Equation:Theorem:trace_logdet:optimization}. Thus, \thref{Theorem:trace_logdet}
	yields that
	\begin{align}
		\label{Equation:Theorem:skc:eq1}
			\norm{\ZMat-\XMat}_{2\rightarrow 2}
		\leq
			\frac{\tau\left(\mathcal{A}\right)}{2}\epsilon.
	\end{align}
	Now set $\WMat':=\ZMat-\SigmaMat=\mathcal{A}\left(\zVec\right)$ so that
	$\zVec$ is an optimizer of $\min_{\zVec\in\mathbb{R}_+^N}\norm{\mathcal{A}\left(\zVec\right)-\WMat'}_{2\rightarrow 2}$.
	By \thref{Theorem:skc_properties} one gets
	\begin{align}
		\nonumber
			\norm{\zVec-\xVec}_2
		\leq&
			\frac{2}{\tau\left(\mathcal{A}\right)}\norm{\WMat'-\mathcal{A}\left(\xVec\right)}_{2\rightarrow 2}
		=
			\frac{2}{\tau\left(\mathcal{A}\right)}\norm{\ZMat-\XMat}_{2\rightarrow 2}.
	\end{align}
	Applying \eqref{Equation:Theorem:skc:eq1} to this yields the claim.
\end{proof}
To the best of the authors' knowledge, this is the first robustness result for
the relaxed maximum likelihood estimator.
Note that \eqref{Equation:Theorem:skc_properties:nnlr_robustness}
is again a robustness result, and indeed one could choose
$\delta\left(\cdot\right):=\frac{2}{\tau\left(\mathcal{A}\right)}\cdot$.
The robustness derived in \thref{Theorem:skc} is weaker, since
$\delta_{skc}$ depends negatively on the
dimension $M$. Due to this, the non-negative least residual has better robustness properties
and better recovery guarantees. An interesting question would be whether \thref{Theorem:skc}
can be improved by removing all dimensional scaling parameters
to make $\delta_{skc}$ scale like $\frac{1}{\tau\left(\mathcal{A}\right)}$
and thus make the result as good as the recovery guarantee for the
non-negative least residual.
This is left for future investigation.
\par
Now it is shown that the signed kernel condition is also an equivalent condition for successful recovery
with the relaxed maximum likelihood estimator
in the infinite antenna case.
This result is not required for the proof of \thref{Theorem:ad_errors} but is of independent interest.
\begin{Theorem}\label{Theorem:skc_recovery}
	Let $\SigmaMat\in\mathbb{HPD}^M$ and
	$\mathcal{A}:\mathbb{C}^N\rightarrow\mathbb{C}^{M\times M}$
	be a linear operator such that
	$\mathcal{H}:=\big\{\mathcal{A}\left(\zVec\right)+\SigmaMat:\zVec\in\mathbb{R}_+^N\big\}
		\subset\mathbb{HPD}^M$.
	Then, the following are equivalent.
	\begin{enumerate}
		\item
			$\mathcal{A}$ has the signed kernel condition of order $S$.
		\item
			For all $\xVec\in\Sigma_S^N\cap\mathbb{R}_+^N$
			the problem
			\begin{align}
				&\label{Equation:Theorem:skc_recovery:optimization}
					\min_{\zVec\in\mathbb{R}_+^N}
					\Trace{\left(\mathcal{A}\left(\zVec\right)+\SigmaMat\right)^{-1}\left(\mathcal{A}\left(\xVec\right)+\SigmaMat\right)}
					+\Ln{\Det{\mathcal{A}\left(\zVec\right)+\SigmaMat}}
			\end{align}
			has a unique minimizer and it is $\xVec$.
	\end{enumerate}
\end{Theorem}
\begin{proof}
	By \eqref{Equation:Lemma:trace_logdet_tuple:trace_logdet=sumg} from
	\thref{Lemma:trace_logdet_tuple} the minimizers of
	\eqref{Equation:Theorem:skc_recovery:optimization} are exactly the minimizers of
	\begin{align}
		\label{Equation:Theorem:skc_recovery:optimization_g}
			\min_{\zVec\in\mathbb{R}_+^N}
			\sum_{m=1}^M\g{\lambda_m\left(\left(\mathcal{A}\left(\xVec\right)+\SigmaMat\right)^\frac{1}{2}\left(\mathcal{A}\left(\zVec\right)+\SigmaMat\right)^{-1}\left(\mathcal{A}\left(\xVec\right)+\SigmaMat\right)^\frac{1}{2}\right)}
	\end{align}
	for some $g$ of a sufficiently convex tuple.
	Since $g$ is part of a sufficiently convex tuple,
	$1$ is the unique minimizer of $\min_{x\in\left(0,\infty\right)}g(x)$.
	It follows that $\zVec$ is a minimizer of \eqref{Equation:Theorem:skc_recovery:optimization_g}
	if and only if
	\begin{align}
		\nonumber
			\g{\lambda_m\left(\left(\mathcal{A}\left(\xVec\right)+\SigmaMat\right)^\frac{1}{2}\left(\mathcal{A}\left(\zVec\right)+\SigmaMat\right)^{-1}\left(\mathcal{A}\left(\xVec\right)+\SigmaMat\right)^\frac{1}{2}\right)}
		=
			\g{1}
	\end{align}
	for all $m\in\SetOf{M}$.
	This, on the other hand, is equivalent to
	$\mathcal{A}\left(\zVec\right)+\SigmaMat=\mathcal{A}\left(\xVec\right)+\SigmaMat$.
	It follows that
	\begin{align}
		\label{Equation:Theorem:skc_recovery:kernel}
			\textnormal{
				$\zVec$ is a minimizer of \eqref{Equation:Theorem:skc_recovery:optimization}
				if and only if $\zVec-\xVec\in\Kernel{\mathcal{A}}$
			}
	\end{align}
	holds true.
	Applying \cite[Theorem~3.2]{NNLR} and \cite[Theorem~2.2]{NNLR}
	yields that
	$\mathcal{A}$ having the signed kernel condition of order $S$
	is equivalent to
	\begin{align}
		\label{Equation:Theorem:skc_recovery:kernel_recovery}
			\left(\mathbb{R}_+^N-\Sigma_S^N\cap\mathbb{R}_+^N\right)\cap\Kernel{\mathcal{A}}=\ZeroSet.
	\end{align}
	It is now shown that the latter is equivalent to the second condition of this theorem.
	\par
	Assume that \eqref{Equation:Theorem:skc_recovery:kernel_recovery} holds true.
	Let $\xVec\in\Sigma_S^N\cap\mathbb{R}_+^N$ and $\zVec$ be a minimizer of
	\eqref{Equation:Theorem:skc_recovery:optimization}.
	Then, $\zVec-\xVec\in\mathbb{R}_+^N-\Sigma_S^N\cap\mathbb{R}_+^N$
	and by \eqref{Equation:Theorem:skc_recovery:kernel} $\zVec-\xVec\in\Kernel{\mathcal{A}}$.
	By \eqref{Equation:Theorem:skc_recovery:kernel_recovery} $\zVec=\xVec$
	and $\xVec$ is the unique minimizer of
	\eqref{Equation:Theorem:skc_recovery:optimization}.
	On the other hand, assume that
	for all $\xVec\in\Sigma_S^N\cap\mathbb{R}_+^N$ the problem
	\eqref{Equation:Theorem:skc_recovery:optimization} has the unique minimizer $\xVec$.
	To prove the converse implication, let $\vVec\in\left(\mathbb{R}_+^N-\Sigma_S^N\cap\mathbb{R}_+^N\right)\cap\Kernel{\mathcal{A}}$.
	Then, there exist
	$\zVec\in\mathbb{R}_+^N$ and $\xVec\in\Sigma_S^N\cap\mathbb{R}_+^N$ so that
	$\zVec-\xVec=\vVec\in\Kernel{\mathcal{A}}$.
	By \eqref{Equation:Theorem:skc_recovery:kernel} $\zVec$ is a minimizer of 
	\eqref{Equation:Theorem:skc_recovery:optimization}. By assumption $\zVec=\xVec$
	so that \eqref{Equation:Theorem:skc_recovery:kernel_recovery} holds true.
\end{proof}
Note that the convergence of \thref{Theorem:ad_errors} can only
hold if the the unique minimizer property in
\thref{Theorem:skc_recovery}
is fulfilled, and by \thref{Theorem:skc_recovery}
the conclusion and convergence of \thref{Theorem:ad_errors}
can only hold if the operator
$\mathcal{A}\left(\zVec\right)=\sum_{n=1}^N\aVec_n\aVec_n^Hz_n$
has the signed kernel condition of order $S$. 
By \cite[Remark~3.14]{NNLR} no other matrix
$\AMat$ can generate a linear operator
$\mathcal{A}\left(\zVec\right)=\sum_{n=1}^N\aVec_n\aVec_n^Hz_n$
with a higher order of the signed kernel condition than
the one from \thref{Theorem:sample_rate}.
Hence, the condition $S\leq\left\lceil\frac{1}{2}M^2\right\rceil-1$ in
\thref{Theorem:ad_errors} is optimal and can not be improved.
\par
The unique identifiability condition in \cite[Theorem~5]{chen}
can only guarantee recovery of $\xVec\in\mathbb{R}_+^N$ if its non-zero entries
are at specific positions.
Such recovery guarantees are called non-uniform.
The signed kernel condition, on the other hand, guarantees
recovery of all vectors $\xVec\in\Sigma_S^N\cap\mathbb{R}_+^N$
independent of where non-zero entries are according to \thref{Theorem:skc_recovery}.
Such recovery guarantees are called uniform.
	\subsection{Conclusion of Theorem \ref{Theorem:ad_errors} by a Concentration Argument}
\label{Subsection:concentration_argument}
In order to prove \thref{Theorem:ad_errors} one can now use an operator
with signed kernel condition and apply \thref{Theorem:skc} to the case of the relaxed maximum likelihood
estimator. For the non-negative least squares this will be even easier.
It remains to show that for $K$ large enough,
$\norm{\frac{1}{K}\YMat\YMat^H-\SigmaMat-\mathcal{A}\left(\xVec\right)}_{2\rightarrow 2}$
is sufficiently small with arbitrarily high probability.
This can be done by applying a concentration inequality for
sub-exponential random variables often called Bernstein type inequality.
In this subsection a precise definition of sub-exponential and sub-Gaussian
random variables is required.
\par
For a random variable $X$ define
$\norm{X}_{\psi_p}:=\inf_{t>0:\frac{\Exp{\abs{X}^p}}{t^2}\leq 2}t$.
A random variable $X$ is called sub-exponential if $\norm{X}_{\psi_1}<\infty$
and sub-Gaussian if $\norm{X}_{\psi_2}<\infty$. See \cite{high_dimensional_probability} for more information. Note that if
$X\sim\GaussianRV{0}{\sigma^2}$, then
$X$ is sub-Gaussian with $\norm{X}_{\psi_2}=2\sqrt{\frac{2}{3}}\sigma$.
\begin{Lemma}\label{Lemma:concentration}
	There exists a constant $c>0$ such that the following holds true:
	Let the $K$ columns of $\YMat\in\mathbb{C}^{M\times K}$ be mutually independent
	$\CGaussianRV{0}{\SigmaMat'}$ random variables
	for some $\SigmaMat'\in\mathbb{HPD}^M$ and $\xi>0$.
	Then
	\begin{align}
		\label{Equation:Lemma:concentration:xi_event}
			\norm{\frac{1}{K}\YMat\YMat^H-\SigmaMat'}_{2}
		\leq
			\xi
	\end{align}
	holds true with probability of at least
	\begin{align}
		\nonumber
			p'
		:=
			1-M\left(M+1\right)
			\Exp{
				-cK\min
				\left\{
					\frac{
						9\xi^2
					}{
						128M^2\sup_{m'\in\SetOf{M}}\left(\SigmaMat'_{m',m'}\right)^2
					},
					\frac{
						3\xi
					}{
						8\sqrt{2}M\sup_{m'\in\SetOf{M}}\SigmaMat'_{m',m'}
					}
				\right\}
			}.
	\end{align}
\end{Lemma}
\begin{proof}
	Let $c>0$ be the numerical constant from \cite[Theorem~2.8.1]{high_dimensional_probability}.
	Let the columns of $\YMat$ be denoted by $\yVec_k$ for $k\in\SetOf{K}$.
	Given $m_1,m_2\in\SetOf{M},k\in\SetOf{K}$ define
	\begin{align}
		\nonumber
			R_{m_1,m_2,k}
		:=&
			\Real{y_{m_1,k}\CC{y_{m_2,k}}}
			-\Real{\SigmaMat'_{m_1,m_2}}
		\\=&\nonumber
			\Real{y_{m_1,k}}\Real{y_{m_2,k}}
			+\Imag{y_{m_1,k}}\Imag{y_{m_2,k}}
			-\Real{\SigmaMat'_{m_1,m_2}}
		\TextAnd
		\\\nonumber
			I_{m_1,m_2,k}
		:=&
			\Imag{y_{m_1,k}\CC{y_{m_2,k}}}
			-\Imag{\SigmaMat'_{m_1,m_2}}
		\\=&\nonumber
			\Imag{y_{m_1,k}}\Real{y_{m_2,k}}
			-\Real{y_{m_1,k}}\Imag{y_{m_2,k}}
			-\Imag{\SigmaMat'_{m_1,m_2}}.
	\end{align}
	Since
	$\yVec_k\sim\CGaussianRV{0}{\SigmaMat'}$,
	one gets
	$\Real{y_{m,k}},\Imag{y_{m,k}}\sim\GaussianRV{0}{\SigmaMat'_{m,m}}$
	and hence these are sub-Gaussian with
	\begin{align}
		\nonumber
			\norm{\Real{y_{m,k}}}_{\psi_2}
		=
			\norm{\Imag{y_{m,k}}}_{\psi_2}
		=
			2\sqrt{\frac{2}{3}}\sqrt{\SigmaMat'_{m,m}}
		\leq
			2\sqrt{\frac{2}{3}}\sup_{m'\in\SetOf{M}}
			\sqrt{\SigmaMat'_{m',m'}}.
	\end{align}
	Due to \cite[Lemma~2.7.7]{high_dimensional_probability} the random variables $R_{m_1,m_2,k}$
	and $I_{m_1,m_2,k}$ are sub-exponential random variables with
	\begin{align}
		\label{Equation:Lemma:concentration:subexp_bound}
			\norm{R_{m_1,m_2,k}}_{\psi_1}
		\leq
			\frac{8}{3}\sup_{m'\in\SetOf{M}}
			\SigmaMat'_{m',m'}
		\geq
			\norm{I_{m_1,m_2,k}}_{\psi_1}.
	\end{align}
	Since $\yVec_k\sim\CGaussianRV{0}{\SigmaMat'}$,
	it follows that
	$\Expect{\frac{1}{K}\YMat\YMat^H}=\frac{1}{K}\sum_{k=1}^K\Expect{\yVec_k\yVec_k^H}
		=\SigmaMat'$
	and thus, $\Expect{R_{m_1,m_2,k}}=0=\Expect{I_{m_1,m_2,k}}$.
	Since the columns of $\YMat$ are independent, the random variables
	$R_{m_1,m_2,k}$ for $k\in\SetOf{K}$ are mutually independent,
	and the random variables $I_{m_1,m_2,k}$ for $k\in\SetOf{K}$
	are mutually independent.
	By the Bernstein type inequality \cite[Theorem~2.8.1]{high_dimensional_probability} one gets
	\begin{align}
		\nonumber
			\Prob{\abs{\sum_{k=1}^KR_{m_1,m_2,k}}\geq t}
		\leq&
			2\Exp{-c\min\left\{\frac{t^2}{\sum_{k=1}^K\norm{R_{m_1,m_2,k}}_{\psi_1}^2},\frac{t}{\max_{k\in\SetOf{K}}\norm{R_{m_1,m_2,k}}_{\psi_1}}\right\}}
	\end{align}
	for all $t\geq 0$.
	Applying \eqref{Equation:Lemma:concentration:subexp_bound} to this and choosing
	$t:=K2^{-\frac{1}{2}}M^{-1}\xi$ yields
	\begin{align}
		&\nonumber
			\Prob{
				\abs{\frac{1}{K}\sum_{k=1}^KR_{m_1,m_2,k}}
				\geq 2^{-\frac{1}{2}}M^{-1}\xi
			}
		\\\leq&\label{Equation:Lemma:concentration:R_prob}
			2\Exp{
				-cK\min
				\left\{
					\frac{
						9\xi^2
					}{
						128M^2\sup_{m'\in\SetOf{M}}\left(\SigmaMat'_{m',m'}\right)^2
					},
					\frac{
						3\xi
					}{
						8\sqrt{2}M\sup_{m'\in\SetOf{M}}\SigmaMat'_{m',m'}
					}
				\right\}
			}
	\end{align}
	for all $m_1,m_2\in\SetOf{M},K\in\mathbb{N}$.
	Similarly, one can get
	\begin{align}
		&\nonumber
			\Prob{
				\abs{\frac{1}{K}\sum_{k=1}^KI_{m_1,m_2,k}}
				\geq 2^{-\frac{1}{2}}M^{-1}\xi
			}
		\\\leq&\label{Equation:Lemma:concentration:I_prob}
			2\Exp{
				-cK\min
				\left\{
					\frac{
						9\xi^2
					}{
						128M^2\sup_{m'\in\SetOf{M}}\left(\SigmaMat'_{m',m'}\right)^2
					},
					\frac{
						3\xi
					}{
						8\sqrt{2}M\sup_{m'\in\SetOf{M}}\SigmaMat'_{m',m'}
					}
				\right\}
			}
	\end{align}
	for all $m_1,m_2\in\SetOf{M},K\in\mathbb{N}$.
	Due to \eqref{Equation:Lemma:concentration:R_prob}, \eqref{Equation:Lemma:concentration:I_prob}
	and the symmetries $R_{m_1,m_2,k}=R_{m_2,m_1,k}$ and $I_{m_1,m_2,k}=-I_{m_2,m_1,k}$	
	the event
	\begin{align}
		\label{Equation:Lemma:concentration:RI_event}
			\abs{\frac{1}{K}\sum_{k=1}^KR_{m_1,m_2,k}}
		<
			2^{-\frac{1}{2}}M^{-1}\xi
		\TextAnd
			\abs{\frac{1}{K}\sum_{k=1}^KI_{m_1,m_2,k}}
		<
			2^{-\frac{1}{2}}M^{-1}\xi
		\TextForAll
			m_1,m_2\in\SetOf{M}
	\end{align}
	holds true with probability of at least
	\begin{align}
		\nonumber
			p'
		:=
			1-M\left(M+1\right)
			\Exp{
				-cK\min
				\left\{
					\frac{
						9\xi^2
					}{
						128M^2\sup_{m'\in\SetOf{M}}\left(\SigmaMat'_{m',m'}\right)^2
					},
					\frac{
						3\xi
					}{
						8\sqrt{2}M\sup_{m'\in\SetOf{M}}\SigmaMat'_{m',m'}
					}
				\right\}
			}.
	\end{align}
	\par
	It remains to show that if \eqref{Equation:Lemma:concentration:RI_event} is fulfilled, then
	\eqref{Equation:Lemma:concentration:xi_event} is also fulfilled.
	Thus, assume \eqref{Equation:Lemma:concentration:RI_event} is fulfilled.
	Then,
	\begin{align}
		\nonumber
			\abs{\left(\frac{1}{K}\YMat\YMat^H-\mathcal{A}\left(\xVec\right)-\SigmaMat\right)_{m_1,m_2}}^2
		=&
			\abs{\frac{1}{K}\sum_{k=1}^KR_{m_1,m_2,k}+i\frac{1}{K}\sum_{k=1}^KI_{m_1,m_2,k}}^2
		\\=&\nonumber
			\left(\frac{1}{K}\sum_{k=1}^KR_{m_1,m_2,k}\right)^2
			+\left(\frac{1}{K}\sum_{k=1}^KI_{m_1,m_2,k}\right)^2
		\leq
			M^{-2}\xi^2
	\end{align}
	for all $m_1,m_2\in\SetOf{M}$.
	It follows that
		\nonumber
	\begin{align}
			\norm{\frac{1}{K}\YMat\YMat^H-\mathcal{A}\left(\xVec\right)-\SigmaMat}_{2}
		=&
			\sqrt{\sum_{m_1=1}^M\sum_{m_2=1}^M\abs{\left(\frac{1}{K}\YMat\YMat^H-\mathcal{A}\left(\xVec\right)-\SigmaMat\right)_{m_1,m_2}}^2}
		\leq
			\xi
	\end{align}
	which finishes the proof.
\end{proof}
Combining this with \thref{Theorem:skc_properties} yields the part about the non-negative least
squares estimator in \thref{Theorem:ad_errors}.
\begin{proof}[Proof of \thref{Theorem:ad_errors} with \eqref{Equation:Theorem:ad_errors:nnlr_optimization}]
	Choose the norm $\norm{\cdot}:=\norm{\cdot}_{2\rightarrow 2}$ for the robustness constant.
	Let $\SigmaMat\in\mathbb{HPD}^M$, $\xVec\in\mathbb{R}^N_+$, $\epsilon>0$ and $p\in\left(0,1\right)$,
	and choose
	\begin{align}
		\nonumber
			K
		\geq&
			K_0
		:=
			-\frac{1}{c}\Ln{\frac{1-p}{M\left(M+1\right)}}
		\\&\label{Equation:Theorem:nnlr_errors:def_K}\hspace{30pt}
			\cdot\max\left\{
				\frac{
					512M^2\sup_{m'\in\SetOf{M}}\left(\mathcal{A}\left(\xVec\right)+\SigmaMat\right)_{m',m'}^2
				}{
					9\tau\left(\mathcal{A}\right)^2\epsilon^2
				},
				\frac{
					16\sqrt{2}M\sup_{m'\in\SetOf{M}}\left(\mathcal{A}\left(\xVec\right)+\SigmaMat\right)_{m',m'}
				}{
					3\tau\left(\mathcal{A}\right)\epsilon
				}
			\right\}
	\end{align}
	where $c>0$ is the numerical constant from \thref{Lemma:concentration}.
	Applying \thref{Lemma:concentration} with $\SigmaMat':=\mathcal{A}\left(\xVec\right)+\SigmaMat$
	and $\xi:=\frac{\tau\left(\mathcal{A}\right)}{2}\epsilon$ yields that
	\begin{align}
		\label{Equation:Theorem:nnlr_errors:delta_event}
			\norm{\frac{1}{K}\YMat\YMat^H-\mathcal{A}\left(\xVec\right)-\SigmaMat}_{2}
		\leq
			\frac{\tau\left(\mathcal{A}\right)}{2}\epsilon
	\end{align}
	is fulfilled with probability at least $p'\geq p$.
	If \eqref{Equation:Theorem:nnlr_errors:delta_event} is fulfilled,
	one can apply \thref{Theorem:skc_properties} with $\WMat':=\frac{1}{K}\YMat\YMat^H-\SigmaMat$,
	which yields
	\begin{align}
		\nonumber
			\norm{\xVec-\zVec}_2
		\leq
			\frac{2}{\tau\left(\mathcal{A}\right)}\norm{\frac{1}{K}\YMat\YMat^H-\SigmaMat-\mathcal{A}\left(\xVec\right)}_2
		\leq
			\epsilon
	\end{align}
	for any minimizer $\zVec$ of \eqref{Equation:Theorem:ad_errors:nnlr_optimization}.
\end{proof}
On the other hand, \thref{Lemma:concentration} can be combined with
\thref{Theorem:skc} to give the part about the relaxed maximum likelihood
estimator in \thref{Theorem:ad_errors}.
\begin{proof}[Proof of \thref{Theorem:ad_errors} with \eqref{Equation:Theorem:ad_errors:ml_optimization}]
	Let $\SigmaMat\in\mathbb{HPD}^M$, $\xVec\in\mathbb{R}^N_+$, $\epsilon>0$ and $p\in\left(0,1\right)$,
	and choose any $0<\beta<\lambda_1\left(\mathcal{A}\left(\xVec\right)+\SigmaMat\right)$ and $\eta>0$.
	Note that $\mathcal{H}:=\left\{\mathcal{A}\left(\zVec\right)+\SigmaMat:\zVec\in\mathbb{R}_+^N\right\}
		\subset\mathbb{HPD}^M$
	so that all conditions of \thref{Theorem:skc} are fulfilled.
	Let $\delta$ be from \thref{Theorem:skc},
	and choose
	\begin{align}
		\nonumber
			K
		\geq
			K_0
		:=&
			\max\Bigg\{
				M,
				-\frac{1}{c}\Ln{\frac{1-p}{M\left(M+1\right)}}
		\\&\label{Equation:Theorem:ml_errors:def_K}\hspace{20pt}
				\cdot\max\left\{
					\frac{
						128M^2\sup_{m'\in\SetOf{M}}\left(\mathcal{A}\left(\xVec\right)+\SigmaMat\right)_{m',m'}^2
					}{
						9\delta\left(\epsilon\right)^2
					},
					\frac{
						8\sqrt{2}M\sup_{m'\in\SetOf{M}}\left(\mathcal{A}\left(\xVec\right)+\SigmaMat\right)_{m',m'}
					}{
						3\delta\left(\epsilon\right)
					}
				\right\}
			\Bigg\}
	\end{align}
	where $c>0$ is the numerical constant from \thref{Lemma:concentration}.
	Applying \thref{Lemma:concentration} with $\SigmaMat':=\mathcal{A}\left(\xVec\right)+\SigmaMat$
	and $\xi:=\delta\left(\epsilon\right)$ yields that
	\begin{align}
		\label{Equation:Theorem:ml_errors:delta_event}
			\norm{\frac{1}{K}\YMat\YMat^H-\mathcal{A}\left(\xVec\right)-\SigmaMat}_{2\rightarrow 2}
		\leq&
			\norm{\frac{1}{K}\YMat\YMat^H-\mathcal{A}\left(\xVec\right)-\SigmaMat}_{2}
		\leq
			\delta\left(\epsilon\right)
	\end{align}
	is fulfilled with probability at least $p'\geq p$.
	Let the columns of $\YMat$ be denoted by $\yVec_k$ for $k\in\SetOf{K}$.
	Since the columns of $\YMat$ are complex normal distributed and $K\geq M$,
	$\YMat$ has full rank and thus $\frac{1}{K}\YMat\YMat^H\in\mathbb{HPD}^M$ with probability
	of at least $1$.
	If additionally \eqref{Equation:Theorem:ml_errors:delta_event}
	is fulfilled,
	one can apply \thref{Theorem:skc} with $\WMat:=\frac{1}{K}\YMat\YMat^H$,
	which yields $\norm{\xVec-\zVec}_2\leq\epsilon$ for
	any minimizer $\zVec$ of \eqref{Equation:Theorem:ad_errors:ml_optimization}.
\end{proof}
	\section{Proof of Theorem \ref{Theorem:coord_desc_convergence}: Coordinate Descent for Relaxed Maximum Likelihood Estimation}
\label{Subsection:proof_coordinate_descent}
\noindent
The proof of \thref{Theorem:coord_desc_convergence} is based on
\cite[Theorem~4.1(c)]{convergence}. For this, the compactness
of level sets is required. This is shown first for the
non-negative least residual.
\begin{Lemma}
	\label{Lemma:nnlr_level_set_compact}
	Let $\mathcal{A}:\mathbb{C}^N\rightarrow\mathbb{C}^{M\times M}$ have
	signed kernel condition of order $S$.
	For any $\WMat'\in\mathbb{C}^{M\times M}$ and $\gamma>0$ the level set
	\begin{align}
		\nonumber
			\mathcal{G}
		:=
			\left\{\zVec\in\mathbb{R}^N_+:\norm{\mathcal{A}\left(\zVec\right)-\WMat'}_{2\rightarrow 2}\leq \gamma\right\}
	\end{align}
	is compact.
\end{Lemma}
\begin{proof}
	It is clear that $\mathcal{G}$ is closed, hence it remains to show that it is bounded.
	This follows solely from the fact that the operator
	$\mathcal{A}$ has the signed kernel condition of order $S$.
	Let $\xVec\in\Sigma_S^N\cap\mathbb{R}_+^N\supset\ZeroSet$ be arbitrary.
	Let $\zVec\in\mathcal{G}$, then $\zVec$ is a minimizer of
	$\min_{\zVec'\in\mathbb{R}_+^N}\norm{\mathcal{A}\left(\zVec'\right)-\mathcal{A}\left(\zVec\right)}_{2\rightarrow 2}$. By
	\thref{Theorem:skc_properties} one gets
	\begin{align}
		\nonumber
			\norm{\zVec-\xVec}_2
		\leq&
			\frac{2}{\tau\left(\mathcal{A}\right)}\norm{\mathcal{A}\left(\zVec\right)-\mathcal{A}\left(\xVec\right)}_{2\rightarrow 2}
		\leq
			\frac{2}{\tau\left(\mathcal{A}\right)}\norm{\mathcal{A}\left(\zVec\right)-\WMat'}_{2\rightarrow 2}
			+\frac{2}{\tau\left(\mathcal{A}\right)}\norm{\WMat'-\mathcal{A}\left(\xVec\right)}_{2\rightarrow 2}
		\\\leq&\nonumber
			\frac{2}{\tau\left(\mathcal{A}\right)}\gamma
			+\frac{2}{\tau\left(\mathcal{A}\right)}\norm{\WMat'-\mathcal{A}\left(\xVec\right)}_{2\rightarrow 2}.
	\end{align}
	Thus, $\mathcal{G}$ is bounded and hence compact.
\end{proof}
The compactness of level sets of the relaxed maximum likelihood estimator
follows from this.
\begin{Lemma}
	\label{Lemma:ml_level_set_compact}
	Let $\mathcal{A}:\mathbb{C}^N\rightarrow\mathbb{C}^{M\times M}$ defined by
	$\mathcal{A}\left(\zVec\right)=\sum_{n=1}^N\aVec_n\aVec_n^Hz_n$ have
	signed kernel condition of order $S$.
	Then, $\aVec_n\neq 0$ for all $n\in\SetOf{N}$. Further, for any $\SigmaMat\in\mathbb{HPD}^M$,
	$\WMat\in\mathbb{HPD}^M$ and $\gamma>0$ the level set
	\begin{align}
		\nonumber
			\mathcal{G}
		:=
			\left\{\zVec\in\mathbb{R}^N_+:
				\Trace{\left(\sum_{n=1}^N\aVec_{n}\aVec_{n}^Hz_{n}+\SigmaMat\right)^{-1}\WMat}
			+\Ln{\Det{\sum_{n=1}^N\aVec_{n}\aVec_{n}^Hz_{n}+\SigmaMat}}
			\leq \gamma\right\}
	\end{align}
	is compact.
\end{Lemma}
\begin{proof}
	Since $\mathcal{A}\left(\zVec\right)=\sum_{n=1}^N\aVec_{n}\aVec_n^Hz_n$
	has the signed kernel condition of order $S$
	every real kernel vector of $\mathcal{A}$ needs to have at least $S+1>0$ negative entries.
	Hence, the standard unit vector $\eVec$ with $e_n=1$ and $e_{n'}=0$ for all $n'\neq n$
	is not a kernel vector. It follows that $0\neq\mathcal{A}\left(\eVec\right)=\aVec_n\aVec_n^H$
	and thus $\aVec_n\neq 0$.
	\par
	Due to continuity the level set $\mathcal{G}$ is closed. It remains to show that it is bounded.
	Let $\left(\g{},g_1,g_2,\nu,\epsilon_0\right)$ be the sufficiently convex tuple of
	\thref{Lemma:trace_logdet_tuple}.
	By \thref{Lemma:tuple_convex} there exists $\delta_1,\delta_2,g_1,g_2$ such the tuple
	$\left(\g{},g_1,g_2,\delta_1,\delta_2\right)$ is sufficiently nice.
	Further, let
	$\mathcal{H}:=\left\{\mathcal{A}\left(\zVec\right)+\SigmaMat:\zVec\in\mathbb{R}_+^N\right\}
		\subset\mathbb{HPD}^M$.
	This allows one to apply \thref{Lemma:level_set_compact1}, which yields that the level set
	$\mathcal{G}':=\left\{\ZMat\in\mathcal{H}:\sum_{m=1}^M\g{\lambda_m\left(\WMat^\frac{1}{2}\ZMat^{-1}\WMat^\frac{1}{2}\right)}\leq \gamma'\right\}$
	is compact for every $\gamma'\in\mathbb{R}$.
	By \eqref{Equation:Lemma:trace_logdet_tuple:trace_logdet=sumg} it follows that
	$\mathcal{G}'':=\left\{\ZMat\in\mathcal{H}:\Trace{\ZMat^{-1}\WMat}+\Ln{\Det{\ZMat}}\leq \gamma\right\}$
	is compact, and thus there exists an $\alpha>0$ such that
	$\norm{\ZMat}_{2\rightarrow 2}\leq \alpha$ for all $\ZMat\in\mathcal{G}''$.
	Since $\mathcal{A}\left(\zVec\right)+\SigmaMat\in\mathcal{G''}$ for all
	$\zVec\in\mathcal{G}$, one gets
	$\norm{\mathcal{A}\left(\zVec\right)+\SigmaMat}_{2\rightarrow 2}\leq \alpha$
	for all $\zVec\in \mathcal{G}$.
	Thus, $\mathcal{G}\subset\left\{\zVec\in\mathbb{R}^N_+:\norm{\mathcal{A}\left(\zVec\right)+\SigmaMat}_{2\rightarrow 2}\leq \alpha\right\}$.
	The latter is a level set and due to \thref{Lemma:nnlr_level_set_compact} bounded.
	Hence, $\mathcal{G}$ is bounded and thus compact.
\end{proof}
It remains to prove the statement about stationary points of the coordinate
descent method.
\begin{proof}[Proof of \thref{Theorem:coord_desc_convergence}]
	The proof follows from \cite[Theorem~4.1(c)]{convergence}.
	Let
	\begin{align}
		\nonumber
			f_0\left(\zVec\right)
		:=&
			\Trace{\left(\sum_{n''=1}^N\aVec_{n''}\aVec_{n''}^Hz_{n''}+\SigmaMat\right)^{-1}\frac{1}{K}\YMat\YMat^H}
			+\Ln{\Det{\sum_{n''=1}^N\aVec_{n''}\aVec_{n''}^Hz_{n}+\SigmaMat}}
	\end{align}
	and
	\begin{align}
		\nonumber
			f_{n''}\left(z\right)
		:=
			\begin{Bmatrix}
					0 & \TextIf & z\geq 0
				\\
					\infty & \TextIf & z<0
			\end{Bmatrix}
	\end{align}
	for all $n''\in\SetOf{N}$
	as well as
	$f\left(\zVec\right):=f_0\left(\zVec\right)+\sum_{n''=1}^Nf_{n''}\left(z\right)$
	so that \eqref{Equation:Theorem:ad_errors:ml_optimization} can be written as
	\begin{align}
		\nonumber
			\min_{\zVec\in\mathbb{R}^N}f\left(\zVec\right).
	\end{align}
	At first, note that at the end of every for iteration
	$\SigmaMat'=\left(\SigmaMat+\sum_{n''=1}^N\aVec_{n''}\aVec_{n''}^Hx_{n''}\right)^{-1}$
	and
	$\left(\xVec'_{i,n'}\right)_{\sigma\left(n'\right)}$ is a unique minimizer of coordinate update
	\begin{align}
		\nonumber
			\min_{z\in\mathbb{R}}
		&
			f\left(x_{1},\dots,x_{\sigma\left(n'\right)-1},z,x_{\sigma\left(n'\right)+1},\dots,x_N\right)
	\end{align}
	which was established in \cite[Equation~(19)-(23)]{fengler}.
	Further, note that at the end of every iteration
	$\SigmaMat\in\mathbb{HPD}^M$, $\xVec\in\mathbb{R}_+^N$
	and thus $\SigmaMat'\in\mathbb{HPD}^M$.
	Since $\aVec_n\neq 0$ by \thref{Lemma:ml_level_set_compact},
	one gets $\aVec_n^H\SigmaMat'\aVec_n\neq 0$
	so that the next minimizer is indeed well defined and unique.
	Thus, $\left(\xVec_{i}\right)_{i\in\mathbb{N}}$
	is the sequence generated by the block coordinate descent method in
	\cite[page~478]{convergence}.
	\par
	Since $\YMat$ has full rank, one gets
	$\WMat:=\frac{1}{K}\YMat\YMat^H\in\mathbb{HPD}^M$.
	By $\xVec_0\in\mathbb{R}_+^N$ and
	applying \thref{Lemma:ml_level_set_compact} the level set
	$\mathcal{G}
		=\left\{\zVec\in\mathbb{R}^N_+:f_0\left(\zVec\right)\leq f_0\left(\xVec_0\right)\right\}
		=\left\{\zVec\in\mathbb{R}^N:f\left(\zVec\right)\leq f\left(\xVec_0\right)\right\}$
	is compact. Further, $f_0$ is smooth on $\mathbb{R}_+^N$.
	The indices $n=\sigma\left(n'\right)$
	are chosen to satisfy the essentially cyclic rule as defined in \cite[page~478]{convergence}.
	Since $f_0$ is smooth in an open neighborhood $\textnormal{dom}f_0$ around $\mathbb{R}_+^N$,
	\cite[Lemma~3.1]{convergence} yields that $f$ is regular in all $\zVec\in\mathbb{R}_+^N$.
	By \cite[Theorem~4.1(c)]{convergence} every cluster point of
	$\left(\xVec_{\left(i-1\right)N+N-1}\right)_{i\in\mathbb{N}}
		=\left(\xVec'_{i,N-1}\right)_{i\in\mathbb{N}}$
	is a stationary point and a coordinate-wise global minimum and of $f$.
	\par
	For any $n'''\in\SetOf{N}$ \thref{Algorithm:coord_desc} with the input
	$\SigmaMat+\sum_{n''=1}^N\aVec_{n''}\aVec_{n''}^H\left(\xVec_{1+n'''}\right)_{n''}$
	instead of $\SigmaMat$, with $\xVec_{1+n'''}$ instead of $\xVec_0$ and with
	$\tau\left(\cdot\right):=\sigma\left(\cdot+1+n'''\right)$ instead of $\sigma$
	generates the sequence $\left(\xVec_{i+1+n'''}\right)_{i\in\mathbb{N}}$.
	From the same argument as before it follows that for any $n'''\in\SetOf{N}$
	any cluster point of
	$\left(\xVec_{\left(i-1\right)N+N-1+1+n'''}\right)_{i\in\mathbb{N}}
		=\left(\xVec_{iN+n'''}\right)_{i\in\mathbb{N}}
		=\left(\xVec'_{i+1,n'''}\right)_{i\in\mathbb{N}}$
	is a stationary point and a coordinate-wise global minimum and of $f$.
	Since any cluster point $\left(\xVec_i\right)_{i\in\mathbb{N}}$ needs to be a cluster point
	of $\left(\xVec_{iN+n'''}\right)_{i\in\mathbb{N}}$ for some $n'''\in\SetOf{N}$,
	the claim follows.
\end{proof}
	\section{Number Of Receive Antennas}
\label{Section:number_of_receive_antennas}
In this section the scaling of the number of receive antennas $K$ in
\thref{Theorem:ad_errors} is discussed. The number of receive antennas is dependent on which of the two
estimators \eqref{Equation:Theorem:ad_errors:nnlr_optimization} and
\eqref{Equation:Theorem:ad_errors:ml_optimization} is chosen
and needs to satisfy \eqref{Equation:Theorem:nnlr_errors:def_K}
or \eqref{Equation:Theorem:ml_errors:def_K} respectively.
The first part of the maximum in \eqref{Equation:Theorem:ml_errors:def_K}
can simply be evaluated so that the second part of this maximum
is of interest.
The term $-\frac{8}{3}\Ln{\frac{1-p}{M\left(M+1\right)}}>0$ appears in both \eqref{Equation:Theorem:nnlr_errors:def_K} and \eqref{Equation:Theorem:ml_errors:def_K},
comes from the union bound, and can also be evaluated.
The term
$\sup_{m\in\SetOf{M}}
	\left(\mathcal{A}\left(\xVec\right)+\SigmaMat\right)_{m,m}$ requires knowledge of the vector
of large scale fading coefficients that is unknown prior to choosing the number of receive antennas.
However, if additional box constraints for the large scale fading coefficients are known,
this value can be bounded and evaluated easily as well.
Such constraints are usually present and known in applications,
since devices can only send at a maximum power
or are considered inactive and absorbed in the noise if
their transmit power is too small.
\par
Ignoring $\tau\left(\mathcal{A}\right)$, $\epsilon$, $\delta\left(\epsilon\right)$ as
well as the logarithmic terms the number of receive antennas $K$ needs to scale at least
in the order of $M^2$. However, the more insightful information is within terms
$\tau\left(\mathcal{A}\right)$, $\epsilon$ and $\delta\left(\epsilon\right)$.
According to \eqref{Equation:Theorem:nnlr_errors:def_K}
in the case of the non-negative least squares
the number of receive antennas
depends directly on the product $\tau\left(\mathcal{A}\right)\epsilon$.
According to \eqref{Equation:Theorem:ml_errors:def_K}
in the case of the relaxed maximum likelihood estimator
the number of receive antennas depends on
$\delta\left(\epsilon\right)$ from \thref{Theorem:skc} which can be chosen as $\delta_{skc}$.
Due to \eqref{Equation:Theorem:skc:delta_def} $\delta_{skc}$ scales like the product
$M^{-1}\tau\left(\mathcal{A}\right)\epsilon$ for $\epsilon$ small enough and $M$ large enough.
Due to this, the relaxed maximum likelihood estimator requires significantly more receive
antennas to achieve the same estimation errors in \thref{Theorem:ad_errors}.
It remains open, whether this gap can be overcome by improving \thref{Theorem:skc}.
\par
Consider $K_0=K_0\left(\epsilon\right)$ from
\eqref{Equation:Theorem:nnlr_errors:def_K}
or \eqref{Equation:Theorem:ml_errors:def_K} respectively
as a function of $\epsilon$.
If all factors but $\epsilon$ are constant and $\epsilon$ is small enough,
$\delta_{skc}$ is linear.
Thus, in both cases $K_0\left(\epsilon\right)$ is also invertible for $\epsilon$
small enough, and it scales like $K_0\left(\epsilon\right)=C\epsilon^{-2}$.
Hence, if all other factors are constant and $\epsilon=\norm{\xVec-\zVec}_2$ is small enough,
\begin{align}
	\label{Equation:scaling_K_est_error}
		K
	\geq
		K_0\left(\epsilon\right)
	=
		C\norm{\xVec-\zVec}_2^{-2}
\end{align}
is sufficient to achieve the estimation error $\norm{\xVec-\zVec}_2$.
\par
However, both estimators depend on the robustness constant $\tau\left(\mathcal{A}\right)$.
It should be noted that the robustness constants are different in each case because
they depend on different norms, see the proof of \thref{Theorem:ad_errors} with \eqref{Equation:Theorem:ad_errors:nnlr_optimization} and the proof of \thref{Theorem:skc} respectively.
However, this is not that important since independent of the norm
the robustness constant of any normalized operator $\mathcal{A}$ satisfies
\begin{align}
		\tau\left(\mathcal{A}\right)
	\leq	
		2\sqrt{\frac{2}{3}}\left(\Exp{\frac{S}{4M^2}\Ln{\frac{N}{4S}}}-1\right)^{-1}
\end{align}
by \cite[Theorem~5.2]{NNLR}.
Thus, in both cases $\tau\left(\mathcal{A}\right)$ can only stay constant
if the number of pilot symbols satisfies $M^2\asymp S\Ln{\frac{N}{4S}}$
or is higher.
If $M^2\asymp S\Ln{\frac{N}{4S}}$ is violated,
the robustness constant decreases exponentially in $\frac{S}{4M^2}\Ln{\frac{N}{4S}}$
and, ignoring other factors, the number of receive antennas also grows exponentially in
$\frac{S}{4M^2}\Ln{\frac{N}{4S}}$.
So \thref{Theorem:ad_errors} allows one to reduce the required number of pilot symbols
from $M^2\asymp S\left(\Ln{\ExpE\frac{N}{S}}\right)^2$ to $M^2\asymp S$,
but it pays the price for a significantly increased number of receive antennas due to
$\tau\left(\mathcal{A}\right)$ decreasing.
\par
Consider the choice $\epsilon:=\frac{1}{4}\min_{n:x_n\neq 0}\abs{x_n}$ in
\thref{Remark:thresholding}
and the ratio $\xi=\frac{\epsilon}{\sup_{m'\in\SetOf{M}}\left(\mathcal{A}\left(\xVec\right)+\SigmaMat\right)_{m',m'}}$,
and assume all other factors stay constant.
For large $\epsilon$ and small $\SigmaMat$ the ratio $\xi$ stays constant,
and due to \eqref{Equation:Theorem:nnlr_errors:def_K} $K$ can be chosen
constant for the non-negative least squares.
However, for small $\epsilon$ and large $\SigmaMat$ the ratio $\xi$ is small,
and according to \eqref{Equation:Theorem:nnlr_errors:def_K} $K$ needs to scale
like $\frac{1}{\xi^2}$ for the non-negative least squares.
On the other hand, for small $\epsilon$ and large $\SigmaMat$
the ratio $\xi^2$ scales roughly like
$\frac{\epsilon^2}{\sup_{m'\in\SetOf{M}}\left(\SigmaMat\right)_{m',m'}^2}$
which is basically the signal-to-noise ratio of the weakest user.
Thus, the number of active users $K$ for the non-negative least squares
has to scale like one over the signal-to-noise ratio of the weakest user
if the signal-to-noise ratio of the weakest user is small.
\par
In the case of the relaxed maximum likelihood estimator one needs to consider
the denominator $\delta\left(\epsilon\right)$ instead of $\epsilon$
due to \eqref{Equation:Theorem:ml_errors:def_K}.
$\delta\left(\epsilon\right)$ can be chosen as $\delta_{skc}\left(\epsilon\right)$ from
\eqref{Equation:Theorem:skc:delta_def} which is scaling linearly in $\epsilon$ for small
$\epsilon$. Hence, the number of active users $K$ for the relaxed maximum likelihood
estimator also has to scale like one over the signal-to-noise ratio of the weakest user
if the signal-to-noise ratio of the weakest user is small.
In any case, any change in the signal-to-noise ratio can be compensated by potentially
increasing the number of receive antennas.
	\section{Simulations}
\label{Section:Simulations}
\begin{figure}[!t]
\centering
\subfloat[Robustness constant $\tau'$ and estimation error as a function of
$S$ for adversarial large scale fading coefficients in the case of infinitely many receive antennas with $N=17,M=4$.]{\includegraphics[width=2.5in]{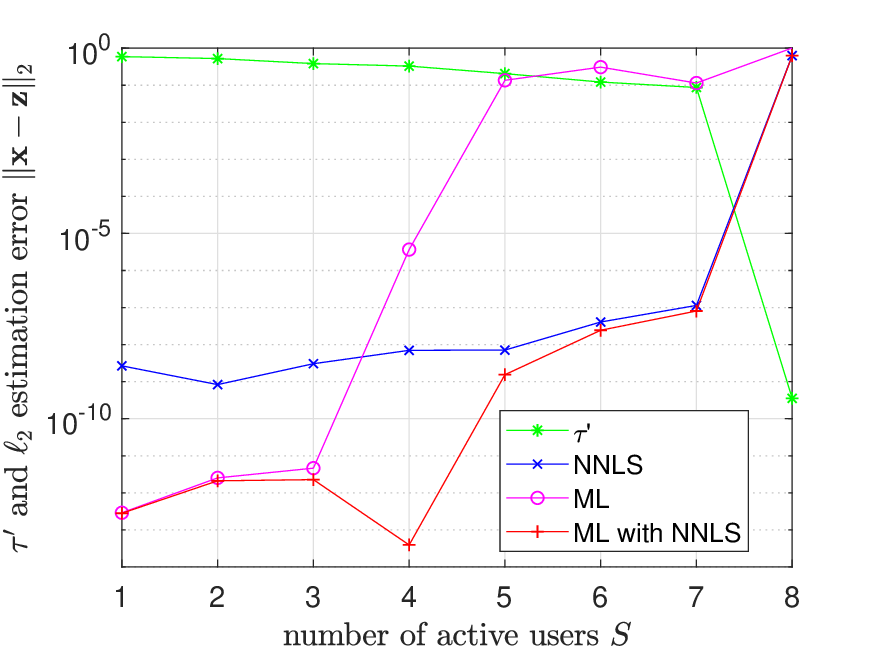}%
\label{Figure:Numerics:adversarial}}
\hfil
\subfloat[Estimation error as a function of $S$ for random large scale fading coefficients in the case of infinitely many receive antennas with $N=17,M=4$.]{\includegraphics[width=2.5in]{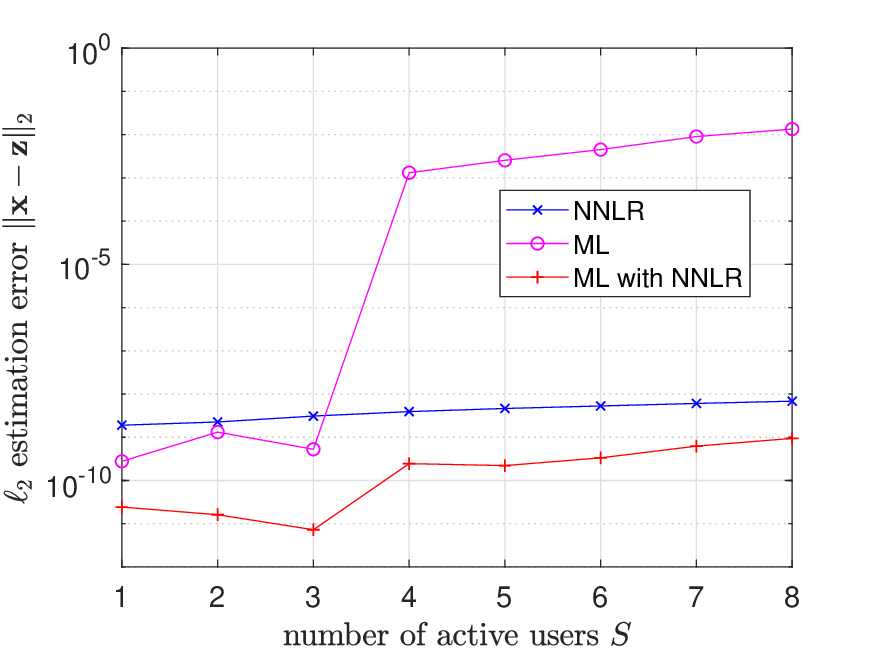}%
\label{Figure:Numerics:infinite}}
\hfil
\subfloat[Estimation error as a function of the magnitude of the perturbation $\tau$ when $\WMat$ is artificially perturbed with magnitude $\rho$ for random large scale fading coefficients with $N=17,M=4,S=7$.]{\includegraphics[width=2.5in]{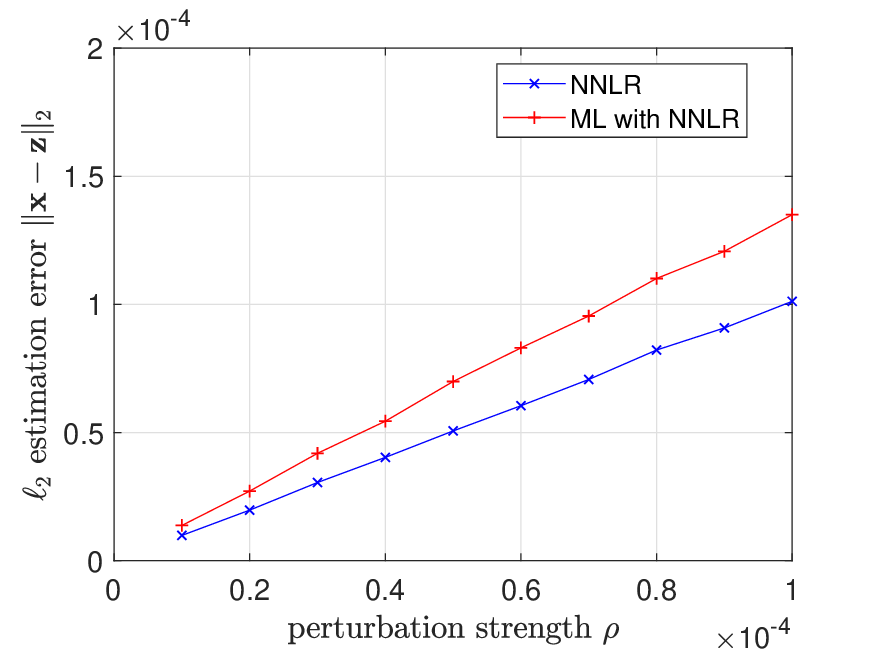}%
\label{Figure:Numerics:hermitian}}
\hfil
\subfloat[$\norm{\xVec-\zVec}_2^{-2}$ as a function of $K$ in the case of finitely many receive antennas for random large scale fading coefficients with $N=17,M=4,S=7$.]{\includegraphics[width=2.5in]{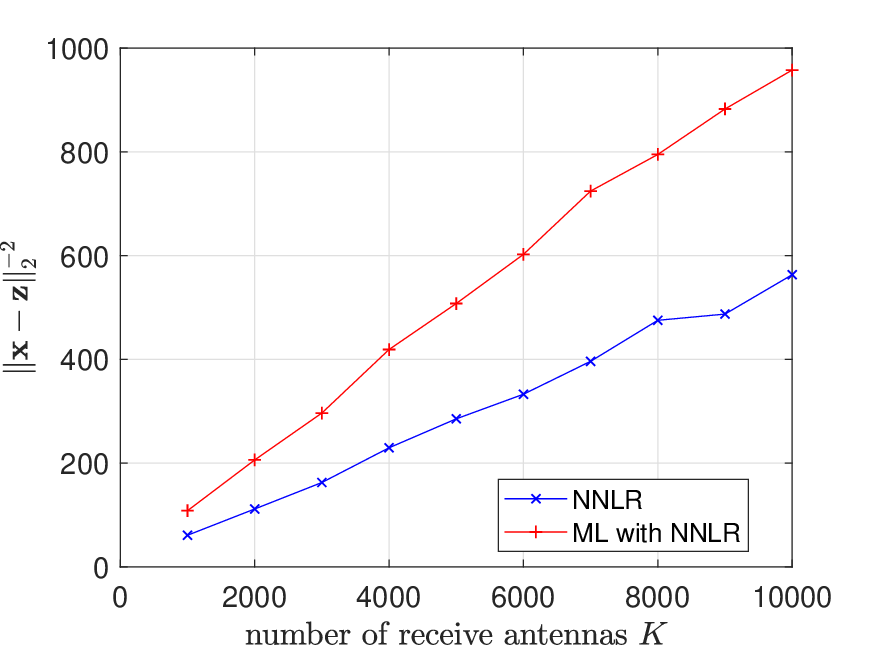}%
\label{Figure:Numerics:finite}}
\caption{Numerical verification of several results.}
\label{Figure:Numerics}
\end{figure}
In this section some numerical results are presented. 
These numerical results are not supposed to show superior performance,
but they are supposed to verify the theoretical predictions of this work.
The codebook used in the proof of \thref{Theorem:sample_rate}, i.e. from \eqref{Equation:Theorem:sample_rate:defA},
has a very bad robustness constant, thus it performs poorly in simulations.
Due to this, one codebook $\AMat$ with independent $\CGaussianRV{0}{1}$
entries is drawn and the method \cite[Theorem~3.8]{NNLR}
is used to verify that the linear operator
$\mathcal{A}\left(\zVec\right)=\sum_{n=1}^N\aVec_n\aVec_n^Hz_n$
has signed kernel condition of order $S_0$ but
does not have signed kernel condition of order $S_0+1$.
See \cite[Section~6]{NNLR} for an explanation on how to adapt this
to the complex case.
Due to the combinatorial nature of \cite[Theorem~3.8]{NNLR} and
since these results are supposed to verify theoretical predictions,
all dimensions except the number of receive antennas $K$ are kept small.
To be precise, for all simulations the codebook named above is fixed with
$M=4,N=17$ and $S_0=7$, and $\Sigma$ is chosen as $10^{-4}$ times the identity.
All convex optimization problems, including
\eqref{Equation:Theorem:ad_errors:nnlr_optimization}
and the ones appearing in \cite[Theorem~3.8]{NNLR}, will be solved with
the CVX package for Matlab \cite{CVX1,CVX2}.
The resulting minimizer of \eqref{Equation:Theorem:ad_errors:nnlr_optimization}
is denoted by NNLS in plots and in this section.
The problem \eqref{Equation:Theorem:ad_errors:ml_optimization}
is solved by \refP{Algorithm:coord_desc} with a permutation chosen uniformly at random and terminates after $100$ while iterations.
If \refP{Algorithm:coord_desc} is initialized with the zero vector, it will
be denoted by ML in plots and in this section.
If it is initialized by NNLS,
it is denoted by ML with NNLS in plots and in this section.
Given the codebook $\AMat$ let $\BMat\in\mathbb{C}^{M^2\times N}$
be the matrix whose $n$-th column is a reordering of $\aVec_n\aVec_n^H$
and further let $\BMat_{real}$ be the real part
and $\BMat_{imag}$ be the imaginary part of $\BMat$.
The method from \cite[Theorem~3.8]{NNLR} will be used with the
matrix $\begin{bmatrix}\BMat_1\\\BMat_2\end{bmatrix}\in\mathbb{R}^{2M^2\times N}$,
the norm $\norm{\cdot}=\norm{\cdot}_2$ on $\mathbb{R}^{2M^2}$
and some $S$ as input to calculate
\begin{align}
	\label{Equation:tau_prime}
		\tau'
	:=
		\inf_{
			\zVec'\in\mathbb{R}_+^N,\xVec'\in\Sigma_S^N\cap\mathbb{R}_+^N,\zVec\neq\xVec
		}
		\frac{
			\norm{
				\begin{bmatrix}\BMat_1\\\BMat_2\end{bmatrix}
				\left(\zVec'-\xVec'\right)
			}_2
		}{
			\norm{\zVec'-\xVec'}_1
		}
\end{align}
as described in \cite[Section~6]{NNLR}.
$\tau'$ is one robustness constant of the linear operator
$\mathcal{A}\left(\zVec\right)=\sum_{n=1}^N\aVec_n\aVec_n^Hz_n$;
however, it is not one of the robustness constants used in this work,
see \cite[Definition~2.1,Theorem~3.8]{NNLR}.
\par
According to \thref{Theorem:skc_recovery} and the corresponding
result for the non-negative least squares from \cite{NNLR}
the signed kernel condition is a necessary and sufficient condition
for recovery of all $\xVec\in\Sigma_S^N\cap\mathbb{R}_+^N$ to succeed
with the non-negative least squares and the relaxed maximum likelihood estimator
when no perturbations in the covariance matrix are present.
Further, according to
\cite[Theorem~3.2, Theorem~2.2, Proposition~2.8]{NNLR}
$\tau'>0$ is equivalent to
the linear operator
$\mathcal{A}\left(\zVec\right)=\sum_{n=1}^N\aVec_n\aVec_n^Hz_n$
having the signed kernel condition of order $S$.
To verify this, the constant $\tau'$ is calculated.
Further, the minimizers $\zVec',\xVec'$ of \eqref{Equation:tau_prime}
are used to create the adversarial vector of large-scale fading coefficients
$\xVec=\frac{\xVec'}{\norm{\xVec'}_2}$.
Then, the problems
\eqref{Equation:Theorem:ad_errors:nnlr_optimization}
and
\eqref{Equation:Theorem:ad_errors:ml_optimization}
are solved with
$\frac{1}{K}\YMat\YMat^H$ replaced by $\WMat:=\sum_{n=1}^N\aVec_n\aVec_n^Hx_n+\SigmaMat$.
This simulates an infinite number of receive antennas for an adversarial vector
of large-scale fading coefficients.
The results are plotted in \refP{Figure:Numerics:adversarial}.
The robustness constant $\tau'$ is non-zero for $S\leq 7$ but
nearly zero for $S=8$. Thus, the linear operator
$\mathcal{A}\left(\zVec\right)=\sum_{n=1}^N\aVec_n\aVec_n^Hz_n$
has signed kernel condition of order $S\leq S_0=7=\left\lceil\frac{1}{2}M^2\right\rceil-1$
which is exactly what is possible according to \thref{Theorem:sample_rate}.
Since $\tau'$ is non-zero for $S=8$, the linear operator
$\mathcal{A}\left(\zVec\right)=\sum_{n=1}^N\aVec_n\aVec_n^Hz_n$
likely does not have signed kernel condition of order
$S=8>\left\lceil\frac{1}{2}M^2\right\rceil-1$
which is exactly as predicted by \cite[Remark~3.14]{NNLR}.
NNLS and ML with NNLS recover the vector of large scale fading coefficients
sufficiently well for $S\leq 7$ and but fail for $S>7$.
Thus, they succeed exactly whenever the linear operator
$\mathcal{A}\left(\zVec\right)=\sum_{n=1}^N\aVec_n\aVec_n^Hz_n$
has the signed kernel condition.
The recovery with ML fails whenever $S>4$.
Since ML with NNLS succeeds for $4<S\leq 7$,
the reason for this must be that \refP{Algorithm:coord_desc}
does not find a global minimizer, and that the results only hold
for global minimizers.
It could be that \refP{Algorithm:coord_desc} does not use enough iterations,
gets stuck in a stationary point that is not a global minimizer,
or just uses a bad initialization.
In total, the simulation supports the prediction of \cite[Remark~3.14]{NNLR}, namely that
the signed kernel condition of order $S>\left\lceil\frac{1}{2}M^2\right\rceil-1$
can not be fulfilled.
Further, the simulation supports the prediction of \thref{Theorem:skc_recovery},
namely that the signed kernel condition is a necessary condition
for recovery of all $\xVec\in\Sigma_S^N\cap\mathbb{R}_+^N$ to succeed
with the non-negative least squares and the relaxed maximum likelihood estimator
when no perturbations in the covariance matrix are present.
\par
In all further simulations $\xVec$ is drawn uniformly at random from
$\Sigma_S^N\cap\mathbb{R}_+^N\cap\left\{\zVec:\norm{\zVec}_2=1\right\}$
instead of being the adversarial construction,
and for every simulation the average of
$\norm{\xVec-\zVec}_2$ over $1000$ samples is calculated and plotted.
In order to investigate the sufficiency of the signed kernel condition
for recovery, the problem
\eqref{Equation:Theorem:ad_errors:nnlr_optimization}
and
\eqref{Equation:Theorem:ad_errors:ml_optimization}
are solved with
$\frac{1}{K}\YMat\YMat^H$ replaced by the $\WMat:=\sum_{n=1}^N\aVec_n\aVec_n^Hx_n+\SigmaMat$ but this time $\xVec$ is the randomly chosen as described.
This again simulates an infinite number of receive antennas.
The results are plotted in \refP{Figure:Numerics:infinite}.
NNLS and ML with NNLS can recover $\xVec$ sufficiently for all $S\leq 7$.
Thus, the simulation supports the prediction of \thref{Theorem:skc_recovery},
namely that the signed kernel condition is a sufficient condition
for recovery of all $\xVec\in\Sigma_S^N\cap\mathbb{R}_+^N$ to succeed
with the non-negative least squares and the relaxed maximum likelihood estimator
when no perturbations in the covariance matrix are present.
NNLS and ML with NNLS exceed the theoretically guaranteed performance
as they seem to guarantee recovery even for $S>7$;
however, this does not disprove the theory.
It could be that it is just unlikely to draw a vector $\xVec$ that is similar
to the adversarial vectors created in \refP{Figure:Numerics:adversarial}
so that this just never happens in the $1000$ samples.
It should be noted that ML requires $S\leq 4$ for sufficient recovery.
The bad performance of ML compared to ML with NNLS is again due to
\refP{Algorithm:coord_desc} not reaching a global optimizer of
\eqref{Equation:Theorem:ad_errors:ml_optimization}.
Since all predictions are only about global minimizers of
\eqref{Equation:Theorem:ad_errors:ml_optimization}, ML will be omitted in
further simulations.
For other plots $S:=S_0=7$ is fixed.
\par
In total, \refP{Figure:Numerics:adversarial} and \refP{Figure:Numerics:infinite}
support the prediction that the signed kernel condition
is a sufficient and necessary condition
for recovery of all $\xVec\in\Sigma_S^N\cap\mathbb{R}_+^N$
with the non-negative least squares and the relaxed maximum likelihood estimator
when no perturbations in the covariance matrix are present.
Moreover, both simulations support the prediction that
$S\leq\left\lceil\frac{1}{2}M^2\right\rceil-1$
is the exact condition when recovery of all $\xVec\in\Sigma_S^N\cap\mathbb{R}_+^N$
can be possible.
\par
According to \thref{Theorem:skc} $\delta$ scales at worst linearly in $\epsilon$ for small enough
$\epsilon$. By an argument similar to the one used to get
\eqref{Equation:tuple_convex:homogeinity},
the estimation error $\norm{\xVec-\zVec}_2$ should scale linearly in the magnitude of the perturbation
$\norm{\WMat-\sum_{n=1}^N\aVec_n\aVec_n^Hz_n-\SigmaMat}_{2\rightarrow 2}$
as long as the magnitude of the perturbation is small.
In order to investigate this, the columns of the real and imaginary
part of $\NMat\in\mathbb{C}^{M\times M}$
are drawn mutually independently according to $\GaussianRV{0}{\IDMat}$
and $\NMat':=\frac{\NMat+\NMat^H}{\norm{\NMat+\NMat^H}_{2\rightarrow 2}}$
is normalized to create a Hermitian indefinite perturbation.
Then, 
\eqref{Equation:Theorem:ad_errors:nnlr_optimization}
and
\eqref{Equation:Theorem:ad_errors:ml_optimization}
are solved
with $\frac{1}{K}\YMat\YMat^H$ replaced by the $\WMat:=\sum_{n=1}^N\aVec_n\aVec_n^Hz_n+\SigmaMat+\rho\NMat'$
to generate a perturbation with magnitude $\rho>0$. The results are plotted in \refP{Figure:Numerics:hermitian}.
As predicted, the scaling is linear for both estimators.
\par
According to \eqref{Equation:scaling_K_est_error} and the discussion before it,
the value $\norm{\xVec-\zVec}_2^{-2}$ should scale at worst like $K$
as long as other factors remain constant and $\norm{\xVec-\zVec}_2$ is already small enough.
In order to investigate this, the problems
\eqref{Equation:Theorem:ad_errors:nnlr_optimization}
and
\eqref{Equation:Theorem:ad_errors:ml_optimization} are solved where $\YMat$ is as
specified in \thref{Theorem:ad_errors}.
The results are plotted in \refP{Figure:Numerics:finite}.
The values scale as predicted for both estimators.
	\section*{Acknowledgements}
The work of G. Caire and H. B. Zarucha was supported by the Gottfried Wilhelm Leibniz-Preis 2021 of the German Science Foundation (DFG).

\end{document}